\documentclass[12pt,letterpaper]{article}
\usepackage[T1]{fontenc}
\usepackage{amsmath,amsthm,mathtools}
\usepackage{newtxtext,newtxmath}
\usepackage[letterpaper,margin=1in]{geometry}
\usepackage{graphicx,booktabs,array,adjustbox,float,caption,pdflscape}
\usepackage{xcolor,url,xurl,microtype,setspace,needspace,etoolbox}
\usepackage{scicite}
\usepackage[unicode,colorlinks=true,allcolors=blue,bookmarksnumbered=true]{hyperref}
\definecolor{refblue}{HTML}{005EB8}\hypersetup{allcolors=refblue}
\DeclareCaptionLabelFormat{continued}{#1~#2 (continued)}
\newcommand{\scititle}{Calibrating generative flows through regional probability matching}
\hypersetup{pdftitle={\scititle},pdfauthor={Yubo Huang, Xiongyu Hu, Xin Lai, Zixi Wang, Liao Jian, Jingzehua Xu, Yong Ren, Chuan He}}
\newcommand{\E}{\mathbb E}\newcommand{\R}{\mathbb R}\newcommand{\N}{\mathcal N}
\newcommand{\Law}{\mathcal L}\newcommand{\TV}{\operatorname{TV}}\newcommand{\Var}{\operatorname{Var}}
\providecommand{\burl}[1]{\url{#1}}
\providecommand{\doi}[1]{\href{https://doi.org/#1}{doi:\nolinkurl{#1}}}
\newcommand{\suppref}[1]{\textcolor{refblue}{Supplementary #1}}
\let\ScienceCite\cite\renewcommand{\cite}[1]{\textcolor{refblue}{\ScienceCite{#1}}}
\newtheoremstyle{suppstatement}{8pt}{6pt}{\normalfont}{0pt}{\bfseries}{.}{.5em}{\thmname{#1}\thmnumber{ #2}\thmnote{. #3}}
\theoremstyle{suppstatement}
\newtheorem{suppproposition}{Proposition}
\newtheorem{supptheorem}{Theorem}

\newtheorem{suppcorollary}{Corollary}
\newenvironment{suppproof}{\begin{proof}[Proof.]\normalsize}{\end{proof}}

\allowdisplaybreaks[1]
\begin{document}
\ifdefined\SAonlySupplement
  \input{common/citation_numbers}
  \hypersetup{pdftitle={Supplementary Materials for \scititle}}
\else
\nocite{sohl2015deep,ho2020denoising,song2019generative,song2021score,chen2018neural,grathwohl2019ffjord,papamakarios2021normalizing,duan2025reactot,li2024turbulence,hess2025downscaling,lipman2023flow,tong2024improving,albergo2023building,albergo2023stochastic,benamou2000computational,peyre2019computational,pooladian2023multisample,liu2023flow,kornilov2024optimal,chiang2026cotfm,chapel2025differentiable,lin2026modelaligned,grover2019bias,choi2020fair,sajjadi2018assessing,kynkaanniemi2019improved,naeem2020reliable,benton2024errorbounds,zhou2025erroranalysis,song2021ddim,lu2022dpmsolver,karras2022elucidating,song2023consistency,geng2025meanflow,bonneel2015sliced,heusel2017gans,he2016deep,deng2009imagenet,lloyd1982least,hairer1993solving,krizhevsky2009learning,ronneberger2015unet,wu2018group,loshchilov2019decoupled,gretton2012kernel,binkowski2018demystifying,parmar2022aliased,lipton2018detecting,vandermaaten2008visualizing,holm1979simple,efron1994introduction,tipping1999mixtures,clanuwat2018classical,reyes2013har,huang2026gmpcfdata}
\begin{center}
{\large\bfseries \scititle\par}
\vspace{8pt}
{\small Yubo~Huang$^{1}$, Xiongyu~Hu$^{1*}$, Xin~Lai$^{1}$, Zixi~Wang$^{2}$, Liao~Jian$^{3}$, Jingzehua~Xu$^{4*}$, Yong~Ren$^{5}$, Chuan~He$^{1}$\par}
\end{center}
\begingroup\small\setlength{\parindent}{0pt}\setlength{\parskip}{3pt}
$^{1}$Key Laboratory of Transportation Tunnel Engineering, Ministry of Education, Southwest Jiaotong University, Chengdu 611756, Sichuan, China.\par
$^{2}$School of Information and Software Engineering, University of Electronic Science and Technology of China, Chengdu 611731, Sichuan, China.\par
$^{3}$School of Civil Engineering and Geomatics, Southwest Petroleum University, Chengdu 610500, Sichuan, China.\par
$^{4}$Department of Mechanical Engineering, The University of Hong Kong, Pokfulam, Hong Kong, China.\par
$^{5}$State Key
Laboratory of Space Network and Communication, Department of Electronic Engineering, Tsinghua University, Beijing 100084, China.\par
$^{*}$Corresponding authors. Email: \href{mailto:huxiongyu@swjtu.edu.cn}{huxiongyu@swjtu.edu.cn} (X.H.); \href{mailto:u3014074@connect.hku.hk}{u3014074@connect.hku.hk} (J.X.).\par
\endgroup

\medskip\noindent\textbf{Short title:} Regional probability matching in generative flows

\section*{Abstract}
Generative models for scientific simulation must reproduce the frequencies of rare states as well as their structure. We developed generalized mass-preserving coupled flow matching (G-MPCF), which assigns regions of Gaussian noise the probability masses of corresponding regions in weighted observations. Conditional sampling within matched regions preserves the prescribed target distribution while changing the velocities used to train the flow, without modifying the underlying generative model architecture. On a rare-mode benchmark, an analytic coupling reduced relative score error by 95.0\% and numerical integration error by 89.9\% compared with independent pairing. On CIFAR-10 with a 100:1 class-probability ratio, G-MPCF reduced corrected class-mass error by 44.6\% relative to minibatch optimal transport and improved mean image fidelity at matched computational budgets. These results show how the construction of training pairs can improve the reproduction of prescribed population frequencies by generative flows.
\clearpage
\section{Introduction}
Generative models are becoming powerful tools for exploring complex scientific systems, including populations in which rare states carry important mechanistic information. Their usefulness depends not only on generating plausible observations but also on reproducing how frequently these states occur, as distorted probabilities can bias population statistics and the interpretation of rare phenomena. Diffusion and score-based models learn to reverse progressive data corruption~\cite{sohl2015deep,ho2020denoising,song2019generative,song2021score}, whereas continuous normalizing flows learn differential equations that transform tractable source distributions into complex target distributions~\cite{chen2018neural,grathwohl2019ffjord,papamakarios2021normalizing}. These approaches have enabled studies of molecular transition states, turbulent trajectories and high-resolution Earth-system fields~\cite{duan2025reactot,li2024turbulence,hess2025downscaling}. Yet realistic generated structures do not necessarily imply accurate population frequencies, making the simultaneous recovery of structural fidelity and prescribed occurrence frequencies a central challenge for scientific generative modelling.

Flow matching and stochastic interpolants learn generative velocity fields along probability paths between tractable sources and target distributions~\cite{lipman2023flow,tong2024improving,albergo2023building,albergo2023stochastic}. The source--target coupling determines which displacements supervise these fields and can alter the interpolation geometry without changing the prescribed endpoint marginals. Dynamic optimal transport provides a kinetic-energy formulation of transport~\cite{benamou2000computational,peyre2019computational}. Minibatch transport and multisample matching construct finite-sample assignments~\cite{tong2024improving,pooladian2023multisample}, while rectified flow uses learned dynamics to generate revised training pairs~\cite{liu2023flow}. Optimal flow matching develops direct transport mappings~\cite{kornilov2024optimal}; cluster-wise transport exploits local data structure~\cite{chiang2026cotfm}, and learned plans and model-aligned costs further adapt pairing to the regression problem~\cite{chapel2025differentiable,lin2026modelaligned}. Together, these approaches establish coupling design as an important factor in generative learning, while leaving open how its joint probability structure affects endpoint frequency accuracy under limited model capacity and computational budgets.

For imbalanced generation, the desired distribution also specifies how often each part of the data should occur. Sample weights allow this prescribed population to differ from the empirical data frequencies. Their interaction with source--target assignment is particularly relevant when low-probability regions receive few draws within a training batch. Previous studies have investigated bias correction and distribution calibration in generative models~\cite{grover2019bias,choi2020fair}, while precision, recall, density and coverage characterize sample fidelity and support~\cite{sajjadi2018assessing,kynkaanniemi2019improved,naeem2020reliable}. A region can contain convincing and diverse generated samples while receiving an inaccurate fraction of the total probability. In flow matching, this frequency error can arise even when training pairs have the prescribed endpoint marginals. Couplings with identical source and target laws induce different conditional velocity fields along the interpolation path. Their approximation by a neural network with limited capacity affects the distribution reached during generation~\cite{benton2024errorbounds,zhou2025erroranalysis}.

Regional probability offers a direct way to organize this regression problem. Matching the probability of a source region to that of a target region assigns a corresponding share of source draws to the target's local structure. Conditional sampling within the target region then retains the relative weights of its observations. The remaining freedom in pairing can organize the supervising displacements within each region. Regions fitted in a continuous feature representation can capture structure across semantic classes. Class-frequency error and within-class coverage then provide complementary measurements of the resulting feature-based coupling. This construction connects population frequencies to the geometry of training pairs and motivates examining their effects on field approximation, sample allocation and within-region coverage together.

The learned field also determines the numerical work required for generation. Implicit diffusion sampling and higher-order solvers reduce the cost of following learned dynamics~\cite{song2021ddim,lu2022dpmsolver,karras2022elucidating}; consistency models and mean flows learn mappings for few-step inference~\cite{song2023consistency,geng2025meanflow}. Coupling design acts earlier, through the interpolation states and velocities presented during training. Its endpoint effects therefore involve both approximation of the training field and evolution along the learned trajectory. Joint measurements of field, rollout and integration errors can connect changes in the training pairs to generation under a limited evaluation budget.

Here we introduce generalized mass-preserving coupled flow matching (G-MPCF), which constructs training pairs by aligning Gaussian source regions with weighted data regions according to their probabilities.
As illustrated in Fig.~\ref{fig:method}, target observations are sampled conditionally within matched regions, while additional within-region rank matching refines pair geometry without altering target sampling frequencies.
Rare-mode experiments show reductions of 95.0\% in relative score error and 89.9\% in solver error.
On CIFAR-10 with a 100:1 class-probability ratio, matched-computation comparisons yield 49.3\% and 44.6\% lower corrected class-mass errors than independent pairing and minibatch optimal transport, respectively, alongside improved mean fidelity.
Component interventions identify regional probability matching as the principal source of these gains, with mass-only coupling achieving comparable calibration and fidelity.
Furthermore, rollout diagnostics connect field improvements to endpoint performance across targets and computation budgets.

\section{Results}
\subsection{Mass-calibrated pairing simplifies rare-mode fields}

G-MPCF divides a Gaussian source into regions with the same probabilities as weighted data regions.
It samples targets with replacement within each matched region to establish a calibrated regional correspondence, and the full implementation further pairs source and target scores by rank.
The image implementation uses twenty regions fitted to pretrained features and learns the velocity in image coordinates.
Let $Q$ and $P$ denote source and target probability laws, and let $\Gamma(Q,P)$ be the set of joint laws with these marginals. A coupling $\Pi\in\Gamma(Q,P)$ pairs source noise $Z$ with a target observation $X$. The velocity network $v_\theta(x,t)$ has trainable parameters $\theta$ and takes state $x$ and time $t$ as inputs.
For $(Z,X)\sim\Pi$ and $t$ uniform on $[0,1]$, the flow-matching objective is
\begin{equation}
\begin{aligned}
X_t&=(1-t)Z+tX,\qquad U=X-Z,\\
\mathcal L_\Pi(\theta)
&=\E_{t,(Z,X)\sim\Pi}\|v_\theta(X_t,t)-U\|^2.
\end{aligned}
\label{eq:objective}
\end{equation}
Here $X_t$ is the interpolated state, $U$ is the pair displacement, $\E$ denotes expectation and $\|\cdot\|$ is the Euclidean norm. The population minimizer is the conditional velocity $u_\Pi(x,t)=\E[U\mid X_t=x]$, which supplies the vector field associated with the chosen interpolation law~\cite{lipman2023flow,albergo2023building}.

For $N$ observations $x_i$ with positive weights $w_i$ summing to one, $P=\sum_{i=1}^N w_i\delta_{x_i}$, where $\delta_{x_i}$ is the point mass at $x_i$. We partition their indices into $m$ label-free regions $C_j$, with $j=1,\ldots,m$, and assign probabilities $p_j=\sum_{i\in C_j}w_i$. A unit projection of a standard Gaussian is divided at quantiles into regions with successive probabilities $p_j$. Writing $J$ for the source region and $I$ for the sampled target index gives, for $i\in C_j$,
\begin{equation}
\begin{aligned}
\Pr(J=j)&=p_j,\qquad \Pr(I=i\mid J=j)=w_i/p_j,\\
\Pr(I=i)&=p_j(w_i/p_j)=w_i.
\end{aligned}
\label{eq:mass-preservation}
\end{equation}
Within each region, source and target scores are paired by rank using a stable sorting procedure.
The prescribed target marginal in equation~\ref{eq:mass-preservation} is already established by calibrated regional sampling, while the subsequent permutation preserves sampled target-index multiplicities and rearranges source--target correspondences.
For fixed one-dimensional score multisets, equal-rank matching minimizes the sum of squared score differences over within-region permutations, providing a geometric criterion without altering target probabilities.

Regional probability assignment and within-region geometric ordering can both change the regression problem, despite preserving the prescribed source and target marginals.
At fixed $t$, for any square-integrable velocity field $f$,
\begin{equation}
\begin{aligned}
\E\|U-f(X_t,t)\|^2
&=\E\!\left[\operatorname{tr}\Var(U\mid X_t)\right]\\
&\quad+\E\|u_\Pi(X_t,t)-f(X_t,t)\|^2.
\end{aligned}
\label{eq:risk}
\end{equation}
Here $\Var(U\mid X_t)$ is the conditional covariance matrix and $\operatorname{tr}$ denotes its trace. The two terms are conditional velocity variance and conditional-mean approximation error.
Both depend on the coupling, so preserving $Q$ and $P$ permits different regression problems.
To examine how this change affects approximation at fixed model capacity, we use analytically tractable targets whose population velocity and score can be evaluated directly.

We first evaluated a rare-mode target with an explicit coupling and reference field. Rare12 contains eight common modes of probability 0.115 and four outer modes of probability 0.02, with centred and whitened coordinates. We partition a two-dimensional standard Gaussian into alternating sectors of probability 0.115 and 0.135, then split the larger sectors at radial threshold $r_\star=\sqrt{-2\log(0.02/0.135)}=1.95425$ to assign outer probability 0.02 to each rare mode. We call this analytic construction MPCF and compare it with independent conditional flow matching (I-CFM). Independent interpolation paths cross the central region towards different targets. MPCF assigns rare targets to outer Gaussian sectors, separating common- and rare-mode contributions throughout interpolation, as shown in Fig.~\ref{fig:controlled}, A to D. The exact velocity and scaled score $s_t(x)=-x+t u_\Pi(x,t)$ provide references for evaluating the learned field.

MPCF reduced relative score mean squared error (MSE) by 95.0\%, from 0.16476 to 0.00817, with improvement in all 32 paired training runs in Fig.~\ref{fig:controlled}E. Relative velocity MSE fell by 66.7\%, and Euler-to-reference MSE at 16 evaluations fell by 89.9\%. Endpoint sliced-Wasserstein distance (SWD)~\cite{bonneel2015sliced} decreased by 35.2\%, from 0.2061 to 0.1335, alongside 43.6\% lower mode-mass error. The paired SWD difference was $-0.07253$, with a 95\% bootstrap interval from $-0.08080$ to $-0.06472$. All six exact two-sided paired tests in \suppref{Table~S1} favoured MPCF, with $P\leq5.13\times10^{-9}$. With one field evaluation, MPCF reached SWD 0.1606, below the independent flow's 0.1982 at 32 evaluations. Its error stabilized near 0.13 from two evaluations onwards, while independent pairing improved gradually as the budget increased. Figure~\ref{fig:images}A follows the generated densities and endpoint errors across inference budgets.

The equal-mass gaussian8 target tests the same construction with uniform mode probabilities. Score error fell by 79.3\% and velocity error by 46.1\% in Fig.~\ref{fig:controlled}F. MPCF also improved the one- and two-evaluation endpoints. At 16 evaluations, independent pairing reached SWD 0.2543 and MPCF 0.2722, accompanied by Euler-to-reference MSE of 0.00489 and 0.01019. The ordering of endpoint and integration errors therefore changes with the inference budget. Figure~\ref{fig:images}B follows the generated densities and endpoint errors across the same budgets.

A six-method comparison with five training runs per method and target extends the analysis to minibatch optimal-transport conditional flow matching (OT-CFM), two-stage rectified flow (2-RF), Optimal Flow Matching (OFM) and MeanFlow~\cite{tong2024improving,liu2023flow,kornilov2024optimal,geng2025meanflow}. At a budget of 16 evaluations, MPCF reached mean rare12 SWD 0.1392, compared with 0.1456 for OT-CFM, 0.1522 for OFM and 0.2026 for I-CFM. Its mode-mass error was 0.0260, the lowest among the six methods in Fig.~\ref{fig:controlled}, I and J. At one evaluation, OFM gave the lowest mean SWD on both targets; MPCF gave the lowest mode-mass error, reaching 0.0248 on rare12 and 0.0171 on gaussian8 in Fig.~\ref{fig:controlled}, G and H.

\subsection{Regional probability matching improves image generation}

We compared six coupling methods on CIFAR-10 using a shared 20.1-million-parameter U-Net and $5.425\times10^{17}$ floating-point operations (FLOPs) per training run, including pairing and auxiliary computation. Importance-weighted sampling of 45,000 training images produced a 100:1 class-probability ratio. Eight paired runs supplied 50,000 images each at 16 Euler evaluations. A validation confusion matrix corrected classifier-based class frequencies; minority recall measured feature-manifold coverage within the three least probable classes.

The balanced feature maps in Fig.~\ref{fig:images}C resolve class allocation across all six methods. The matched truck-region enlargements contain 146 target observations, 92 under I-CFM and 143 under G-MPCF. Across the complete evaluations, the predicted truck frequency increases from 6.92\% under I-CFM to 9.13\% under G-MPCF, towards the 10\% target. The absolute frequency error decreases by 71.8\%, and G-MPCF is closest to the truck target among the six methods in each of the three runs. Figure~\ref{fig:images}D shows the angular distributions of the same feature maps about the reference centre. Long-tail enlargements and complete class-frequency comparisons appear in \suppref{Figs.~S3 and S4}.

G-MPCF reduced corrected MassTV to 0.04863, compared with 0.09593 for I-CFM and 0.08777 for OT-CFM. These values correspond to reductions of 49.3\% and 44.6\%, respectively. G-MPCF had the lowest mean allocation error among all six methods in Fig.~\ref{fig:ablation}A; two-stage rectified flow followed at 0.05858. The paired reductions were 0.04730 relative to I-CFM and 0.03915 relative to OT-CFM. The respective 95\% intervals ranged from $0.04534$ to $0.04908$ and from $0.03672$ to $0.04136$. Both superiority comparisons had Holm-adjusted $P=0.016$.

Weighted FID~\cite{heusel2017gans} was 19.80 for G-MPCF, 22.35 for I-CFM and 22.06 for OT-CFM. Both prespecified FID non-inferiority tests passed with adjusted $P=0.008$. Model-aligned coupling (MAC) reached the lowest mean FID of 19.60. Minority recall was 0.674 for G-MPCF, 0.667 for I-CFM and 0.697 for OT-CFM, with paired difference intervals spanning zero. The improvement in aggregate class allocation was accompanied by better mean fidelity than the two primary controls, while within-minority coverage varied with the coupling.

We repeated the comparison with equal class probabilities to assess the effect of removing the imposed class imbalance. On balanced CIFAR-10, G-MPCF yielded FID 17.71 and MassTV 0.08317, compared with 18.43 and 0.09686 for I-CFM over three independent training runs. OT-CFM reached 17.11 and 0.07629, and MAC reached FID 16.77. G-MPCF improved both aggregate metrics over independent pairing in this setting, with the larger allocation advantage appearing in the long-tail task. We next isolated regional probabilities and rank ordering through ablations to identify which part of the coupling produced the allocation gain.

\subsection{Regional probability assignment explains the image gain}
We evaluated regional calibration and rank ordering through five coupling operations, with three independent training runs each.
Mass-only coupling used quantile assignment and conditional target sampling; geometry-only coupling used global rank matching; and uniform-region pairing combined equal source-region probabilities, regional sorting and the weighted conditional target law.
All operations shared the model, training budget and inference setting.
Because mass-only and independent pairing preserve the same target marginal, their comparison isolates the effect of regional source--target correspondence.
Comparing G-MPCF with the mass-only variant isolates the contribution of rank matching within each region.

Mass-only coupling reproduced the allocation and fidelity gains of G-MPCF in Fig.~\ref{fig:ablation}B.
Mean MassTV was 0.05023 for mass-only coupling, 0.05044 for G-MPCF and 0.09561 for independent pairing; their weighted FIDs were 19.98, 20.16 and 22.48, respectively.
Geometry-only matching yielded MassTV 0.10020 and FID 22.53, compared with 0.10097 and 21.42 for uniform-region pairing.
Notably, mass-only coupling achieved slightly lower mean MassTV and FID than G-MPCF, indicating no additional aggregate benefit from within-region rank matching in this setting.
Thus, calibrated regional sampling accounts for the principal allocation and fidelity improvements, although differences between the two calibrated variants remain descriptive given the three independent runs.

For the second most probable class, the predicted-to-target frequency ratio increased from 0.68 under independent pairing to 0.86 under mass-only coupling and 0.85 under G-MPCF.
Despite similar frequency profiles, the two calibrated variants exhibited differences in class-specific recall, as shown in Fig.~\ref{fig:ablation}, C and D.
Minority recall was 0.694 for mass-only coupling and 0.565 for G-MPCF, with corresponding minority precision of approximately 0.894 and 0.900, respectively, as detailed in \suppref{Table~S9}.
These results suggest that rank ordering can affect local feature-manifold coverage without necessarily improving aggregate calibration or fidelity.
Therefore, regional probability matching drives the main calibration gains, whereas geometric ordering provides a distinct, task-dependent refinement.

We further examined the effect of relaxing regional probability assignments by setting
$q_{\rho}=(1-\rho)p+\rho u$, where $p=(p_1,\ldots,p_m)$ is the calibrated regional probability vector, $u_j=1/m$ defines the uniform vector and $\rho\in[0,1]$ controls relaxation. The $j$th relaxed probability is $q_{\rho,j}$.
Retaining conditional target probabilities $w_i/p_j$ yields
$P_{I,q_{\rho}}(i)=q_{\rho,j}w_i/p_j$ for $i\in C_j$, whose total variation distance $\mathrm{TV}$ from the intended index law $P_I(i)=w_i$ is
\begin{equation}
\begin{aligned}
\!\!\!\!\mathrm{TV}(P_{I,q_{\rho}},P_I)
&=\frac{1}{2}\sum_j\sum_{i\in C_j}
\left|\frac{q_{\rho,j}}{p_j}-1\right|w_i\\
&=\frac{1}{2}\sum_j|q_{\rho,j}-p_j|
=\rho\,\mathrm{TV}(u,p).
\end{aligned}
\end{equation}
Across five relaxation levels, generated MassTV increased from 0.05044 through 0.05998, 0.07099 and 0.08397 to 0.10097, while weighted FID rose from 20.16 to 21.42 and minority recall varied non-monotonically, as shown in Fig.~\ref{fig:ablation}, E and F.
Because relaxation also changes the induced target-index marginal, these trends reflect sensitivity to regional mass misspecification without independently isolating geometric effects.
Together with the mass-only intervention, these findings identify calibrated regional correspondence as the principal source of improvement, with rank matching providing additional geometric organization.

\subsection{Rollout diagnostics connect field and endpoint accuracy}

The effect of coupling on a fitted field propagates through both the learned dynamics and their numerical integration. Let $A_t$ and $B_t$ be the second and fourth moments of the field error under the interpolation law, and let $\chi_t^2$ measure the chi-squared divergence of the learned-rollout law from that law. Let $\widehat X_T^{(k)}$ be the learned flow\textquotesingle s numerical endpoint after $k$ Euler steps of size $h=T/k$, with terminal time $0<T\leq1$. Writing $\Law$ for a random variable\textquotesingle s probability law and $W_2$ for the 2-Wasserstein distance, the regularity conditions in Methods give
\begin{equation}
\begin{aligned}
&W_2(\Law(\widehat X_T^{(k)}),P)\\
&\quad\leq\int_0^T e^{L(T-t)}
       \left(A_t+\sqrt{\chi_t^2B_t}\right)^{1/2}dt\\
&\qquad+\frac{hM_a}{2}\frac{e^{\widehat L T}-1}{\widehat L}\\
&\qquad+(1-T)\left(\E_\Pi\|X-Z\|^2\right)^{1/2}.
\end{aligned}
\label{eq:endpoint}
\end{equation}
Here $L$ and $\widehat L$ are the spatial Lipschitz constants of the exact and learned fields, and $M_a$ bounds the learned flow's material acceleration. The three terms quantify field error under rollout, Euler integration error and terminal truncation. The first term evaluates field error along the learned trajectory, and its density-ratio factor links those visited states to the interpolation distribution used in training. The bound therefore motivates measuring regression error on both distributions together with the numerical endpoint discrepancy.

We examined these stages in 36 mixture conditions spanning dimensions 8, 32 and 128, rare probabilities from 0.01 to 0.2, and changes in mode separation and source displacement. Three paired training runs at each condition used common source and target laws with shared randomness. The single-factor slices in Fig.~\ref{fig:boundary}, A to C isolate each input change at dimension 32. Increasing rare-mode probability strengthened the field improvement and eventually the endpoint improvement, while changing separation or displacement produced distinct endpoint responses. The diagnostics in Fig.~\ref{fig:boundary}, D to F measure displacement-regression improvement on interpolation states, the change in regression loss during rollout and Euler-to-Heun endpoint MSE. Their distributions compare the effects of coupling across the three dimensions.

A regression using all three diagnostics explained 88.8\% of the variation in the 36 condition-mean endpoint differences, with fitted values and residuals in Fig.~\ref{fig:boundary}, G and I. Across 108 training pairs, 106 improved field loss, 58 improved endpoint SWD and 56 improved both in Fig.~\ref{fig:boundary}H. Together, these diagnostics link field approximation to endpoint accuracy through the learned dynamics and numerical integration. The image-field comparison in \suppref{Fig.~S8} adds path length and path-length ratio to the regression, rollout and integration measurements.

\subsection{Image-generation gains extend across computation budgets}

Eight image configurations varied training-data fraction, U-Net width or computation around the full-data, width-96 reference, with three paired training runs per configuration. At 16 evaluations, G-MPCF improved FID in 23 of 24 pairs and corrected MassTV in 21. Figure~\ref{fig:budget}, A and B shows fidelity across five inference budgets and paired MassTV reductions at 16 evaluations. The gains extended across both data and model scales. With half of the training data, FID fell from 32.01 to 27.14 and MassTV from 0.17529 to 0.09534. At base width 64, FID fell from 27.62 to 24.25; at width 128, it fell from 20.79 to 19.18. At half the reference computation, mean FID improved from 29.09 to 27.39 and MassTV from 0.10314 to 0.06550. With one-quarter computation, both FIDs were near 224 and MassTV was 0.71993 for I-CFM and 0.73777 for G-MPCF.

Reducing the inference budget of the reference model also preserved a fidelity advantage over independent pairing. At four evaluations, G-MPCF reached FID 49.75 compared with 69.51 for I-CFM, a 28.4\% reduction. At eight evaluations, the values were 27.57 and 33.87; at 16, they were 21.74 and 24.22. This comparison changes the cost of sampling from already trained fields. Together with the training-resource interventions, it shows how the observed image gains vary when computation is changed before or after fitting the model.

We selected a source and sampler from 122 validation candidates spanning two sources, three couplings, four displacement fractions and five inference budgets, together with source-retention controls. The rule screened class-mass error against the best source-retention control, compared validation KID and selected the smallest admissible NFE. At full displacement and 16 evaluations, G-MPCF gave the lowest validation KID for the Gaussian source, and OT-CFM performed best within the fitted mixture source in Fig.~\ref{fig:budget}C. Averaging over three training runs selected Gaussian-source G-MPCF with full displacement and 16 evaluations in Fig.~\ref{fig:budget}D. Its test FID was 21.78 and MassTV 0.05489, compared with 24.36 and 0.09765 for I-CFM and 23.87 and 0.09197 for OT-CFM under the same source and budget in Fig.~\ref{fig:budget}, E and F. \suppref{Fig.~S7} shows image sequences from shared Gaussian inputs across seven budgets from 1 to 64 Euler evaluations under these three couplings. The validation procedure therefore retained the full Gaussian-source flow for this image task, and the independent test comparison reproduced its allocation and fidelity advantages over the two controls. As the evaluated coupling candidates included independent pairing, minibatch optimal transport and G-MPCF, this selection identifies the preferred action within the considered search space, while the incremental contribution of rank matching relative to mass-only coupling is examined separately in the component interventions.

\section{Discussion}

We developed G-MPCF, a coupling approach that incorporates regional probability calibration into flow matching. The method matches Gaussian source regions to weighted data regions according to their probability masses, preserving the prescribed target distribution while reorganizing the velocity supervision used for training. On the rare-mode benchmark, the analytic coupling reduces relative score error by 95.0\% and numerical integration error by 89.9\% at 16 Euler field evaluations, compared with independent flow matching. On CIFAR-10 with a 100:1 class-probability ratio, G-MPCF reduces corrected class-mass error by 44.6\% relative to minibatch optimal transport and improves mean image fidelity under matched computational budgets. G-MPCF uses the standard flow-matching objective, allowing regional calibration to be incorporated into existing training pipelines without changing the generative architecture.

Extending G-MPCF to scientific data will require representations that distinguish the states whose populations are relevant to the application. In the image experiments, pretrained features define the regions and observation weights determine their probabilities. For molecular and materials systems, conformational coordinates or structural order parameters could provide a physical basis for defining these regions. Relative state populations obtained from reference simulations or experiments could then inform the target probabilities across a range of relevant physical conditions. Learning representations that resolve transitions between these states, and adapting the partition as additional observations become available, would be useful for applying the method to more complex systems. Independent estimates of generated state populations would allow regional calibration to be evaluated together with structural fidelity and coverage.

With regional probabilities specified, the remaining freedom in the coupling can be used to organize the geometry of training pairs. The image ablations support this separation, with calibrated sampling accounting for the principal allocation and fidelity gains and local pairing affecting coverage. Matching rules that incorporate physical symmetries or appropriate transport costs could further adapt the construction to scientific data. Conditional extensions could represent changes in state populations with temperature, pressure or external forcing, providing a route to generating ensembles across physical conditions. We anticipate that regional probability matching will be useful for scientific sampling tasks in which structural accuracy and relative populations jointly determine the quantities of interest, including molecular ensembles, competing material states and rare configurations in complex dynamical systems.

\section{Materials and Methods}
\subsection{Study design}
We tested whether regional probability matching improves the allocation and fidelity of generated samples at matched training and inference budgets. Analytically tractable mixtures provided reference fields and endpoint probabilities. The primary image comparison used eight paired training runs per method on long-tailed CIFAR-10; component experiments varied regional sampling and rank matching separately. Resource and validation experiments varied the training budget, source distribution and number of field evaluations.
\subsection{Mass-calibrated coupling}

Training images resized to 224 pixels are represented by 512-dimensional penultimate features of an ImageNet-pretrained ResNet-18~\cite{he2016deep,deng2009imagenet}. Features are averaged over each image and its horizontal reflection, then standardized with target-weighted training means and standard deviations. Weighted deterministic $k$-means~\cite{lloyd1982least} produces $m=20$ non-empty regions, ordered along a sign-anchored principal axis of their centroids. The representation and regions remain fixed during flow training.

For a unit source direction $a_0$ and standard-normal cumulative distribution function $\Phi$, its inverse defines quantile boundaries $b_j$ and the region index $J(z)$ of source point $z$,
\begin{equation}
\begin{aligned}
b_j&=\Phi^{-1}\!\left(\sum_{\ell=1}^j p_\ell\right),\\
J(z)=j&\quad\Longleftrightarrow\quad
       b_{j-1}<a_0^\top z\leq b_j,
\end{aligned}
\label{eq:quantiles}
\end{equation}
with interior boundaries indexed by $j=1,\ldots,m-1$, $b_0=-\infty$ and $b_m=\infty$. Here $\top$ denotes transpose. Then $\Pr(J=j)=\Phi(b_j)-\Phi(b_{j-1})=p_j$. Targets are sampled with replacement within region $j$ with probabilities $w_i/p_j$. Eight weighted principal-component directions define target scores through within-region empirical Gaussian quantiles. A fixed orthonormal source frame supplies $a_0$ and eight geometry directions. The geometry direction cycles across updates while $a_0$ remains fixed, thereby keeping the regional probability assignment unchanged.

Stable sorting pairs the source and target scores within each region.
The prescribed target marginal is already established by calibrated conditional sampling before this permutation; rank matching subsequently reorganizes the source--target assignments while preserving all sampled target-index multiplicities.
If $I_b$ and $I_b'$ are target indices before and after sorting in a batch of size $B$, and $K$ is an independent uniform batch position, then
\begin{equation}
\begin{aligned}
\sum_{b=1}^B\mathbf1\{I'_b=i\}
&=\sum_{b=1}^B\mathbf1\{I_b=i\},\\
\Pr(I'_K=i)
&=\E\!\left[\frac1B\sum_{b=1}^B\mathbf1\{I'_b=i\}\right]
=w_i.
\end{aligned}
\label{eq:batch-marginal}
\end{equation}
The indicator $\mathbf1\{\cdot\}$ equals one when its condition holds and zero otherwise. The unchanged source at $K$ retains law $Q$. Independently flipped images supply the velocity targets. \suppref{Sections~1 and 4.4} give the marginal and augmentation proofs and the numerical-precision and cost analysis.

\subsection{Flow dynamics and endpoint error}

At fixed time, the residual $U-u_\Pi(X_t,t)$ has zero conditional mean given $X_t$, yielding the loss decomposition in equation~\ref{eq:risk}. Write $\mu_t=\Law(X_t)$ and $u_t=u_\Pi(\cdot,t)$. Finite second moments give, for a smooth compactly supported test function $\varphi$,
\begin{equation}
\begin{aligned}
\frac{d}{dt}\E[\varphi(X_t)]
&=\E[\nabla\varphi(X_t)^\top U]\\
&=\int\nabla\varphi(x)^\top u_t(x)\,d\mu_t(x).
\end{aligned}
\label{eq:continuity}
\end{equation}
Here $\nabla$ is the spatial gradient. Under uniqueness of the flow and continuity equation, the trajectory $Y_t$ satisfying $\dot Y_t=u_t(Y_t)$ with $Y_0\sim Q$ has law $\mu_t$; a dot denotes differentiation in time.

For $0<T\leq1$, assume spatial Lipschitz constants $L$ and $\widehat L$ for $u_t$ and the learned field $\widehat u_t=v_{\widehat\theta}(\cdot,t)$. Here $\widehat\theta$ denotes the fitted network parameters. Couple $Y_t$ and the learned trajectory $\dot{\widehat X}_t=\widehat u_t(\widehat X_t)$ through $Y_0=\widehat X_0\sim Q$. Set $e_t=\widehat u_t-u_t$, $\widehat\mu_t=\Law(\widehat X_t)$ and $d(t)=\|\widehat X_t-Y_t\|_{L^2}$. The $L^2$ norm is the root mean square over this initial draw; the subscript $L^2(\nu)$ specifies averaging under probability law $\nu$. Subtracting the flow equations and applying Gr\"onwall's inequality gives
\begin{equation}
\begin{aligned}
d(t)&\leq\int_0^t\|e_s\|_{L^2(\widehat\mu_s)}\,ds
         +L\int_0^t d(s)\,ds,\\
d(T)&\leq\int_0^T e^{L(T-t)}
                  \|e_t\|_{L^2(\widehat\mu_t)}\,dt.
\end{aligned}
\label{eq:propagation}
\end{equation}
Assume the rollout law $\widehat\mu_t$ is absolutely continuous with respect to $\mu_t$, written $\widehat\mu_t\ll\mu_t$, with density ratio $r_t=d\widehat\mu_t/d\mu_t$. Define $A_t=\E_{\mu_t}\|e_t\|^2$, $B_t=\E_{\mu_t}\|e_t\|^4$ and $\chi_t^2=\E_{\mu_t}(r_t-1)^2$, with finite moments and integrable bounds on $[0,T]$. Changing measure and applying Cauchy--Schwarz yields
\begin{equation}
\begin{aligned}
\|e_t\|_{L^2(\widehat\mu_t)}^2
&=\int r_t\|e_t\|^2\,d\mu_t\\
&=A_t+\int(r_t-1)\|e_t\|^2\,d\mu_t\\
&\leq A_t+\sqrt{\chi_t^2B_t}.
\end{aligned}
\label{eq:rollout-error}
\end{equation}
For a learned field continuously differentiable along its flow, $D\widehat u_t$ denotes the spatial Jacobian and $\partial_t$ the time derivative. Define material acceleration $a_t=\partial_t\widehat u_t+D\widehat u_t\,\widehat u_t$ and $M_a=\sup_{t\leq T}\|a_t(\widehat X_t)\|_{L^2}<\infty$. The Euler local remainder is bounded by $h^2M_a/2$ for $h=T/k$~\cite{hairer1993solving}. With Euler states $V_j$ at $t_j=jh$, $V_0=\widehat X_0$, $V_k=\widehat X_T^{(k)}$ and root mean squared errors $\delta_j=\|\widehat X_{t_j}-V_j\|_{L^2}$,
\begin{equation}
\begin{aligned}
\delta_{j+1}&\leq(1+h\widehat L)\delta_j+\tfrac12h^2M_a,
\quad \delta_0=0,\\
\delta_k&\leq\frac{h^2M_a}{2}\sum_{j=0}^{k-1}(1+h\widehat L)^j\\
&\leq\frac{hM_a}{2}\frac{e^{\widehat L T}-1}{\widehat L}.
\end{aligned}
\label{eq:euler-error}
\end{equation}
The ratio is interpreted as $T$ when $\widehat L=0$. The interpolation gives $W_2(\mu_T,P)\leq(1-T)(\E_\Pi\|X-Z\|^2)^{1/2}$. Combining these bounds by the triangle inequality yields equation~\ref{eq:endpoint}. \suppref{Theorem~S1} gives the complete derivation and terminal-limit conditions.

\subsection{Low-dimensional experiments}

The two-dimensional benchmarks use the rare12 target and its equally weighted eight-mode counterpart, gaussian8. Both are centred and whitened, and are paired with a standard Gaussian source through independent sampling or the analytic polar construction. Each method uses the same 9,218-parameter network, trained for 3,500 updates with batch size 512 across 32 paired runs per target. Relative score and velocity MSE use the exact conditional fields. Endpoints are evaluated at Euler NFE 1, 2, 4, 8, 16 and 32. Exact two-sided paired tests and 50,000 paired bootstrap resamples supply the probabilities and intervals reported for this comparison. The constructions and reference fields are derived in \suppref{Section~3}. The six-method endpoint comparison uses five independent runs per method and target, with 3,500 optimizer updates per method. OFM is a direct map using one effective evaluation. \suppref{Section~8.4} gives the model-specific settings and complete budget comparison.

The transfer study uses source $\mathcal N(\delta e_2,I_d)$ and target $p\mathcal N(-(1-p)se_1,0.35^2I_d)+(1-p)\mathcal N(pse_1,0.35^2I_d)$. Here $\mathcal N$ denotes a Gaussian law, $d$ is the dimension, $I_d$ the identity matrix, $e_1,e_2$ the first two standard basis vectors, $p$ the rare-mode probability, $s$ the mode separation and $\delta$ the source displacement. A maximin design selects 36 combinations of these four parameters, with three paired runs per condition. Methods share target-component counts, source draws and conditional target noise. Each width-64 residual network is trained for 5,000 updates and generates 5,000 points at Euler NFE 1, 2, 4, 8 and 16. The diagnostic reference uses 64 Heun steps, or 128 field evaluations. Endpoint SWD uses root-mean-square projected discrepancies for the two-dimensional targets and mean absolute projected discrepancies for this transfer study.

Field gain is the relative reduction in displacement-regression loss on each method's interpolation law. Rollout penalty is the excess loss on learned trajectories with the same displacement labels; solver error is Euler-to-Heun endpoint MSE. Ordinary least squares relates condition-mean endpoint improvement to field gain and absolute reductions in rollout and solver errors. \suppref{Sections~6.1 and 9} give the full design, diagnostic definitions and prediction procedure.

\subsection{Image training and evaluation}

CIFAR-10~\cite{krizhevsky2009learning} is divided into 45,000 training and 5,000 validation images with class stratification; its official 10,000-image test set supplies final evaluation. Balanced targets use class probability 0.1. Long-tail targets use $p_c=100^{-c/9}/\sum_{\ell=0}^{9}100^{-\ell/9}$ for classes $c=0,\ldots,9$ and image weights $w_i=p_{y_i}/n_{y_i}$, where $y_i$ is the class label of image $i$ and $n_c$ is the number of training images in class $c$. Pixels are scaled to $[-1,1]$, with independent horizontal flips of probability 0.5 during training.

All image methods share a residual U-Net~\cite{ronneberger2015unet} with base width 96, GroupNorm~\cite{wu2018group}, attention and 20,118,915 parameters. AdamW~\cite{loshchilov2019decoupled} uses learning rate $2\times10^{-4}$ and batch size 256. Equation~\ref{eq:objective} is averaged over image coordinates and uniformly stratified times; evaluation uses exponential-moving-average weights with decay 0.9999. I-CFM, minibatch OT-CFM, two-stage rectified flow, DGSWP-CFM, MAC and G-MPCF receive a common envelope of $5.425\times10^{17}$ FLOPs per training run, matched within 1\%. Charges include method-specific pairing, auxiliary models and teacher computation. The shared dataset and evaluator preparation are excluded. \suppref{Tables~S6--S7} give the architecture and training schedules.

Generation integrates $\dot x_t=v_{\widehat\theta}(x_t,t)$ from a fresh source draw at $t=0$ to $t=1$ with uniform-step Euler. Primary endpoints use 50,000 images per model and 16 evaluations, with common initial noise within paired runs. Weighted FID uses class-weighted real Inception moments~\cite{heusel2017gans}; KID uses a cubic polynomial kernel~\cite{gretton2012kernel,binkowski2018demystifying}. Preprocessing and feature extraction are identical across methods~\cite{parmar2022aliased}. Minority precision and recall average three-neighbour feature-manifold measures~\cite{kynkaanniemi2019improved} for classes 7--9, the three least probable classes in WLT100. The balanced comparison uses the same class subset. Overall support measures use balanced real samples.

A native-resolution ResNet-18, trained on the training split and selected on validation, estimates generated class frequencies. Let $C_{ab}$ be its validation probability of predicting class $a$ given true class $b$, and let $q_{\rm pred}$ contain predicted frequencies in generated images. For a candidate class-probability vector $p$, simplex-constrained least squares under the label-shift assumption gives the corrected estimate $\widehat p$ and its distance from the prescribed target vector $p_{\rm target}$,
\begin{equation}
\begin{aligned}
\widehat p&=\arg\min_{\substack{p\geq0\\\mathbf1^\top p=1}}
                       \|Cp-q_{\rm pred}\|^2,\\
\operatorname{MassTV}&=\tfrac12\|\widehat p-p_{\rm target}\|_1.
\end{aligned}
\label{eq:masstv}
\end{equation}
Here $\mathbf1$ is the all-ones vector and $\|\cdot\|_1$ sums absolute coordinate values. This observation model follows class-probability estimation under label shift~\cite{lipton2018detecting}. Class-frequency plots use raw predicted probabilities or their ratios to the target law; MassTV uses the corrected probabilities. Feature maps use 50 principal components of the evaluator's 512-dimensional representations and t-SNE~\cite{vandermaaten2008visualizing}, fitted jointly across both target laws within each comparison. Enlargements share target-defined bounds for the displayed class to facilitate direct comparisons across methods. Unit-circle projections summarize directions about the common reference-map centre. \suppref{Sections~4.1 and 4.5} give evaluator calibration and feature-map settings.

The primary comparison uses eight independent pairs of training runs with shared random-number streams across methods. Benefits are comparator minus G-MPCF for errors and G-MPCF minus comparator for recall. One-sided paired tests enumerate all $2^8$ sign patterns, with separate Holm corrections for rare-event superiority, FID non-inferiority and rare-event non-inferiority~\cite{holm1979simple}. Percentile intervals resample complete pairs 100,000 times~\cite{efron1994introduction}. A component passes when its adjusted $P<0.05$ and its two-sided 95\% interval lies above zero. FID non-inferiority uses margin $\max(1,0.05\overline{\mathrm{FID}}_B)$ against comparator $B$, where the overbar averages its FID across training runs; MassTV and recall margins are 0.01 and 0.02. Non-inferiority tests add the corresponding margin to each paired benefit. \suppref{Table~S4} reports the component tests and joint criteria. Balanced transfer uses three paired runs and descriptive summaries; the component, scaling and selection experiments use separate sets of three paired runs.

\subsection{Ablations and model selection}

Five coupling operations isolate regional calibration and rank ordering. Mass relaxation adds $\rho=0.25,0.5,0.75$ between calibrated and uniform regional probabilities, sharing the endpoint models with the component experiment. Three runs at each of the eight distinct settings give 24 trained models, with initializations independent of the primary comparison.

Scaling varies training-data fraction, model width or computation around the full-data, width-96 reference. Eight distinct configurations each use three paired runs of I-CFM and G-MPCF, with 10,000 images per run at Euler NFE 1, 2, 4, 8 and 16. The settings and resource measurements are specified in \suppref{Tables~S12--S13}.

Validation compares a standard Gaussian source with a fitted 20-component, rank-64 probabilistic principal-component mixture~\cite{tipping1999mixtures}. For a fitted network with parameters $\theta$, let $T_{\theta,k}(z)$ be its numerical endpoint from source input $z$ at NFE $k$. Then the action $A_{\lambda,k}(z)=z+\lambda[T_{\theta,k}(z)-z]$ retains fraction $\lambda$ of the full displacement. Two sources, independent, OT-CFM and G-MPCF pairing, $\lambda\in\{0.25,0.5,0.75,1\}$ and $k\in\{1,2,4,8,16\}$, plus two source-retention controls, give 122 candidates. Using the three-run means, the rule admits actions within 0.01 of the best source-retention MassTV, retains those within one standard error of the minimum validation KID and selects the smallest NFE with fixed tie breaks. Each candidate uses 10,000 images per run over three runs. The selected action is evaluated on test data, followed by the full action grid. \suppref{Section~7} specifies the candidates and selection rule. Mean-based summaries use sample standard deviations across training runs; validation KID uses the standard error of the mean.

\clearpage
\begingroup\setstretch{1.13}
\bibliographystyle{scienceadvances_fullnames}
\bibliography{references}
\endgroup
\section*{Acknowledgments}
The authors used OpenAI Codex (GPT-6) to assist with manuscript drafting and the presentation of quantitative results.

\noindent\textbf{Funding:}\par
National Key R\&D Program of China, grant 2024YFB2606101.\par
National Natural Science Foundation of China, grant 52478418.\par
Sichuan Science and Technology Program, grant 2025ZNSFSC1294.

\noindent\textbf{Author contributions:}\par
Methodology: Y.H., X.L., Z.W., J.X.\par
Investigation: Y.H., X.L., Z.W., J.X.\par
Formal analysis: Y.H., X.L., Z.W., J.X.\par
Validation: L.J., Y.R.\par
Visualization: L.J.\par
Resources: X.H., Y.R.\par
Supervision: X.H., Y.R., C.H.\par
Funding acquisition: X.H.\par
Writing---original draft: Y.H.\par
Writing---review and editing: Y.H., X.H., X.L., Z.W., L.J., J.X., Y.R., C.H.

\noindent\textbf{Competing interests:} The authors declare they have no competing interests.

\noindent\textbf{Data, code, and materials availability:} CIFAR-10 is available from the \href{https://www.cs.toronto.edu/~kriz/cifar.html}{University of Toronto}~\cite{krizhevsky2009learning}, Kuzushiji-MNIST from the \href{https://github.com/rois-codh/kmnist}{KMNIST repository}~\cite{clanuwat2018classical}, and the UCI Human Activity Recognition dataset from \href{https://doi.org/10.24432/C54S4K}{UCI}~\cite{reyes2013har}. Figure and table source data and trained models are archived on \href{https://doi.org/10.5281/zenodo.23277850}{Zenodo}~\cite{huang2026gmpcfdata}. Data S1 contains a fixed snapshot of the code, also available from the \href{https://github.com/ybfo/g-mpcf/tree/7932bc5a4e7bbd7b375618dc4a83178ae55b58cb}{G-MPCF repository}. This study generated no new physical materials.

\section*{Supplementary Materials}
Supplementary Text; figs. S1 to S18; tables S1 to S19; data S1.
\clearpage
\begin{figure}[H]\centering
\includegraphics[width=\linewidth,height=0.68\textheight,keepaspectratio]{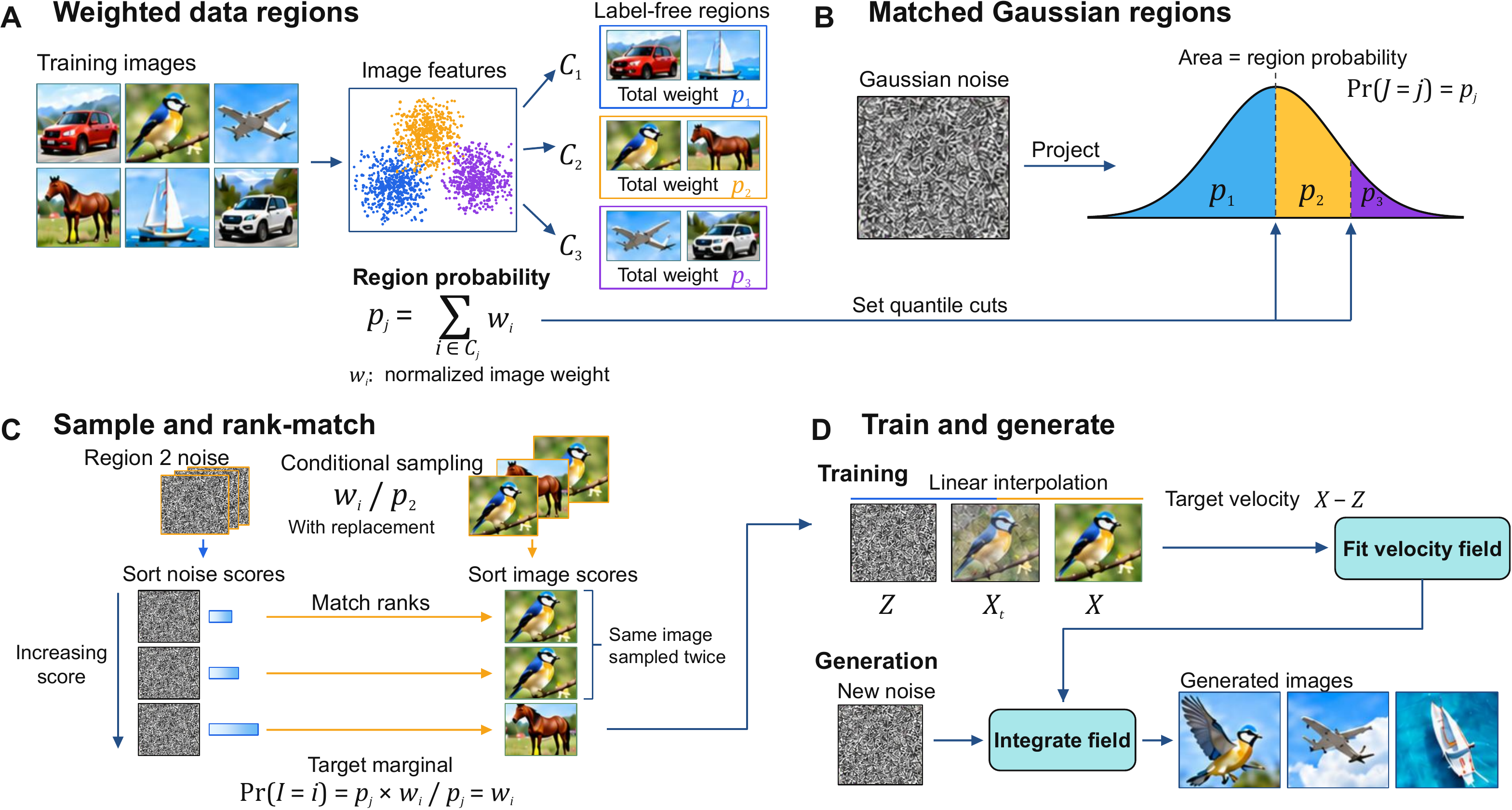}
\caption{\textbf{Mass-calibrated construction of training pairs.}
\textbf{(A)} Image features define label-free regions $C_j$ with probabilities
$p_j=\sum_{i\in C_j}w_i$, where $w_i$ is the normalized image weight.
\textbf{(B)} Standard Gaussian projection quantiles define source regions with matching probabilities and colours.
\textbf{(C)} Conditional sampling uses probabilities $w_i/p_j$ with replacement within each matched region.
This calibrated sampling step establishes the prescribed target marginal, while subsequent rank matching preserves repeated draws and rearranges source--target correspondences without changing the sampled target counts.
\textbf{(D)} Paired noise $Z$ and image $X$ define interpolants $X_t$ and target velocity $X-Z$.
Integrating the learned field from fresh noise generates images.
Images, feature clusters and interpolation states are schematic.}\label{fig:method}
\end{figure}
\clearpage
\begin{figure}[H]\centering
\includegraphics[width=\linewidth,height=0.77\textheight,keepaspectratio]{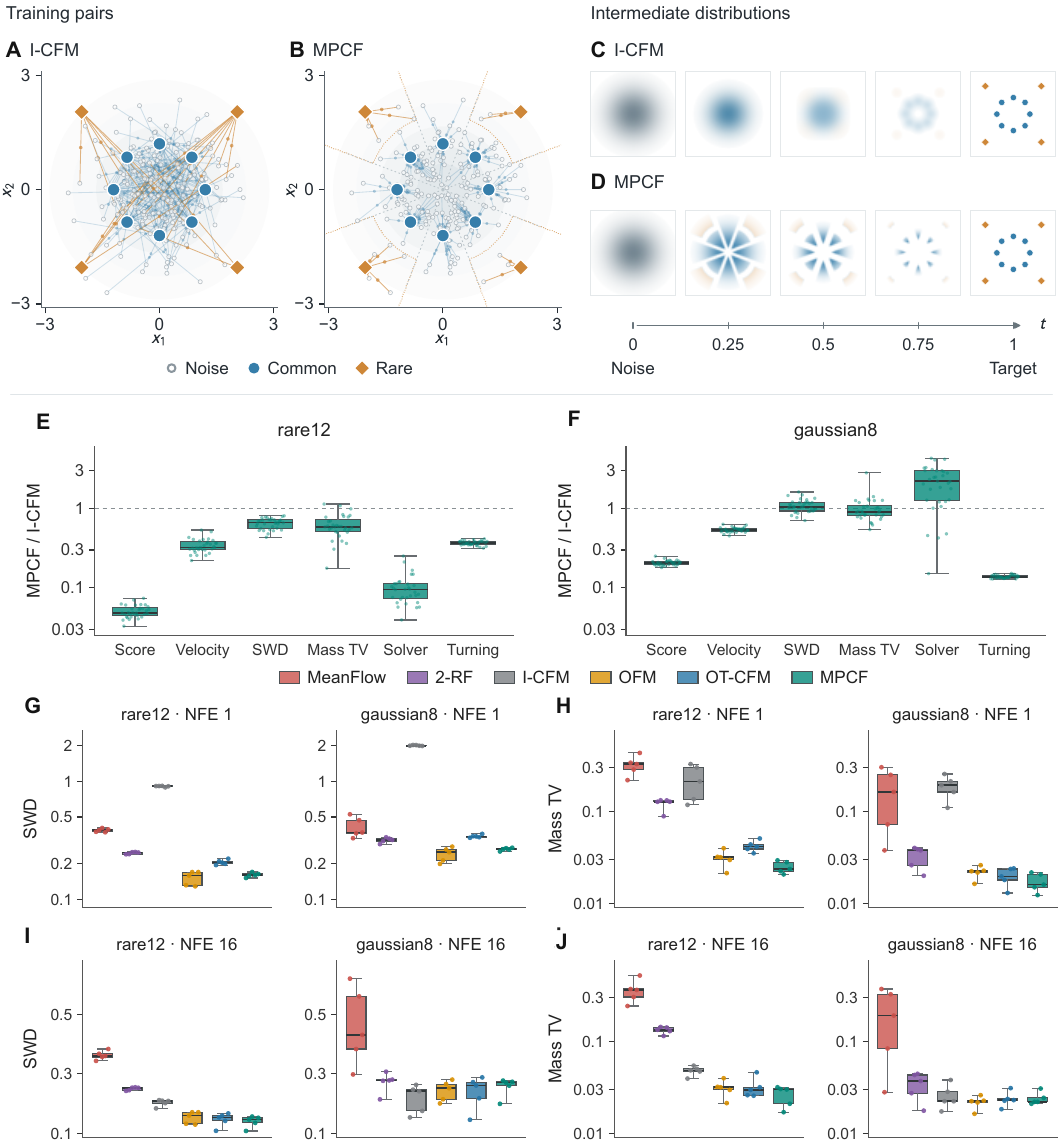}
\caption{\textbf{Rare-mode transport and field accuracy.} \textbf{(A, B)} Independent and mass-calibrated training pairs for rare12. Open circles denote sources; blue circles and orange diamonds identify common and rare targets. \textbf{(C, D)} Exact intermediate densities with shared intensity scales within each time column. \textbf{(E, F)} MPCF-to-I-CFM ratios for six diagnostics across 32 paired runs per target; endpoint and path metrics use 16 evaluations. Dashed lines mark equality. \textbf{(G, H)} Endpoint SWD and mode-mass total variation at one evaluation. \textbf{(I, J)} Corresponding results at 16 evaluations, with five runs per method and target throughout \textbf{(G--J)}. Methods follow decreasing mean rare12 SWD at 16 evaluations. OFM uses one evaluation throughout. Boxes show medians and interquartile ranges, whiskers span the observed ranges and dots denote individual runs. Vertical axes in \textbf{(E--H, J)} are logarithmic; those in \textbf{(I)} are linear. NFE, number of field evaluations.}\label{fig:controlled}
\end{figure}
\clearpage
\begin{figure}[H]\centering
\includegraphics[width=\linewidth,height=0.80\textheight,keepaspectratio]{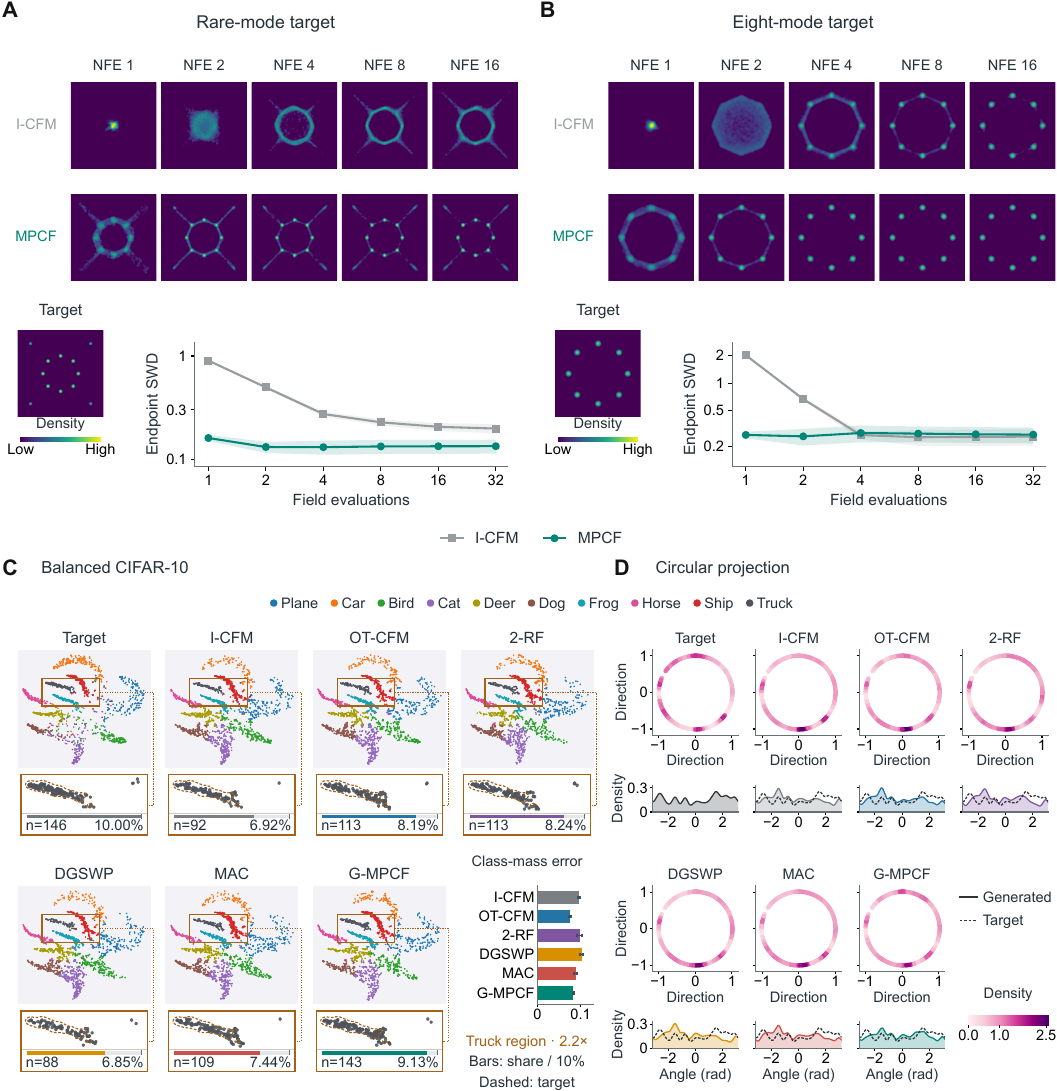}
\caption{\textbf{Target recovery and feature-distribution structure.} \textbf{(A, B)} Endpoint densities and sliced-Wasserstein distance for rare12 and gaussian8. Maps pool 9,000 outputs from three runs with common smoothing and logarithmic colour scales within each target; curves summarize 32 paired runs. Both SWD axes are logarithmic. \textbf{(C)} Balanced CIFAR-10 t-distributed stochastic neighbour embedding (t-SNE) maps and corrected class-mass total variation (MassTV). Maps contain 1,536 images; colours denote reference labels or predicted classes. Truck enlargements share target-defined bounds and the target 80\% density contour. Counts refer to displayed trucks; inset bars and percentages give full-evaluation truck frequencies against the 10\% target. \textbf{(D)} Unit-circle projections and angular densities of the same maps; dashed curves show the target. Image evaluations use three training runs and 50,000 images per model at 16 Euler evaluations. SWD curves and metric bars show means; bands and error bars show standard deviations.}\label{fig:images}
\end{figure}
\clearpage
\begin{figure}[H]\centering
\includegraphics[width=\linewidth,height=0.68\textheight,keepaspectratio]{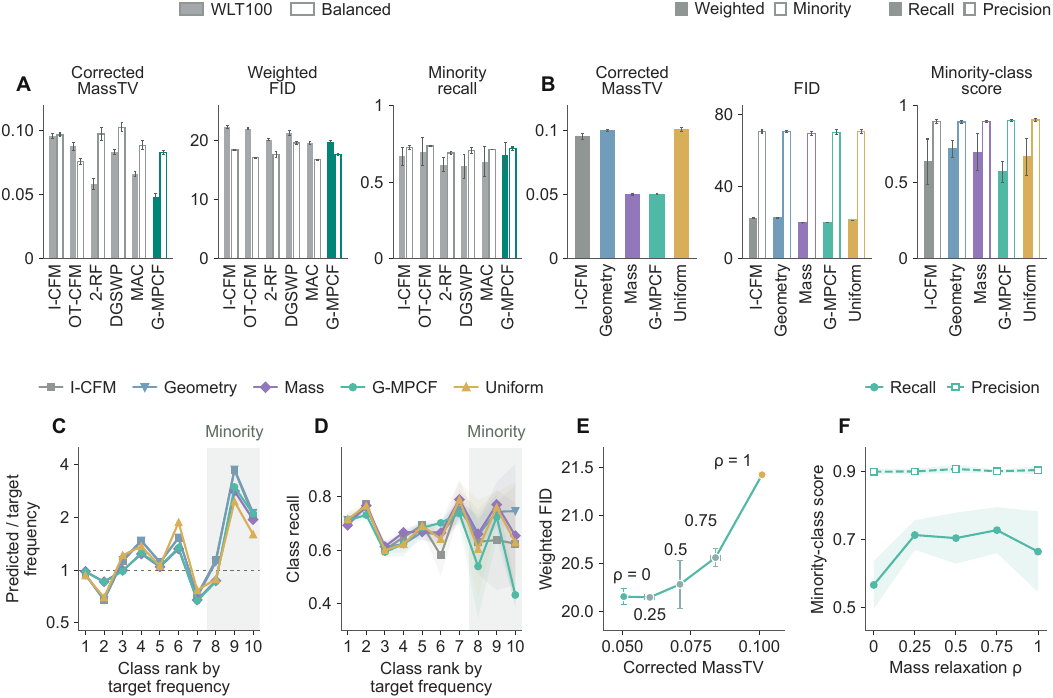}
\caption{\textbf{Image-generation performance and its regional basis.} \textbf{(A)} Corrected class-mass total variation, weighted Fr\'echet Inception distance (FID) and minority recall across six coupling methods. Filled and open bars denote WLT100 and balanced targets, with eight and three training runs, respectively. \textbf{(B)} Allocation, fidelity and coverage under five coupling operations. Filled and open bars denote weighted and minority-class FID in the middle chart, and recall and precision in the right chart. \textbf{(C, D)} Predicted-to-target class-frequency ratios and class recall, ordered by decreasing target probability. Shading marks minority classes; the dashed line marks matched frequencies. \textbf{(E, F)} Allocation, fidelity and coverage as regional probabilities vary from the target law at $\rho=0$ to uniform at $\rho=1$. Component and relaxation comparisons use three training runs. All image evaluations use 50,000 images per model at 16 Euler evaluations. Bars and curves show means; error bars and bands show standard deviations.}\label{fig:ablation}
\end{figure}
\clearpage
\begin{figure}[H]\centering
\includegraphics[width=\linewidth,height=0.98\textheight,keepaspectratio]{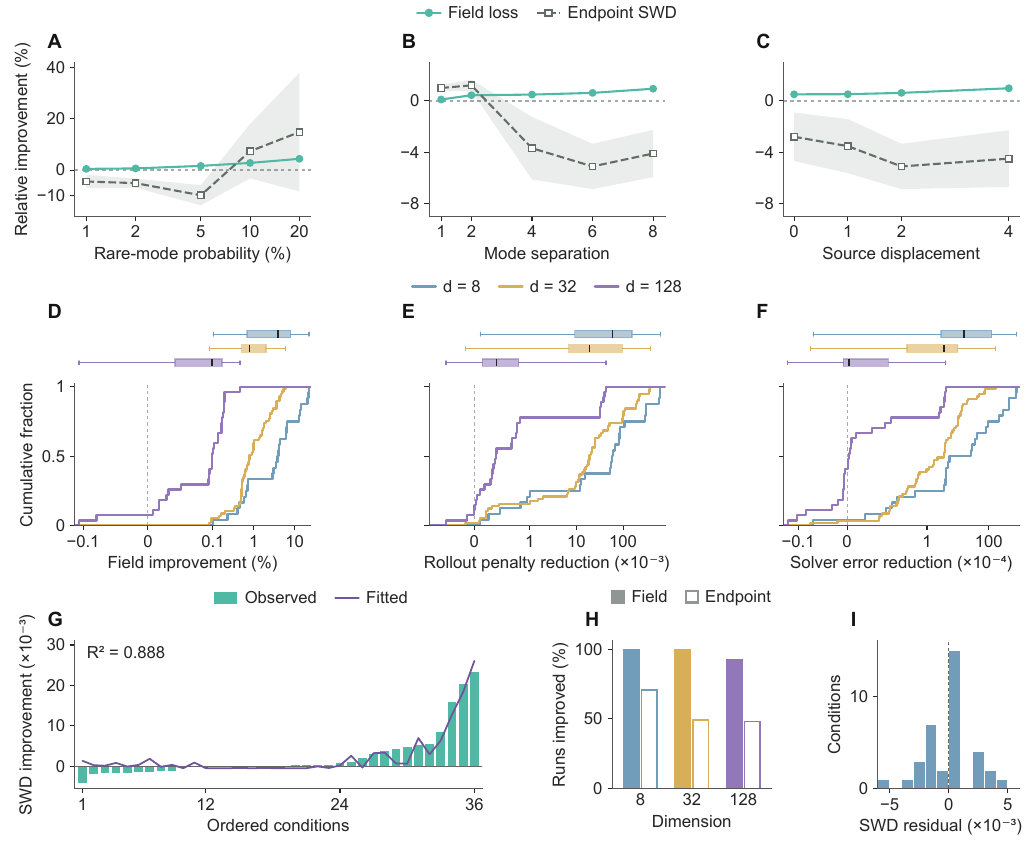}
\caption{\textbf{From field approximation to endpoint accuracy.} \textbf{(A--C)} Relative field-loss and endpoint-SWD reductions as rare probability, mode separation or source displacement varies at dimension 32. Curves show means and bands show standard deviations. \textbf{(D--F)} Cumulative distributions of relative field improvement and of rollout-penalty and solver-error differences. The latter are I-CFM minus G-MPCF. Boxes show medians and interquartile ranges; whiskers span all observations. Horizontal axes are symmetric logarithmic. \textbf{(G)} Observed endpoint improvement and the three-diagnostic regression, ordered by observed improvement; $R^2$ is the coefficient of determination. \textbf{(H)} Fractions of runs with improved field loss or endpoint SWD. \textbf{(I)} Observed minus fitted endpoint improvement. The 36 conditions use three paired training runs each, giving 24, 57 and 27 runs in dimensions 8, 32 and 128. Evaluation uses 5,000 points per training run, 16 Euler evaluations and a 64-step Heun reference.}\label{fig:boundary}
\end{figure}
\clearpage
\begin{figure}[H]\centering
\includegraphics[width=\linewidth,height=0.68\textheight,keepaspectratio]{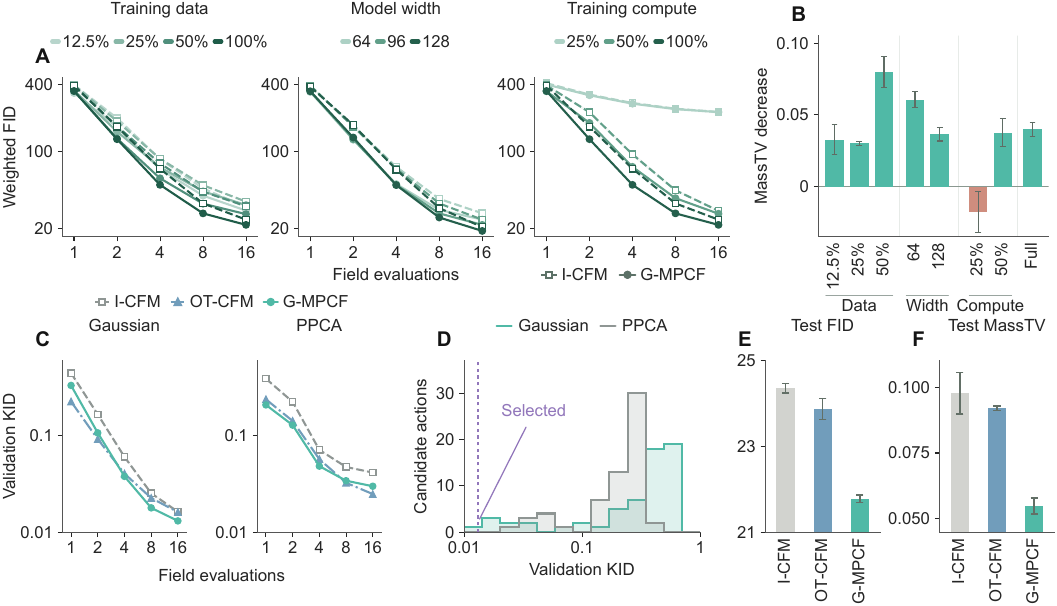}
\caption{\textbf{Resource budgets and validation-guided image generation.} \textbf{(A)} Weighted FID across inference budgets and resource levels, relative to full training data, base width 96 and full computation. \textbf{(B)} Paired MassTV decrease at 16 evaluations, I-CFM minus G-MPCF. \textbf{(C)} Validation kernel Inception distance (KID) for Gaussian and probabilistic principal-component analysis (PPCA) mixture sources at full displacement. Both axes in \textbf{(A, C)} are logarithmic. \textbf{(D)} Mean validation KID across 122 candidates in logarithmic bins; the dashed line marks selected Gaussian-source G-MPCF with full displacement at 16 evaluations. \textbf{(E, F)} Test FID and MassTV under the selected source and budget; vertical axes start at 21 and 0.045. \textbf{(A--C, E, F)} show means across three paired training runs, with 10,000 images per run. Bands and error bars show standard deviations, except for standard errors of the mean in \textbf{(C)}.}\label{fig:budget}
\end{figure}
\fi
\ifdefined\SAonlySupplement\clearpage
\begingroup
\setstretch{1.14}\normalsize
\setcounter{section}{0}\setcounter{figure}{0}\setcounter{table}{0}\setcounter{equation}{0}
\setcounter{secnumdepth}{2}
\renewcommand{\thefigure}{S\arabic{figure}}\renewcommand{\thetable}{S\arabic{table}}\renewcommand{\theequation}{S\arabic{equation}}
\renewcommand{\theHfigure}{s.\arabic{figure}}\renewcommand{\theHtable}{s.\arabic{table}}\renewcommand{\theHequation}{s.\arabic{equation}}\renewcommand{\theHsection}{s.\arabic{section}}
\begin{center}
{\Large\bfseries Supplementary Materials for\par}\vspace{10pt}
{\large\bfseries \scititle\par}\vspace{9pt}
Yubo Huang \textit{et al.}\par
\small Corresponding authors: \href{mailto:huxiongyu@swjtu.edu.cn}{huxiongyu@swjtu.edu.cn} and \href{mailto:u3014074@connect.hku.hk}{u3014074@connect.hku.hk}
\end{center}
\vspace{16pt}
\noindent\textbf{This PDF file includes:}\par
Supplementary Text (sections 1 to 9)\par
Figs. S1 to S18\par
Tables S1 to S19\par
\par
\vspace{8pt}\noindent\textbf{Other Supplementary Material:}\par
Data S1. Code for regional probability matching and the reported experiments.\par
\vspace{16pt}\noindent\textbf{Supplementary Text contents}\par
1. Mass-calibrated coupling\par
2. From a regression field to a numerical endpoint\par
3. Exact low-dimensional constructions and replication\par
4. Image experiments and primary statistics\par
5. Component interventions\par
6. Transfer conditions and image scaling\par
7. Validation-based source and action selection\par
8. Source adequacy and additional coupling comparisons\par
9. Additional validation and computational characterization\par
\clearpage
\section*{Supplementary Text}
\section{Mass-calibrated coupling}
\label{sec:coupling-proofs}
\subsection{Calibrated sampling and rank matching}

Fix $N$ observations $x_1,\ldots,x_N\in\R^D$ in $D$-dimensional real space, positive weights $w_i$ summing to one and a partition of their indices into $m$ regions $C_1,\ldots,C_m$. We write $\delta_x$ for a point mass at $x$, $\Pr$ for probability, $\E$ for expectation and $\Law$ for the probability law of a random variable. Define the index law, target law and regional probabilities by
\begin{equation}
 P_I(i)=w_i,\qquad P_X=\sum_{i=1}^Nw_i\delta_{x_i},\qquad
 p_j=\sum_{i\in C_j}w_i>0,\qquad \sum_{j=1}^mp_j=1.
\end{equation}
Let $Q=\N(0,I_D)$ be the standard Gaussian source, with $I_D$ the identity matrix, and let $a_0$ be a fixed unit vector and let $\Phi$ be the standard-normal cumulative distribution function. Gaussian quantiles partition the source into regions with the prescribed probabilities,
\begin{equation}
\begin{aligned}
 b_0&=-\infty,\quad b_m=\infty,\quad
 b_j=\Phi^{-1}\!\left(\sum_{\ell=1}^jp_\ell\right),\quad j=1,\ldots,m-1,\\
 J(z)=j&\quad\Longleftrightarrow\quad b_{j-1}<a_0^\top z\leq b_j.
\end{aligned}
\end{equation}
Here $b_j$ are projection cut points, $J(z)$ is the source-region index, $\Phi^{-1}$ is the Gaussian quantile function and $\top$ denotes transpose. For a source draw $Z\sim Q$, sample target index $I$ from its region with
\begin{equation}
 \Pr(I=i\mid Z=z)=\frac{w_i}{p_{J(z)}}\mathbf1\{i\in C_{J(z)}\}.
 \label{eq:conditional-index}
\end{equation}

The indicator $\mathbf1\{\cdot\}$ equals one when its condition holds and zero otherwise. For a batch of size $B$ with independent source draws $Z_1,\ldots,Z_B$, write $B_j=\{b:J(Z_b)=j\}$ and $n_j=|B_j|$. Draw $I_b$ from this conditional law at each position. Within region $j$, sort the source and target scores and match equal ranks. The resulting permutation $\sigma_j:B_j\to B_j$ gives $I'_b=I_{\sigma_j(b)}$. Stable sorting resolves ties and repeated target indices; empty and singleton regions require no reassignment. Write $\Gamma(Q,P_X)$ for the joint laws with marginals $Q$ and $P_X$.

\begin{suppproposition}[Marginals of the calibrated batch coupling]
\label{prop:population-marginals}\label{prop:batch-marginals}
The population sampler satisfies $\Pr(J(Z)=j)=p_j$ and $\Pr(I=i)=w_i$. Rank matching preserves every realized target-index multiplicity. If $K$ is uniform on $\{1,\ldots,B\}$ and independent of the batch, then $\Pi_B=\Law(Z_K,x_{I'_K})\in\Gamma(Q,P_X)$.
\end{suppproposition}
\begin{suppproof}
Since $a_0^\top Z\sim\N(0,1)$, the quantile construction gives
\begin{equation}
\begin{aligned}
 \Pr(J(Z)=j)
 &=\Phi(b_j)-\Phi(b_{j-1})\\
 &=\sum_{\ell=1}^jp_\ell-\sum_{\ell=1}^{j-1}p_\ell=p_j,\\
 \Pr(I=i)
 &=\sum_{\ell=1}^m\Pr(J(Z)=\ell)\Pr(I=i\mid J(Z)=\ell)
 =p_j\frac{w_i}{p_j}=w_i,\qquad i\in C_j.
\end{aligned}
\end{equation}

Let $N_i$ and $N'_i$ count index $i$ before and after matching. Because $\sigma_j$ is a permutation of the positions in $B_j$,
\begin{equation}
\begin{aligned}
 N'_i
 &=\sum_{b\in B_j}\mathbf1\{I_{\sigma_j(b)}=i\}
 =\sum_{b\in B_j}\mathbf1\{I_b=i\}=N_i,\qquad i\in C_j,\\
 \sum_{i\in C_j}N'_i&=\sum_{i\in C_j}N_i=n_j.
\end{aligned}
\label{eq:multiset}
\end{equation}
Conditioning first on the full source batch yields
\begin{equation}
\begin{aligned}
 \E[N'_i\mid Z_1,\ldots,Z_B]
 &=\sum_{b\in B_j}\Pr(I_b=i\mid Z_1,\ldots,Z_B)
 =n_j\frac{w_i}{p_j},\\
 \E[N'_i]&=\frac{w_i}{p_j}\E[n_j]
 =\frac{w_i}{p_j}\sum_{b=1}^B\Pr(J(Z_b)=j)=Bw_i.
\end{aligned}
\end{equation}
The independent uniform position therefore has the required target law. For every measurable source set $A$,
\begin{equation}
\begin{aligned}
 \Pr(I'_K=i)&=\E\!\left[\frac1B\sum_{b=1}^B\mathbf1\{I'_b=i\}\right]
 =\frac{\E[N'_i]}{B}=w_i,\\
 \Pr(Z_K\in A)&=\frac1B\sum_{b=1}^B\Pr(Z_b\in A)=Q(A).
\end{aligned}
\end{equation}
Consequently, $\Law(x_{I'_K})=P_X$ and $\Pi_B\in\Gamma(Q,P_X)$. For an integrable pair loss $\ell$, the same uniform-position identity gives
\begin{equation}
 \E\!\left[\frac1B\sum_{b=1}^B\ell(Z_b,x_{I'_b})\right]
 =\E[\ell(Z_K,x_{I'_K})]
 =\int\ell(z,x)\,d\Pi_B(z,x).
\end{equation}
Thus the expected minibatch objective is the flow-matching risk for the calibrated coupling.
\end{suppproof}

\subsection{Mass perturbations and augmentation}

\begin{suppcorollary}[Regional probability perturbations]
\label{cor:mass-relaxation}
Let $p=(p_1,\ldots,p_m)$ and $q=(q_1,\ldots,q_m)$ denote calibrated and perturbed regional probability vectors, and let $\TV$ denote total variation distance. If the source regions have probabilities $q_j$ and retain the conditional target weights $w_i/p_j$, the induced index law is $P_{I,q}(i)=q_jw_i/p_j$ for $i\in C_j$, and $\TV(P_{I,q},P_I)=\TV(q,p)$. For the uniform vector $u$ with $u_j=1/m$ and relaxation fraction $0\leq\rho\leq1$, the mixture $q_\rho=(1-\rho)p+\rho u$ satisfies
\begin{equation}
 \TV(P_{I,q_\rho},P_I)=\rho\TV(u,p).
 \label{eq:relax}
\end{equation}
\end{suppcorollary}
\begin{suppproof}
The total-probability calculation above gives $P_{I,q}(i)=q_jw_i/p_j$. Summing within each region gives
\begin{equation}
\begin{aligned}
 \TV(P_{I,q},P_I)
 &=\frac12\sum_{j=1}^m\sum_{i\in C_j}
       \left|\frac{q_jw_i}{p_j}-w_i\right|\\
 &=\frac12\sum_{j=1}^m\frac{|q_j-p_j|}{p_j}\sum_{i\in C_j}w_i
 =\frac12\sum_{j=1}^m|q_j-p_j|=\TV(q,p),\\
 \TV(P_{I,q_\rho},P_I)
 &=\frac12\sum_j|\rho(u_j-p_j)|=\rho\TV(u,p).
\end{aligned}
\end{equation}
The induced observation law is $P_{X,q}=\sum_iP_{I,q}(i)\delta_{x_i}$. Summing any coincident observations gives $\TV(P_{X,q},P_X)\leq\TV(P_{I,q},P_I)$, with equality for distinct $x_i$.
\end{suppproof}

Independent augmentation preserves the corresponding weighted mixture. For a measurable augmentation $H(x,\xi)$, draw its random augmentation variable $\xi$ independently of the batch and $K$. Proposition~\ref{prop:batch-marginals} gives, for every measurable set $A$,
\begin{equation}
\begin{aligned}
 \Pr(H(x_{I'_K},\xi)\in A)
 &=\sum_i\Pr(I'_K=i)\Pr(H(x_i,\xi)\in A\mid I'_K=i)\\
 &=\sum_iw_i\Pr(H(x_i,\xi)\in A).
\end{aligned}
\label{eq:augmentation}
\end{equation}

\section{From a regression field to a numerical endpoint}
\label{sec:endpoint-proof}
\subsection{Continuity and regression identities}

Let $Q$ and $P$ be source and target laws, and fix $\Pi\in\Gamma(Q,P)$ with finite second moments and $(Z,X)\sim\Pi$. For interpolation time $t\in[0,1]$, write $X_t=(1-t)Z+tX$, $U=X-Z$, $\mu_t=\Law(X_t)$ and $u_t(x)=\E[U\mid X_t=x]$. Here $X_t$ is the interpolated state, $U$ the displacement, $\mu_t$ the interpolation law and $u_t$ its conditional velocity field. We use Euclidean vector norms and the continuity and regression identities underlying flow matching~\cite{lipman2023flow,tong2024improving,albergo2023stochastic}.

Let $\varphi\in C_c^1(\R^D)$ be a continuously differentiable test function with compact support. Writing $\nabla$ for the spatial gradient, the chain-rule integrand is bounded by $\|\nabla\varphi\|_\infty\|U\|$, where the infinity norm is the supremum of the gradient norm. Dominated differentiation and conditional expectation give
\begin{equation}
\begin{aligned}
 \frac{d}{dt}\int\varphi(x)\,d\mu_t(x)
 &=\frac{d}{dt}\E[\varphi(X_t)]
 =\E[\nabla\varphi(X_t)^\top U]\\
 &=\E\!\left[\nabla\varphi(X_t)^\top\E(U\mid X_t)\right]
 =\int\nabla\varphi(x)^\top u_t(x)\,d\mu_t(x).
\end{aligned}
\label{seq:continuity}
\end{equation}
Hence $\mu_t$ satisfies $\partial_t\mu_t+\nabla\!\cdot(\mu_tu_t)=0$ weakly, with $\mu_0=Q$ and $\mu_1=P$; $\partial_t$ and $\nabla\!\cdot$ denote the time derivative and spatial divergence.

For a square-integrable field $f(\cdot,t)\in L^2(\mu_t)$, define the residual $R_t=U-u_t(X_t)$ and the field difference $D_t=u_t(X_t)-f(X_t,t)$. Their conditional orthogonality gives
\begin{equation}
\begin{aligned}
 \E[R_t\mid X_t]&=\E[U\mid X_t]-u_t(X_t)=0,\\
 \E[R_t^\top D_t]&=\E\!\left[\E[R_t\mid X_t]^\top D_t\right]=0,\\
 \E\|R_t\|^2&=\E\!\left[\operatorname{tr}\E[R_tR_t^\top\mid X_t]\right]
 =\E[\operatorname{tr}\Var(U\mid X_t)].
\end{aligned}
\end{equation}
Here $\Var(U\mid X_t)$ is the conditional covariance matrix and $\operatorname{tr}$ is the trace. Expanding the regression loss therefore yields
\begin{equation}
\begin{aligned}
 \E\|U-f(X_t,t)\|^2
 &=\E\|R_t+D_t\|^2\\
 &=\E\|R_t\|^2+2\E[R_t^\top D_t]+\E\|D_t\|^2\\
 &=\E[\operatorname{tr}\Var(U\mid X_t)]
  +\int\|u_t(x)-f(x,t)\|^2\,d\mu_t(x).
\end{aligned}
\label{eq:decomp}
\end{equation}
Integrating over the training-time distribution gives the training-risk decomposition.

\subsection{Endpoint error theorem and proof}

Fix terminal time $0<T\leq1$ and write $\widehat u_t=v_{\widehat\theta}(\cdot,t)$ for the velocity network with fitted parameters $\widehat\theta$. Assume that $u_t$ and $\widehat u_t$ are uniformly spatially Lipschitz on $[0,T]$, with constants $L$ and $\widehat L$. Their ordinary differential equations $\dot Y_t=u_t(Y_t)$ and $\dot{\widehat X}_t=\widehat u_t(\widehat X_t)$ have unique square-integrable solutions from a common draw $Y_0=\widehat X_0\sim Q$, where dots denote time derivatives. Equation~\ref{seq:continuity} and uniqueness give $\Law(Y_t)=\mu_t$.

Assume the learned-rollout law $\widehat\mu_t=\Law(\widehat X_t)$ is absolutely continuous with respect to $\mu_t$, written $\widehat\mu_t\ll\mu_t$, for almost every $t$. Define the field error, density ratio and moments by
\begin{equation}
\begin{aligned}
 e_t&=\widehat u_t-u_t,\qquad r_t=\frac{d\widehat\mu_t}{d\mu_t},\\
 A_t&=\E_{\mu_t}\|e_t\|^2,\qquad
 B_t=\E_{\mu_t}\|e_t\|^4,\qquad
 \chi_t^2=\E_{\mu_t}(r_t-1)^2.
\end{aligned}
\end{equation}
The quantity $\chi_t^2$ is the chi-squared divergence of $\widehat\mu_t$ from $\mu_t$. These moments are finite, and $t\mapsto(A_t+\sqrt{\chi_t^2B_t})^{1/2}$ is integrable on $[0,T]$. Assume also that $\widehat u_t$ is continuously differentiable in space and time along its flow, with material acceleration satisfying
\begin{equation}
 a_t(x)=\partial_t\widehat u_t(x)+D\widehat u_t(x)\widehat u_t(x),\qquad
 M_a=\sup_{0\leq t\leq T}\|a_t(\widehat X_t)\|_{L^2}<\infty.
\end{equation}
Here $D\widehat u_t$ is the spatial Jacobian, and $M_a$ bounds the root mean squared material acceleration. The $L^2$ norm averages squared vector norms over the common initial draw and takes the square root; $\|f\|_{L^2(\nu)}^2=\int\|f(x)\|^2\,d\nu(x)$ for a probability law $\nu$. We use $W_2$ for the 2-Wasserstein distance between probability laws.

\begin{supptheorem}[Error of the numerical endpoint]
\label{thm:endpoint}
Under the preceding assumptions, take $k\geq1$ Euler steps with $h=T/k$ from the common initial draw. For the resulting endpoint $\widehat X_T^{(k)}$,
\begin{equation}
\begin{aligned}
 W_2(\Law(\widehat X_T^{(k)}),P)
 &\leq\int_0^T e^{L(T-t)}
          \left(A_t+\sqrt{\chi_t^2B_t}\right)^{1/2}dt\\
 &\quad+\frac{hM_a}{2}\frac{e^{\widehat L T}-1}{\widehat L}
       +(1-T)\left(\E_\Pi\|X-Z\|^2\right)^{1/2}.
\end{aligned}
\label{eq:bound}
\end{equation}
At $\widehat L=0$, the ratio in the Euler term is defined by continuity as $T$.
\end{supptheorem}

\begin{suppproof}
\emph{Flow approximation.} Subtract the two ODEs and add and subtract $u_t(\widehat X_t)$. Their common initial value gives
\begin{equation}
\begin{aligned}
 \widehat X_t-Y_t
 &=\int_0^t\left[\widehat u_s(\widehat X_s)-u_s(Y_s)\right]ds\\
 &=\int_0^t\left[e_s(\widehat X_s)+u_s(\widehat X_s)-u_s(Y_s)\right]ds.
\end{aligned}
\end{equation}
Minkowski's inequality and the Lipschitz bound for $u_t$ imply
\begin{equation}
 d(t):=\|\widehat X_t-Y_t\|_{L^2}
 \leq\int_0^t\|e_s(\widehat X_s)\|_{L^2}ds+L\int_0^td(s)ds.
\end{equation}
Set $g(t)=\|e_t\|_{L^2(\widehat\mu_t)}$ and $D(t)=\int_0^tg(s)\,ds+L\int_0^td(s)\,ds$. Then $d(t)\leq D(t)$ and, almost everywhere,
\begin{equation}
\begin{aligned}
 D'(t)&=g(t)+Ld(t)\leq g(t)+LD(t),\qquad D(0)=0,\\
 \frac{d}{dt}\bigl(e^{-Lt}D(t)\bigr)&\leq e^{-Lt}g(t),\\
 D(T)&\leq e^{LT}\int_0^Te^{-Lt}g(t)\,dt.
\end{aligned}
\end{equation}
Consequently,
\begin{equation}
 d(T)\leq\int_0^Te^{L(T-t)}\|e_t\|_{L^2(\widehat\mu_t)}dt.
 \label{eq:gronwall}
\end{equation}
Changing measure to $\mu_t$ and applying Cauchy--Schwarz gives
\begin{equation}
\begin{aligned}
 \|e_t\|_{L^2(\widehat\mu_t)}^2
 &=\int r_t\|e_t\|^2\,d\mu_t\\
 &=A_t+\int(r_t-1)\|e_t\|^2\,d\mu_t\\
 &\leq A_t+
       \left(\int(r_t-1)^2d\mu_t\right)^{1/2}
       \left(\int\|e_t\|^4d\mu_t\right)^{1/2}\\
 &=A_t+\sqrt{\chi_t^2B_t}.
\end{aligned}
\label{eq:change-measure}
\end{equation}
Substituting the square root into equation~\ref{eq:gronwall} gives the flow-approximation bound.

\emph{Euler integration.} Set $t_j=jh$, $V_0=\widehat X_0$ and $V_{j+1}=V_j+h\widehat u_{t_j}(V_j)$. Along the exact learned flow,
\begin{equation}
\begin{aligned}
 \frac{d}{dt}\widehat u_t(\widehat X_t)
 &=\partial_t\widehat u_t(\widehat X_t)
   +D\widehat u_t(\widehat X_t)\frac{d\widehat X_t}{dt}\\
 &=\partial_t\widehat u_t(\widehat X_t)
   +D\widehat u_t(\widehat X_t)\widehat u_t(\widehat X_t)
 =a_t(\widehat X_t).
\end{aligned}
\end{equation}
Integrating first the velocity and then the position gives a local Euler remainder $R_j$,
\begin{equation}
\begin{aligned}
 \widehat u_s(\widehat X_s)
 &=\widehat u_{t_j}(\widehat X_{t_j})+\int_{t_j}^{s}a_r(\widehat X_r)\,dr,\\
 \widehat X_{t_{j+1}}-\widehat X_{t_j}
 &=\int_{t_j}^{t_{j+1}}\widehat u_s(\widehat X_s)\,ds\\
 &=h\widehat u_{t_j}(\widehat X_{t_j})+R_j,\qquad
 R_j=\int_{t_j}^{t_{j+1}}\int_{t_j}^{s}a_r(\widehat X_r)\,dr\,ds,\\
 \|R_j\|_{L^2}
 &\leq\int_{t_j}^{t_{j+1}}\int_{t_j}^{s}M_a\,dr\,ds
 =\tfrac12h^2M_a.
\end{aligned}
\end{equation}
With $E_j=\widehat X_{t_j}-V_j$ and $\delta_j=\|E_j\|_{L^2}$, subtracting the Euler update yields
\begin{equation}
\begin{aligned}
 E_{j+1}&=E_j+h\bigl[\widehat u_{t_j}(\widehat X_{t_j})-\widehat u_{t_j}(V_j)\bigr]+R_j,\\
 \delta_{j+1}
 &\leq\delta_j+h\|\widehat u_{t_j}(\widehat X_{t_j})-\widehat u_{t_j}(V_j)\|_{L^2}+\|R_j\|_{L^2}\\
 &\leq(1+h\widehat L)\delta_j+\tfrac12h^2M_a,\qquad\delta_0=0.
\end{aligned}
\end{equation}
For $\widehat L>0$, iterating the recurrence and summing the geometric series gives
\begin{equation}
\begin{aligned}
 \delta_k
 &\leq\frac{h^2M_a}{2}\sum_{j=0}^{k-1}(1+h\widehat L)^j\\
 &=\frac{hM_a}{2\widehat L}\left((1+h\widehat L)^k-1\right)\\
 &\leq\frac{hM_a}{2}\frac{e^{\widehat L T}-1}{\widehat L}.
\end{aligned}
\end{equation}
The last inequality uses $1+h\widehat L\leq e^{h\widehat L}$ and $kh=T$. For $\widehat L=0$, the original sum equals $kh^2M_a/2=hTM_a/2$, giving the stated continuous extension.

\emph{Terminal truncation and combination.} The original coupling pairs $X_T=(1-T)Z+TX$ with $X$. Hence
\begin{equation}
\begin{aligned}
 W_2(\mu_T,P)^2
 &\leq\E\|X_T-X\|^2\\
 &=(1-T)^2\E_\Pi\|Z-X\|^2.
\end{aligned}
\end{equation}
\Needspace{5\baselineskip}
Applying the triangle inequality with $\Law(Y_T)=\mu_T$ gives
\begin{equation}
\begin{aligned}
 W_2(\Law(V_k),P)
 &\leq W_2(\Law(V_k),\widehat\mu_T)
       +W_2(\widehat\mu_T,\mu_T)+W_2(\mu_T,P)\\
 &\leq\|V_k-\widehat X_T\|_{L^2}
       +\|\widehat X_T-Y_T\|_{L^2}+W_2(\mu_T,P).
\end{aligned}
\end{equation}
Using the bounds above and $V_k=\widehat X_T^{(k)}$ gives equation~\ref{eq:bound}.
\end{suppproof}

The estimate holds at $T=1$ whenever the hypotheses extend to the endpoint. For a singular target it holds on each admissible interval with $T<1$; an endpoint limit requires control of the $T$-dependent constants and density-ratio moments. Related endpoint estimates for flow matching are given in refs.~\cite{benton2024errorbounds,zhou2025erroranalysis}.

\section{Exact low-dimensional constructions and replication}
\subsection{Rare12 and its polar mass partition}

The eight common centres lie on the circle of radius 2.5 and the four rare centres on the circle of radius 6 before whitening,
\begin{align}
 \widetilde c_j&=2.5(\cos(2\pi j/8),\sin(2\pi j/8)),&&j=0,\ldots,7,\\
 \widetilde c_{8+j}&=6(\cos(\pi/4+2\pi j/4),\sin(\pi/4+2\pi j/4)),&&j=0,\ldots,3.
\end{align}
The weights are $w_j=0.115$ for common modes and $0.02$ for rare modes. Opposite centres have equal weights, so the weighted mean is zero. Rotational symmetry gives weighted covariance
\begin{equation}
 \widetilde\Sigma
 =\left(0.92\frac{2.5^2}{2}+0.08\frac{6^2}{2}\right)I_2
 =4.315I_2.
\end{equation}
The implemented whitening is $c_j=\widetilde\Sigma^{-1/2}\widetilde c_j$, and $P=\sum_jw_j\delta_{c_j}$ has mean zero and covariance $I_2$.

For $Z\sim\N(0,I_2)$, the radius $R=\|Z\|$ and polar angle $\Theta\in[0,2\pi)$ have joint density $f_{R,\Theta}(r,\theta)=r\exp(-r^2/2)/(2\pi)$ for $r\geq0$. Factorization gives independence of $R$ and $\Theta$, with $\Theta$ uniform and
\begin{equation}
 \Pr(R>r)=\int_r^\infty s e^{-s^2/2}\,ds
 =\left[-e^{-s^2/2}\right]_r^\infty=e^{-r^2/2}.
\end{equation}
Starting at angle $\theta_0=-0.115\pi$, form eight consecutive sectors with probability sequence $(0.115,0.135,0.115,0.135,0.115,0.135,0.115,0.135)$. Each sector's angular width is $2\pi$ times its probability. In each larger sector, set the radial cutoff
\begin{equation}
 r_\star=\sqrt{-2\log(0.02/0.135)}\approx1.9542.
\end{equation}
Independence implies that its outer portion has mass $0.135e^{-r_\star^2/2}=0.02$ and the inner portion has mass $0.135-0.02=0.115$. Each unsplit sector already has mass 0.115. Mapping the eight common cells and four outer cells to the corresponding centres gives a measurable map $T$ with
\begin{equation}
 Q(T^{-1}(\{c_j\}))=w_j\quad\hbox{for every }j,
 \qquad T_\#Q=P.
\end{equation}
Here $T_\#Q=\Law(T(Z))$ for $Z\sim Q$ is the pushforward of $Q$ by $T$. The gaussian8 target retains eight equally weighted inner centres, standardized to identity covariance, and uses eight equal angular source sectors.

\subsection{Conditional fields and the score readout}

For endpoint $c_j$ at $t<1$, write $z_j=(x-tc_j)/(1-t)$ and let $\phi$ be the standard Gaussian density in two dimensions. The conditional assignment probabilities are
\begin{equation}
 \gamma_j(z)=
 \begin{cases}
   w_j,&\text{independent coupling},\\
   \mathbf1\{z\in A_j\},&\text{polar coupling},
 \end{cases}
\end{equation}
where $A_j$ is the polar cell assigned to $c_j$. For a bounded measurable test function $f$, the joint law of the interpolation state and endpoint obeys
\begin{equation}
\begin{aligned}
 \E[f(X_t)\mathbf1\{X=c_j\}]
 &=\int f((1-t)z+tc_j)\gamma_j(z)\phi(z)\,dz\\
 &=\int f(x)\frac{\gamma_j((x-tc_j)/(1-t))
                         \phi((x-tc_j)/(1-t))}{(1-t)^2}\,dx.
\end{aligned}
\end{equation}
The second line uses the Jacobian determinant $(1-t)^2$. It identifies the density contribution $\rho_{j,t}$ for endpoint $c_j$; summing gives interpolation density $p_t$, and Bayes' rule gives posterior endpoint probabilities $\alpha_j$,
\begin{equation}
\begin{aligned}
 \rho_{j,t}(x)&=\frac{\gamma_j(z_j)\phi(z_j)}{(1-t)^2},\qquad
 p_t(x)=\sum_j\rho_{j,t}(x),\\
 \alpha_j(x,t)&=\Pr(X=c_j\mid X_t=x)=\frac{\rho_{j,t}(x)}{p_t(x)}.
\end{aligned}
\label{eq:exact-density}
\end{equation}
These contributions integrate to $\int\rho_{j,t}(x)\,dx=\int\gamma_j(z)\phi(z)\,dz=w_j$. For $p_t(x)>0$, conditioning on the endpoint and using $c_j-z_j=(c_j-x)/(1-t)$ gives
\begin{equation}
 u_t(x)=\sum_j\alpha_j(x,t)\frac{c_j-x}{1-t},\qquad
 \E[Z\mid X_t=x]=\sum_j\alpha_j(x,t)z_j.
\end{equation}
On each interior where the active-cell indicators are constant, $\nabla_x z_j=I_2/(1-t)$ and $\nabla_z\phi(z)=-z\phi(z)$. Hence
\begin{equation}
\begin{aligned}
 \nabla_x\rho_{j,t}(x)&=-\frac{z_j}{1-t}\rho_{j,t}(x),\\
 (1-t)\nabla_x\log p_t(x)
 &=\frac{1-t}{p_t(x)}\sum_j\nabla_x\rho_{j,t}(x)\\
 &=-\sum_j\alpha_j(x,t)z_j
 =-\E[Z\mid X_t=x]\\
 &=-x+t\E[U\mid X_t=x]=-x+tu_t(x).
\end{aligned}
\label{eq:score-chain}
\end{equation}
The last line uses $Z=X_t-tU$ and defines the reference scaled score $s_t(x)=(1-t)\nabla_x\log p_t(x)=-x+tu_t(x)$ on these interiors. The velocity network supplies the scaled-score readout $\widehat s_t(x)=-x+t\widehat u_t(x)$, whose pointwise error is
\begin{equation}
 \widehat s_t(x)-s_t(x)=t\bigl(\widehat u_t(x)-u_t(x)\bigr).
\end{equation}
Relative score and velocity errors are normalized by their respective reference fields. For reference field $f=u$ or $s$ and its learned counterpart $\widehat f$, sample $n$ state--time pairs $(X_{t_r},t_r)$ from the method's interpolation law. The relative-error estimator is
\begin{equation}
 E_f=\frac{\sum_{r=1}^{n}\|\widehat f_{t_r}(X_{t_r})-f_{t_r}(X_{t_r})\|^2}
              {\sum_{r=1}^{n}\|f_{t_r}(X_{t_r})\|^2}.
\end{equation}

The common- and rare-mode densities in main \textcolor{refblue}{Fig.~2c,d} sum equation~\ref{eq:exact-density} over $j=0,\ldots,7$ and $j=8,\ldots,11$, respectively, with a shared intensity scale at each time.

\subsection{Training units, inference budgets and statistics}
\label{sec:controlled-training}

Each target uses 32 paired training runs of the two methods, with 9,218 parameters, three hidden layers of width 64 and SiLU activations. AdamW uses 3,500 updates, batch size 512, initial learning rate 0.002 and a cosine schedule ending at 0.0001. Training and field evaluation stratify time over $[0.001,0.999]$. Relative field errors use 12,000 interpolation states; endpoint metrics use 3,000 generated and 3,000 target samples with 128 common projection directions. Euler integration starts from a Gaussian draw at time 0.001 and ends at 0.999, using NFE 1, 2, 4, 8, 16 and 32. Integration MSE uses a 256-evaluation Euler reference from the same initial draw. Mode-resolved interpolation uses the same 2,048 source draws across methods and times; endpoint frequencies use 40,000 pairs.

The endpoint maps in main \textcolor{refblue}{Fig.~3a,b} pool 3,000 generated points from each of three paired training runs at each displayed budget. A Gaussian kernel with bandwidth equal to 4.5\% of the minimum target-mode separation smooths both generated points and the weighted target point masses. Each target uses a common display window and logarithmic colour scale across methods and budgets. The accompanying SWD curves summarize all 32 paired runs.

For point sets $A$ and $B$, each containing $n$ observations, and $L$ projection directions, let $a_{(i),\ell}$ and $b_{(i),\ell}$ be their sorted projections onto the $\ell$th fixed random unit direction. The low-dimensional benchmark uses the root mean squared sliced discrepancy~\cite{bonneel2015sliced}
\begin{equation}
 \mathrm{SWD}_2(A,B)=\left[\frac1{Ln}\sum_{\ell=1}^L\sum_{i=1}^n
       (a_{(i),\ell}-b_{(i),\ell})^2\right]^{1/2}.
\end{equation}
Each comparison uses the same sampled projection directions for both methods. For $n$ generated points $y_r$, $r=1,\ldots,n$, assign the nearest mode $\widehat j_r=\arg\min_j\|y_r-c_j\|$ and estimate mode probabilities by $\widehat w_j=n^{-1}\sum_r\mathbf1\{\widehat j_r=j\}$. Mode-mass error is
\begin{equation}
 \mathrm{ModeTV}=\tfrac12\sum_j|\widehat w_j-w_j|.
\end{equation}

For $n$ paired training runs, let $d_i$ be the candidate-minus-comparator metric difference in pair $i$ and let $s=(s_1,\ldots,s_n)$ enumerate signs $s_i\in\{-1,1\}$. Under the paired exchangeability null, the exact two-sided sign-flip probability is
\begin{equation}
 P_{\rm exact}=2^{-n}\sum_{s\in\{-1,1\}^n}
 \mathbf1\left\{\left|n^{-1}\sum_i s_i d_i\right|
             \geq\left|n^{-1}\sum_i d_i\right|\right\}.
\end{equation}
The exact test includes all $2^{32}$ sign assignments. Intervals use 50,000 bootstrap resamples of complete training-run pairs~\cite{efron1994introduction}.

For a discrete path $x_0,\ldots,x_k$ with $k$ steps and increments $v_j=x_j-x_{j-1}$, accumulated turning angle is
\begin{equation}
 \kappa=\sum_{j=1}^{k-1}\arccos\!\left(
 \frac{v_j^\top v_{j+1}}{\|v_j\|\|v_{j+1}\|}\right),
\end{equation}
omitting negligible increments. Solver MSE compares Euler and reference trajectories from common initial coordinates.

\section{Image experiments and primary statistics}
\subsection{Data laws and evaluator calibration}

The image experiments share a stratified CIFAR-10 train--validation split with 45,000 and 5,000 images, respectively. The official 10,000-image test set is used for final endpoints. For the WLT100 task, target class probabilities are $p_c=100^{-c/9}/\sum_{\ell=0}^9 100^{-\ell/9}$ for $c=0,\ldots,9$. Image $i$ has class label $y_i$ and weight $w_i=p_{y_i}/n_{y_i}$, where $n_c$ counts training images of class $c$. Every training image remains in the support. The ratio between the most and least probable classes is 100:1; classes 7, 8 and 9 define the minority set.

The native-resolution ResNet-18 classifier~\cite{he2016deep} has validation accuracy 0.9382, minimum class recall 0.874 and a confusion matrix with condition number approximately 1.203. For columns indexed by true class, $C_{ab}$ is the fraction of validation images of true class $b$ assigned to class $a$ by the maximum-logit prediction. The same prediction convention is used for the generated marginal $q$. With $p$ denoting generated true-class probabilities, the observation model is $q=Cp$. The corrected estimate solves
\begin{equation}
 \widehat p=\arg\min_{p\in\R^{10}}\|Cp-q\|_2^2
 \quad\text{subject to}\quad p_c\geq0,\ \sum_cp_c=1,
 \qquad \mathrm{MassTV}=\tfrac12\|\widehat p-p_{\rm target}\|_1.
\end{equation}
Here $p_{\rm target}$ contains the prescribed target class probabilities; $\|\cdot\|_1$ sums absolute coordinate values and $\|\cdot\|_2$ is the Euclidean norm. The observation model follows class-probability estimation under label shift~\cite{lipton2018detecting}. Across 10,000 synthetic multinomial replicates with 50,000 draws each, the maximum absolute estimated class-mass bias over the tested imbalance factors is $4.07\times10^{-5}$. At WLT100, calibration TV has mean 0.00469 and 95th percentile 0.00712.

Weighted FID uses Inception-v3 features from \texttt{torch-fidelity} 0.4.0, float64 moments and exact class weights for real samples. KID estimates squared maximum mean discrepancy with a cubic polynomial kernel~\cite{gretton2012kernel,binkowski2018demystifying}, using 20 subsets of 1,000 observations. Overall precision and recall~\cite{kynkaanniemi2019improved} use 10,000 samples, three-neighbour radii and balanced real support. Class-specific measures use the evaluator\textquotesingle s features, with up to 1,000 real and 1,000 generated observations per class. The classifier assigns generated classes; minority recall averages classes 7--9 equally for both WLT100 and balanced targets. Class-frequency plots use raw predicted-to-target ratios; MassTV uses confusion-corrected estimates.

\subsection{Paired image comparisons}

We compare six methods on both image tasks using eight paired training runs for WLT100 and three for balanced transfer. Each run contributes 50,000 generated samples at Euler NFE 16. Inference on training variability treats the trained pair as the independent unit; samples within each run are used to estimate its metrics.

With subscripts $B$ and $G$ denoting a comparator and G-MPCF, orient a lower-is-better outcome $M$ as $M_B-M_G$ and recall $R$ as $R_G-R_B$. Overbars denote means across training runs. A non-inferiority test adds its margin to this oriented difference. The FID margin is $\max(1,0.05\overline{\mathrm{FID}}_B)$, giving 1.11738 against I-CFM and 1.10296 against OT-CFM. The MassTV and recall margins are 0.01 and 0.02. All $2^8$ signs are enumerated for each one-sided test. Holm adjustment is performed within the four rare-event superiority components, the two FID non-inferiority components and the two rare-event non-inferiority components as three separate families. The analysis uses 100,000 paired bootstrap resamples.

A component passes when its Holm-adjusted $P$ value is below 0.05 and its two-sided 95\% interval lies above zero. G-MPCF met the MassTV-superiority and FID-non-inferiority criteria against both controls. Recall effects spanned zero, so the joint allocation--recall criteria were not met.

\subsection{Shared model and comparator implementations}

The image methods share an unconditional residual U-Net and optimizer; \suppref{Tables~S6 and S7} give the training schedules and model specification. The flow loss averages squared velocity errors over batch samples, stratified times and image coordinates. Evaluation uses exponential-moving-average weights. Numerical trajectories use uniform time steps from $t=0$ to $t=1$ and evolve without clipping; metric images are clipped to $[-1,1]$ and encoded by rounding $127.5(x+1)$ to unsigned integers, where $x$ is a scaled pixel-channel value. Method-specific update counts match the computation envelope within 1\%; rectified flow uses 50,000 first-stage and 7,893 second-stage updates.

OT-CFM uses the exact squared-Euclidean minibatch assignment, corresponding to the BatchOT instance of multisample flow matching~\cite{tong2024improving,pooladian2023multisample}. Two-stage rectified flow charges the teacher-generated pairs and both training stages. DGSWP uses the official learned-plan module, reusing each learned pool for ten flow batches~\cite{chapel2025differentiable}. MAC adapts the official full endpoint-prediction cost and solves its assignment exactly~\cite{lin2026modelaligned}. This adaptation uses the full cost in place of the source implementation's Sinkhorn warm-up and plan extraction. The six methods share the model, weighted data stream, optimizer and evaluator, with method-specific additional evaluations included in charged computation.

\subsection{Coupling implementation, numerical precision and cost}
\label{sec:coupling-implementation}

The fixed feature extractor is an ImageNet-pretrained ResNet-18~\cite{deng2009imagenet,he2016deep}. Images are resized to 224 pixels using antialiased bilinear interpolation and normalized with the extractor's ImageNet channel statistics. The 512-dimensional penultimate features are averaged over each image and its horizontal reflection, then standardized using target-weighted training means and standard deviations. Weighted $k$-means~\cite{lloyd1982least} with ten initializations fits 20 non-empty regions. A sign-anchored principal axis of their centroids determines the region order.

The eight leading weighted principal-component directions of the standardized features define the target coordinates. Within each region and direction, parentheses around an index denote its rank in the sorted target coordinates. Stable sorting gives conditional weights $\pi_{(i)}=w_{(i)}/p_j$ and midpoint ranks $q_{(i)}=\sum_{\ell<i}\pi_{(\ell)}+\pi_{(i)}/2$. Clipping these ranks to $[10^{-6},1-10^{-6}]$ and applying $\Phi^{-1}$ gives empirical normal scores. A fixed nine-column orthonormal source frame provides one regional coordinate and eight geometry coordinates. The regional coordinate remains fixed, while the active geometry coordinate cycles through the eight directions. Conditional sampling uses replacement, and stable sorting pairs the sampled target scores with source projections. After matching, independent augmentation selects the original or horizontally flipped image with equal probability.

For an exactly Gaussian source, a stored projection vector with norm $\sigma>0$ gives projection law $\N(0,\sigma^2)$. The regional probability accounting for this norm is therefore
\begin{equation}
 p_j^{\rm impl}=\Phi(b_j/\sigma)-\Phi(b_{j-1}/\sigma).
\end{equation}
Numerical checks evaluate the ideal quantile residual, the finite-precision residual and the reconstructed target-row residual separately. Equation~\ref{eq:multiset} supplies the corresponding finite-batch count identity.

For $N$ observations in $D$ image coordinates, $m$ regions, feature dimension $F$, geometry rank $r$ and $T_k$ clustering iterations, the asymptotic fitting cost is $O(NF)+O(T_kNmF)+O(NFr)+O(rN\log N)$. Online projection costs $O(BD(r+1))$, and regional sorting costs $O(B\log B)$ since
\begin{equation}
 \sum_{j:n_j>0}n_j\log n_j
 \leq\sum_{j:n_j>0}n_j\log B=B\log B.
\end{equation}
The coupling state stores $O(Nr+N+D(r+1)+(m+r+2)F)$ scalars in addition to the feature extractor.

\subsection{Class frequencies, image samples and feature distributions}

Each feature map contains 1,536 images, sampled uniformly with equal contributions from each training run. Reference images are sampled from the test set in proportion to the target class probabilities.

Principal-component analysis of the classifier's 512-dimensional features, fitted to all 10,000 reference images, retains 50 components and 99.41\% of their variance. The target, I-CFM and G-MPCF comparison in Supplementary Fig.~\ref{sfig:sample_distributions} uses one joint t-SNE~\cite{vandermaaten2008visualizing}. The six-method maps in main \textcolor{refblue}{Fig.~3c} and Supplementary Fig.~\ref{sfig:all_methods_long_tail} use a separate joint fit. Each fit combines both target laws and shares coordinate limits across its maps. Both use Euclidean distances, perplexity 40, principal-component initialization and 1,500 Barnes--Hut iterations.

Enlargement bounds use the target class's marginal 2nd and 98th percentiles with an 8\% margin. The car enlargements retain all classes inside these bounds, with 373 target observations, 266 under I-CFM and 346 under G-MPCF. Across the full long-tail evaluations, the predicted car frequency increases from 16.45\% under I-CFM to 20.81\% under G-MPCF, towards the 24.16\% target.

The balanced truck enlargements use $2.2\times$ magnification; all classes remain in the full maps. Their dashed contours enclose 80\% of the target-truck density, estimated with an isotropic Gaussian kernel whose bandwidth is the median distance to the sixth nearest target-truck neighbour. Inset bars report predicted truck frequencies averaged over the three complete evaluations, with a common scale from zero to the 10\% target.

For main \textcolor{refblue}{Fig.~3d}, the common centre is the mean of the balanced reference map. Subtracting this centre and normalizing each displacement gives the unit-circle coordinates. The ring colours use a two-dimensional Gaussian kernel with bandwidth 0.075. Angular profiles use a wrapped Gaussian kernel with bandwidth 0.20 radians. All methods share the colour scale, density limits and dashed reference profile.

The image sequences in Supplementary Fig.~\ref{sfig:budget_samples} use the three paired training runs of the image source-selection experiment, evaluated with a standard Gaussian source and full displacement.

\Needspace{8\baselineskip}
\subsection{Image field and trajectory diagnostics}

Field-loss and shift measurements use the definitions in Section~\ref{sec:diagnostics}. For a numerical image trajectory $x_0,\ldots,x_k$ with increments $v_j=x_j-x_{j-1}$, dimensionless roughness is $\sum_j\|v_{j+1}-v_j\|/\sum_j\|v_j\|$, averaged over samples. Path length is $\sum_j\|v_j\|$; dividing by endpoint displacement gives the path-length ratio. Relative Euler-to-reference MSE divides the mean squared coordinate discrepancy by the mean squared reference-endpoint coordinate.

\section{Component interventions}

The component experiment compares five coupling operations and three interior mass-relaxation levels. Independent pairing samples source and target independently; geometry-only pairing uses global rank matching. Mass-only pairing uses calibrated regional assignment and conditional target sampling, and G-MPCF adds the rank permutation. Uncalibrated geometry assigns equal source probability to the 20 fitted regions while retaining their weighted conditional target laws. Its population row-law discrepancy follows from equation~\ref{eq:relax} with $\rho=1$.

The relaxation series adds $\rho=0.25,0.5,0.75$ to the G-MPCF and uncalibrated endpoints. Three independent training runs for each of the eight distinct settings give 24 fitted fields; the five-point relaxation series reuses the six endpoint runs. Across increasing $\rho$, minority recall was 0.565, 0.713, 0.704, 0.727 and 0.664. Geometry-only coupling gave recall 0.714 and MassTV 0.10020. The component runs used independent initializations from the primary experiment.

\Needspace{135pt}

\section{Transfer conditions and image scaling}
\subsection{Experimental design and paired diagnostics}
\label{sec:diagnostics}

The boundary experiment uses source $Q=\N(\delta e_2,I_d)$ and target
\begin{equation}
 P=p\N(-(1-p)se_1,0.35^2I_d)+(1-p)\N(pse_1,0.35^2I_d).
\end{equation}
Here $d$ is dimension, $p$ is rare-mode probability, $s$ is mode separation and $\delta$ is source displacement; $e_1,e_2$ are the first two standard basis vectors and $I_d$ is the identity matrix. Its candidate levels are $d\in\{8,32,128\}$, $p\in\{0.01,0.02,0.05,0.1,0.2\}$, $s\in\{1,2,4,6,8\}$ and $\delta\in\{0,1,2,4\}$. A deterministic maximin procedure, including the anchor conditions, selects 36 of the 300 candidates. The target mean is zero because $p[-(1-p)s]+(1-p)[ps]=0$. The source displacement is orthogonal to the direction separating the target modes. The single-factor slices in main \textcolor{refblue}{Fig.~5a--c} fix $d=32$ and vary one coordinate around $p=0.02$, $s=6$ and $\delta=2$.

Each method uses the same realized component capacities at batch size 200 and shares source and conditional target-noise streams. Independent assignment randomizes the source-to-component relation; the structured assignment uses the first source coordinate to sort this relation. The network is a width-64 residual MLP with $129d+9,856$ parameters, trained for 5,000 updates at learning rate 0.001. There are three independent training runs per condition. Each trained field supplies 5,000 samples at Euler NFE 1, 2, 4, 8 and 16. The reference sampler uses 128 field evaluations, or 64 Heun steps. The boundary endpoint and image feature-shift diagnostics use mean absolute sliced discrepancy,
\begin{equation}
 \mathrm{SWD}_1(A,B)=\frac1{Ln}\sum_{\ell=1}^L\sum_{i=1}^n
       |a_{(i),\ell}-b_{(i),\ell}|.
\end{equation}
Here $n$ is the number of points in each of $A$ and $B$, $L$ is the number of directions, and $a_{(i),\ell},b_{(i),\ell}$ are the sorted projections. The boundary calculation uses $L=128$ sampled unit directions, shared across methods. The low-dimensional replication uses $\mathrm{SWD}_2$ as defined in Section~\ref{sec:controlled-training}.

Own-law and common-law losses compare the predicted field with sampled pair displacements, normalized by mean squared displacement. The common-law evaluation uses the same interpolation pairs across methods; own-law evaluation uses each method's training coupling. Interpolation and rollout losses retain the pair-displacement labels and evaluate the field along the prescribed interpolation or the learned trajectory, respectively. Their difference is the rollout penalty. An eight-neighbour velocity-variance estimate uses up to 1,024 states and at most 16 normalized random projections. State SWD measures the discrepancy between interpolation and rollout distributions.

Denote own-law, interpolation and rollout regression losses by $E_{\rm own}$, $E_{\rm interp}$ and $E_{\rm roll}$, respectively, and Euler-to-Heun endpoint MSE by $S$. For structured coupling G and independent coupling I, define
\begin{align}
 g_{\rm field}&=1-\frac{E_{\rm own,G}}{E_{\rm own,I}},\\
 g_{\rm roll}&=(E_{\rm roll,I}-E_{\rm interp,I})
                  -(E_{\rm roll,G}-E_{\rm interp,G}),\\
 g_{\rm solver}&=S_I-S_G,\\
 g_{\rm end}&=\mathrm{SWD}_I-\mathrm{SWD}_G.
 \label{eq:diagnostic-gains}
\end{align}
All four coordinates are oriented so that positive values indicate improvement under structured coupling. Field gain is a relative reduction, expressed as a fraction in the tables and as a percentage in the main figure. Rollout, solver and endpoint reductions are absolute differences in their respective metrics. The single-factor endpoint curves use $100(1-\mathrm{SWD}_G/\mathrm{SWD}_I)$ for each paired run before averaging. The symmetric logarithmic axes in main \textcolor{refblue}{Fig.~5d--f} have central linear intervals of $\pm0.02$ percentage points, $\pm5\times10^{-4}$ and $\pm2\times10^{-6}$, respectively.

The descriptive regression uses the 36 condition means, with each predictor centered and divided by its sample standard deviation across conditions. Coefficients for field gain, rollout reduction and solver reduction are approximately $6.34\times10^{-4}$, $3.87\times10^{-3}$ and $1.42\times10^{-3}$, with intercept $2.36\times10^{-3}$ and fitted coefficient of determination $R^2=0.8884$.

\subsection{Data, model size and computation}

The image scaling design varies one factor around the full-data, base-width-96, full-compute reference. Data fractions are 0.125, 0.25, 0.5 and 1, corresponding to 5,625, 11,250, 22,500 and 45,000 training images. Base widths are 64, 96 and 128, with 8,952,067, 20,118,915 and 35,746,307 trainable parameters. Charged-compute multipliers are 0.25, 0.5 and 1. Removing duplicate configurations leaves eight settings. Two methods and three independent training runs per configuration give 48 fields. Five NFE values give 240 evaluations, each with 10,000 generated images. At NFE 16, minority recall improved in 13 of 24 paired training runs.

\section{Validation-based source and action selection}
\subsection{Candidate grid and decision rule}

The selector compares a standard Gaussian with a directly sampleable mixture of 20 probabilistic principal-component components~\cite{tipping1999mixtures}, fitted on the training split. Components follow the fitted target regions and have probabilities $p_j$, rank 64 in 3,072-dimensional pixel space and covariance floor $10^{-5}$. For G-MPCF with this source, the sampled component index supplies the source-region label for conditional target sampling; the within-region rank permutation then uses the source projections. Each source is crossed with independent, OT-CFM and G-MPCF coupling. Each of these six source--coupling combinations is trained three times independently, giving 18 fields.

For a velocity network with fitted parameters $\theta$, let $T_{\theta,k}(z)$ be the full numerical endpoint from input $z$ using $k$ field evaluations. Retaining fraction $\lambda$ of its displacement defines the residual action
\begin{equation}
 A_{\lambda,k}(z)=z+\lambda[T_{\theta,k}(z)-z],\qquad
 \lambda\in\{0.25,0.5,0.75,1\},\quad k\in\{1,2,4,8,16\}.
\end{equation}
Adding source-retention actions $A_0(z)=z$ gives $2(1+3\cdot4\cdot5)=122$ candidates. The residual fraction scales the final displacement after integration at the specified NFE.

Selection uses mean validation KID and corrected MassTV across three training runs. Actions qualify on mass when their mean MassTV is at most 0.01 above the best source-retention mean, 0.7929. This threshold of 0.8029 admits 76 of the 122 actions. Among them, the KID threshold is the minimum mean plus its standard error. The rule chooses the smallest NFE among qualifying actions, followed by fixed source, coupling and fraction tie breaks. The minimum validation KID is 0.0131 with standard error $1.23\times10^{-4}$, giving a threshold of 0.0132. A single action qualifies, using a standard Gaussian source, G-MPCF, $\lambda=1$ and NFE 16.

\subsection{Test-set performance across candidate actions}

The selected action has test FID $21.775\pm0.092$, corrected MassTV $0.05489\pm0.00306$ and minority recall $0.627\pm0.099$. The Gaussian independent full-flow control gives $24.359\pm0.108$, $0.09765\pm0.00791$ and $0.671\pm0.184$, respectively; the matched OT-CFM control has FID $23.874\pm0.248$ and MassTV $0.09197\pm0.00082$. Retaining the fitted PPCA source gives FID approximately 265.65. Values are means and sample standard deviations across three training runs, each supplying 10,000 images per candidate. Validation and test class-mass estimates use the same generated predictions and target law, while FID and KID use the appropriate real split.

\section{Source adequacy and additional coupling comparisons}
\subsection{A finite-sample Gaussian shift calculation}

Consider $n$ independent training pairs $Z_i\sim\N(0,I_d)$ and $X_i\sim\N(\mu,I_d)$ in dimension $d$, with target mean $\mu\in\R^d$. Fit a constant velocity $v\in\R^d$ by least squares. The objective and its derivative are
\begin{equation}
 L(v)=\frac1n\sum_{i=1}^n\|X_i-Z_i-v\|^2,
 \qquad \nabla_vL(v)=2v-\frac2n\sum_i(X_i-Z_i).
\end{equation}
Setting the derivative to zero gives $\widehat\mu=n^{-1}\sum_i(X_i-Z_i)$. Independence implies $X_i-Z_i\sim\N(\mu,2I_d)$ and hence
\begin{equation}
 \widehat\mu\sim\N(\mu,2I_d/n),\quad
 \E\widehat\mu=\mu,\quad
 \E\|\widehat\mu-\mu\|^2=\operatorname{tr}(2I_d/n)=2d/n.
\end{equation}
Integrating the constant field maps a fresh source draw to $Z+\widehat\mu$, with conditional law $\N(\widehat\mu,I_d)$. For any coupling $(A,B)$ of $\N(a,I_d)$ and $\N(b,I_d)$, Jensen's inequality gives
\begin{equation}
 \E\|A-B\|^2\geq\|\E A-\E B\|^2=\|a-b\|^2.
\end{equation}
The coupling $A=Z+a$, $B=Z+b$ attains the lower bound, so
\begin{equation}
\begin{aligned}
 W_2^2(\N(a,I_d),\N(b,I_d))
 &\leq\E\|(Z+a)-(Z+b)\|^2=\|a-b\|^2,\\
 W_2^2(\N(a,I_d),\N(b,I_d))&=\|a-b\|^2.
\end{aligned}
\end{equation}
Applying this identity conditionally on the training sample gives
\begin{equation}
\begin{aligned}
 \E W_2^2(\N(\widehat\mu,I_d),\N(\mu,I_d))
 &=\E\|\widehat\mu-\mu\|^2=2d/n,\\
 W_2^2(\N(0,I_d),\N(\mu,I_d))&=\|\mu\|^2.
\end{aligned}
\end{equation}
Source retention has lower expected squared endpoint error when $\|\mu\|^2<2d/n$.

For residual fraction $\lambda$, the fitted mean is $\lambda\widehat\mu$. Write $\widehat\mu=\mu+\varepsilon$, where $\varepsilon$ is the mean-estimation error and $\E\varepsilon=0$. The expected squared endpoint error $R(\lambda)$ expands as
\begin{align}
 R(\lambda)
 &=\E\|\lambda\widehat\mu-\mu\|^2
 =\E\|(\lambda-1)\mu+\lambda\varepsilon\|^2\\
 &=(1-\lambda)^2\|\mu\|^2+\lambda^2\E\|\varepsilon\|^2
     +2\lambda(\lambda-1)\mu^\top\E\varepsilon\\
 &=(1-\lambda)^2\|\mu\|^2+\lambda^2(2d/n).
\end{align}
With $a=\|\mu\|^2$ and $b=2d/n$, completing the square gives
\begin{equation}
\begin{aligned}
 R(\lambda)
 &=a-2a\lambda+(a+b)\lambda^2\\
 &=(a+b)\left(\lambda-\frac{a}{a+b}\right)^2+\frac{ab}{a+b},\\
 \lambda_\star&=\frac{a}{a+b}\in[0,1],\qquad
 R(\lambda_\star)=\frac{ab}{a+b}.
\end{aligned}
\end{equation}
The fraction $\lambda_\star$ minimizes $R$. Euler integration over $k$ steps with $h=1/k$ and $V_0=Z$ is exact for this constant field,
\begin{equation}
 V_k=V_0+\sum_{j=0}^{k-1}h\widehat\mu
     =Z+kh\widehat\mu=Z+\widehat\mu.
\end{equation}

\subsection{Residual actions for a nonlinear map}

Let $T$ be a transport map satisfying $T_\#Q=P$, and let $\widehat T$ and $\widehat T_k$ be the learned continuous endpoint map and its $k$-step numerical approximation. With a common $Z\sim Q$, define
\begin{equation}
 C=\|Z-T(Z)\|_{L^2},\qquad
 E_{\rm fit}=\|\widehat T(Z)-T(Z)\|_{L^2},\qquad
 E_{\rm num}=\|\widehat T_k(Z)-\widehat T(Z)\|_{L^2}.
\end{equation}
Thus $C$ measures transport displacement, $E_{\rm fit}$ fitting error and $E_{\rm num}$ numerical error in root mean square. For $A_{\lambda,k}(Z)=(1-\lambda)Z+\lambda\widehat T_k(Z)$, adding and subtracting $\lambda T(Z)$ and $\lambda\widehat T(Z)$ gives
\begin{equation}
\begin{aligned}
 A_{\lambda,k}(Z)-T(Z)
 &=(1-\lambda)(Z-T(Z))\\
 &\quad+\lambda(\widehat T(Z)-T(Z))
       +\lambda(\widehat T_k(Z)-\widehat T(Z)).
\end{aligned}
\end{equation}
For $0\leq\lambda\leq1$, this coupling and Minkowski's inequality yield
\begin{equation}
\begin{aligned}
 W_2(\Law(A_{\lambda,k}(Z)),P)
 &\leq\|A_{\lambda,k}(Z)-T(Z)\|_{L^2}\\
 &\leq(1-\lambda)\|Z-T(Z)\|_{L^2}
      +\lambda\|\widehat T(Z)-T(Z)\|_{L^2}\\
 &\quad+\lambda\|\widehat T_k(Z)-\widehat T(Z)\|_{L^2}\\
 &=(1-\lambda)C+\lambda E_{\rm fit}+\lambda E_{\rm num}.
\end{aligned}
\end{equation}

\subsection{Source adequacy in feature space}

The source-adequacy experiment evaluates residual actions in 64-dimensional representations of Kuzushiji-MNIST (KMNIST)~\cite{clanuwat2018classical}, CIFAR-10 and UCI Human Activity Recognition (UCI HAR)~\cite{reyes2013har}, with 24 independent training runs per dataset and 72 in total. KMNIST uses 50,000 training, 10,000 validation and 10,000 test observations; CIFAR-10 uses 45,000, 5,000 and 10,000. Subject-disjoint UCI HAR partitions contain 6,063, 1,289 and 2,947 observations. Principal components and a single global variance scale are fitted on the training partition, using flattened pixels for the image datasets and the 561 supplied sensor features for UCI HAR.

A 32-component diagonal Gaussian mixture is fitted by hard $k$-means assignments, with component probabilities given by training frequencies and variance floor $10^{-6}$. Source and target observations share the fitted component label during flow training. The velocity network has three width-256 hidden layers with SiLU activation and uses AdamW, learning rate $5\times10^{-4}$, cosine scheduling and batch size 128. Each method receives $1.8\times10^{11}$ charged FLOPs, including source fitting and pairing. Euler trajectories span $[0.001,0.999]$ from a fresh source draw.

Each run compares source retention with 20 residual actions crossing $\lambda\in\{0.25,0.5,0.75,1\}$ and NFE $1,2,4,8,16$. Validation uses up to 3,000 observations and 512 common projection directions, grouped into 16 blocks of 32. An action qualifies when its mean SWD improvement over source retention is at least 1\% and its one-sided lower confidence bound is positive. The bound uses a $t$ quantile with 15 degrees of freedom and Bonferroni level $0.05/20$ on the paired block improvements. Among actions within 1\% of the best qualifying mean SWD, selection favours the lowest NFE, then lower SWD and smaller $\lambda$; if none qualifies, it retains the source. Test SWD uses the root-mean-square projected discrepancy in Section~\ref{sec:controlled-training}, with up to 3,000 samples and 128 directions. The selected action and full flow share the fitted field.

The comparators include Gaussian-source independent and minibatch-OT flows, mixture-source independent and component-matched flows, and a one-cycle cluster-transport adaptation~\cite{chiang2026cotfm}. The last method allocates half of the available flow computation to Gaussian pretraining, reverse-integrates 8,192 cluster-stratified training observations for eight steps, fits full-covariance component Gaussians with eigenvalue floor $10^{-5}$, and fine-tunes with within-component minibatch OT. All stages enter the same computation budget.

Mean selected SWDs are 0.12004, 0.05422 and 0.10859 for KMNIST, CIFAR-10 features and UCI HAR; source-retention values are 0.12177, 0.05420 and 0.10855. For each dataset, we divide mean full-minus-selected SWD by mean source-retention SWD, then average equally across the three datasets. This normalized improvement is 0.07633, with stratified bootstrap interval $[0.06744,0.08532]$. Mean selected NFE is 3.556, a 77.8\% reduction from 16 for full flow. Dataset-specific means are 6.833, 3.333 and 0.5. Mean full-minus-selected latency is 2.77 milliseconds, with stratified interval $[2.53,2.98]$ milliseconds.

\subsection{Additional coupling comparisons}

The rare12 mechanism study uses five independent training runs per method to compare independent joint velocity/score training, exact sector coupling, a fitted Laguerre assignment and exact polar coupling. Sector coupling assigns angular widths $2\pi w_j$ to target modes; Laguerre cells assign $z$ to the maximizer of $z^\top c_j+h_j$, with offsets $h_j$ fitted to target probabilities. The exact polar construction is given in \textcolor{refblue}{Section~3.1}. All use 3,500 updates and batch size 512. The joint-output baseline has 9,348 parameters; the other three fields have 9,218. The 32-pair comparison matches the two-output network and parameter count. Nearest-mode distance averages each generated point's distance to its closest target centre; accumulated turning angle measures changes in successive trajectory directions.

A six-method comparison uses five independent training runs per method and target for I-CFM, OT-CFM, two-stage rectified flow, Optimal Flow Matching, MeanFlow and MPCF~\cite{kornilov2024optimal,geng2025meanflow}. Each method receives 3,500 optimizer updates with batch size 512. I-CFM, OT-CFM, 2-RF and MPCF use the 9,218-parameter network in \textcolor{refblue}{Section~3.3}; MeanFlow adds interval conditioning, giving 9,282 parameters. OT assignments operate on subbatches of 64. Rectification allocates 1,750 updates to each stage and generates training pairs with a 64-evaluation teacher. OFM uses an 8,804-parameter convex-potential network with two width-88 hidden layers, including 200 identity-pretraining updates; its inverse solver uses at most five iterations and tolerance 0.01. Endpoint evaluation uses 3,000 samples and 128 common SWD projections. Vector-field methods also share the benchmark\textquotesingle s Euler interval and numerical reference; MeanFlow uses equal intervals on $[0,1]$. The six requested budgets are 1, 2, 4, 8, 16 and 32; Optimal Flow Matching uses a direct map with one effective evaluation throughout.

\section{Additional validation and computational characterization}
\subsection{Validation profiles across displacement fractions}

At full displacement and 16 evaluations, Gaussian-source G-MPCF has validation KID 0.01305, compared with 0.01623 for I-CFM and 0.01610 for OT-CFM. Corresponding fitted-PPCA values are 0.02990, 0.04146 and 0.02479.

\subsection{Cross-validation across experimental conditions}

The full-data fit uses the 36 condition means. Cross-validation omits each condition in turn and keeps its three paired training repetitions together. Let $y_j=\overline g_{{\rm end},j}$ and let
\[
 x_j=(\overline g_{{\rm field},j},
       \overline g_{{\rm roll},j},
       \overline g_{{\rm solver},j})^\top
\]
contain its three mean diagnostics. For each condition $j$, define the training set $\mathcal I_{-j}=\{1,\ldots,36\}\setminus\{j\}$, its predictor mean $\mu_{-j}$ and coordinatewise sample standard deviation $s_{-j}$. Standardized predictors and an intercept give
\begin{equation}
 \widetilde x_{\ell,-j}=(x_\ell-\mu_{-j})\oslash s_{-j},
 \qquad a_{\ell,-j}=(1,\widetilde x_{\ell,-j}^\top)^\top,
\end{equation}
where $\oslash$ denotes coordinatewise division. Form $A_{-j}$ by stacking $a_{\ell,-j}^\top$ for $\ell\in\mathcal I_{-j}$ and stack the corresponding responses into $y_{-j}$. For coefficient vector $\beta$ containing an intercept and the three diagnostic slopes, ordinary least squares minimizes
\begin{equation}
 L_{-j}(\beta)=\|y_{-j}-A_{-j}\beta\|_2^2,
 \qquad
 \nabla_\beta L_{-j}=2A_{-j}^\top A_{-j}\beta-2A_{-j}^\top y_{-j}.
\end{equation}
All training designs have full column rank. Setting the derivative to zero and applying the same training-set standardization to the predictors for condition $j$ yields
\begin{equation}
 \widehat\beta_{-j}=(A_{-j}^\top A_{-j})^{-1}A_{-j}^\top y_{-j},
 \qquad
 \widehat y_{j,-j}=a_{j,-j}^\top\widehat\beta_{-j}.
\end{equation}
Here $\widehat\beta_{-j}$ contains the fitted coefficients and $\widehat y_{j,-j}$ is the prediction for the omitted condition. Repeating this operation gives one prediction per condition. The field-only model uses the same procedure with just the first predictor. The constant baseline predicts the mean response among the 35 training conditions.

For the 36 cross-validated predictions, let $\bar y=36^{-1}\sum_jy_j$ be the mean observed response. The root mean squared prediction error $\mathrm{RMSE}_{\rm cv}$ and coefficient of determination $R^2_{\rm cv}$ are
\begin{align}
 \mathrm{RMSE}_{\rm cv}
 &=\left[\frac1{36}\sum_{j=1}^{36}(y_j-\widehat y_{j,-j})^2\right]^{1/2},\\
 R^2_{\rm cv}
 &=1-\frac{\sum_{j=1}^{36}(y_j-\widehat y_{j,-j})^2}
              {\sum_{j=1}^{36}(y_j-\bar y)^2}.
 \label{eq:cv-score}
\end{align}

\Needspace{5\baselineskip}
The three-diagnostic model has cross-validated $R^2=0.655$ and RMSE 0.00347. The field-only model has $R^2=0.520$ and RMSE 0.00410; the training-mean predictor has RMSE 0.00608. The full-data fit has $R^2=0.888$ and RMSE 0.00198. Adding rollout and solver diagnostics reduces prediction error in all three dimension groups in Supplementary Fig.~\ref{sfig:prediction_profile}; \suppref{Table~S19} gives the aggregate prediction errors.

\Needspace{8\baselineskip}
\subsection{Computational cost and memory}

Training uses PyTorch 2.7.1, torchvision 0.22.1 and CUDA 12.8 on one NVIDIA RTX PRO 6000 Blackwell Max-Q GPU. Costs are measured across eight runs per method under the shared computation budget. Host-side timers record training and pair construction, with pairing included in total training time. G-MPCF averages 3.683 hours of training, 11.529 minutes of pair construction and 9.219 GiB peak memory, compared with 3.713 hours, 1.604 minutes and 9.215 GiB for independent pairing.

\clearpage
\begin{figure}[H]\centering
\includegraphics[width=\linewidth,height=0.76\textheight,keepaspectratio]{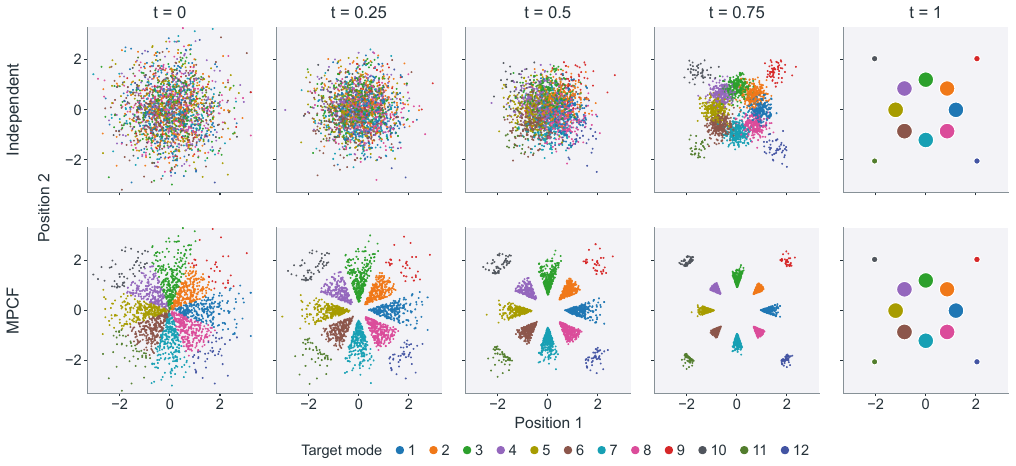}
\caption{\textbf{Mode-resolved interpolation geometry.} Rows compare independent pairing and MPCF for 2,048 shared Gaussian source samples at five interpolation times. Colours identify the twelve rare12 target modes. At $t=1$, circle areas show mode frequencies over 40,000 pairs; target locations are identical across rows.}\label{sfig:paired_fields}
\end{figure}
\clearpage
\begin{figure}[H]\centering
\includegraphics[width=\linewidth,height=0.76\textheight,keepaspectratio]{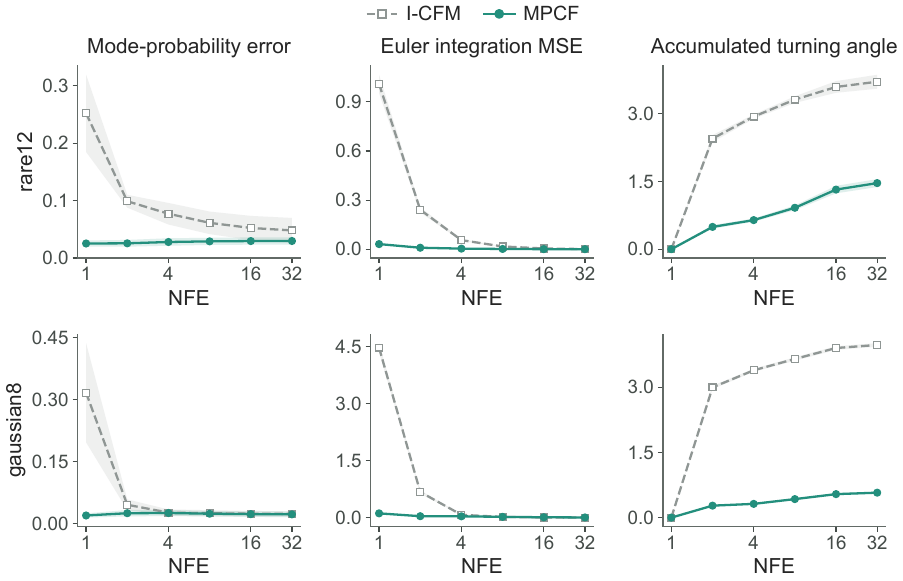}
\caption{\textbf{Inference-budget dependence in the two-dimensional benchmarks.} Rows show rare12 and gaussian8. Columns show mode-probability error, Euler integration MSE and accumulated turning angle. Curves show means and bands show standard deviations across 32 paired training runs.}\label{sfig:toy_mass}
\end{figure}
\clearpage
\begin{figure}[H]\centering
\includegraphics[width=\linewidth,height=0.76\textheight,keepaspectratio]{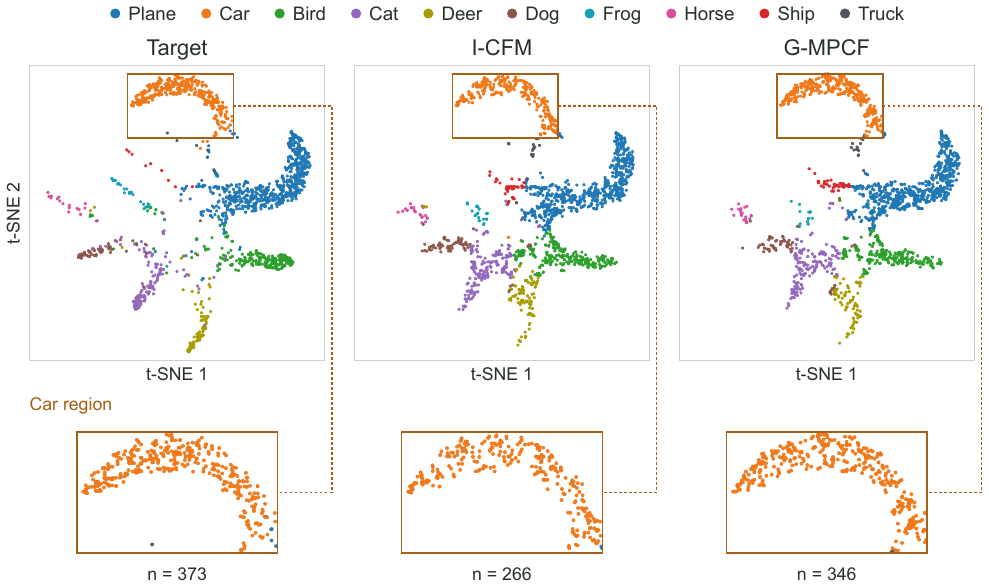}
\caption{\textbf{Long-tail feature distributions and car-region enlargements.} Target, I-CFM and G-MPCF maps contain 1,536 images each in a shared t-distributed stochastic neighbour embedding (t-SNE), pooled across eight long-tail training runs. Colours denote reference labels or predicted classes. Enlargements share target-defined car-region bounds and retain all classes inside; $n$ counts observations in each box.}\label{sfig:sample_distributions}
\end{figure}
\clearpage
\begin{figure}[H]\centering
\includegraphics[width=\linewidth,height=0.76\textheight,keepaspectratio]{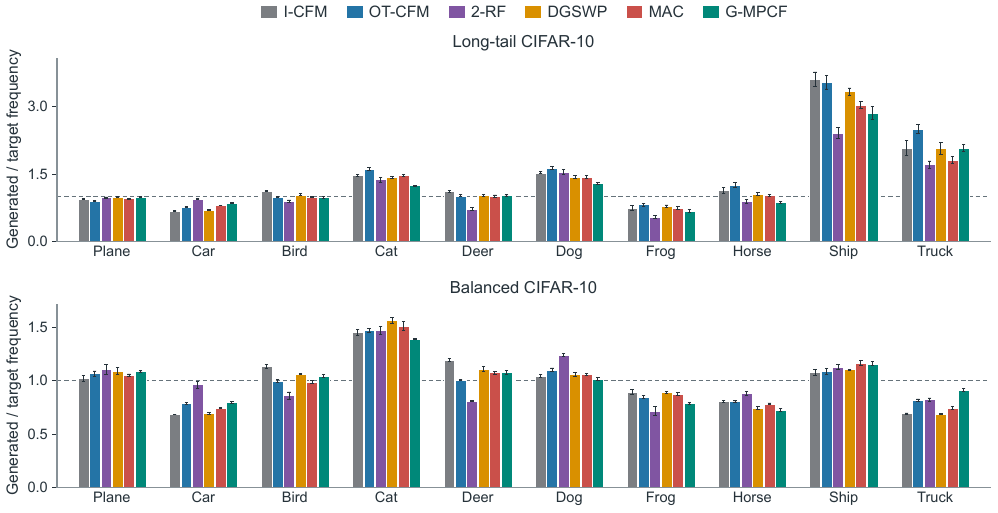}
\caption{\textbf{Class-frequency ratios under long-tail and balanced targets.} Bars show classifier-predicted frequencies divided by target probabilities; the dashed line marks 1. Bars show means and error bars show standard deviations across eight WLT100 and three balanced training runs, with 50,000 images per run at 16 Euler evaluations.}\label{sfig:class_frequencies}
\end{figure}
\clearpage
\begin{figure}[H]\centering
\includegraphics[width=\linewidth,height=0.76\textheight,keepaspectratio]{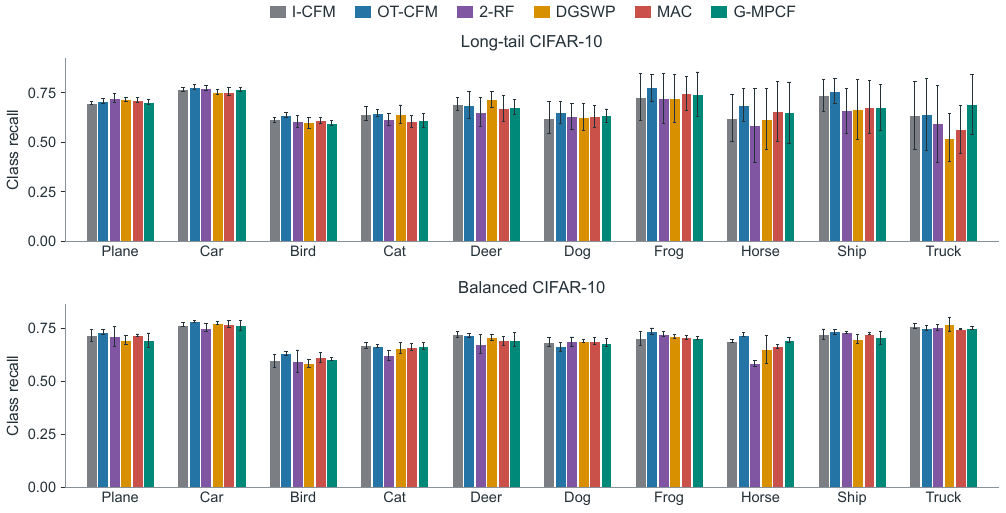}
\caption{\textbf{Class recall under long-tail and balanced targets.} Class-specific feature-manifold recall across six methods. Bars show means and error bars show standard deviations across eight WLT100 and three balanced training runs, with 50,000 images per run at 16 Euler evaluations.}\label{sfig:class_recall}
\end{figure}
\clearpage
\begin{figure}[H]\centering
\includegraphics[width=\linewidth,height=0.76\textheight,keepaspectratio]{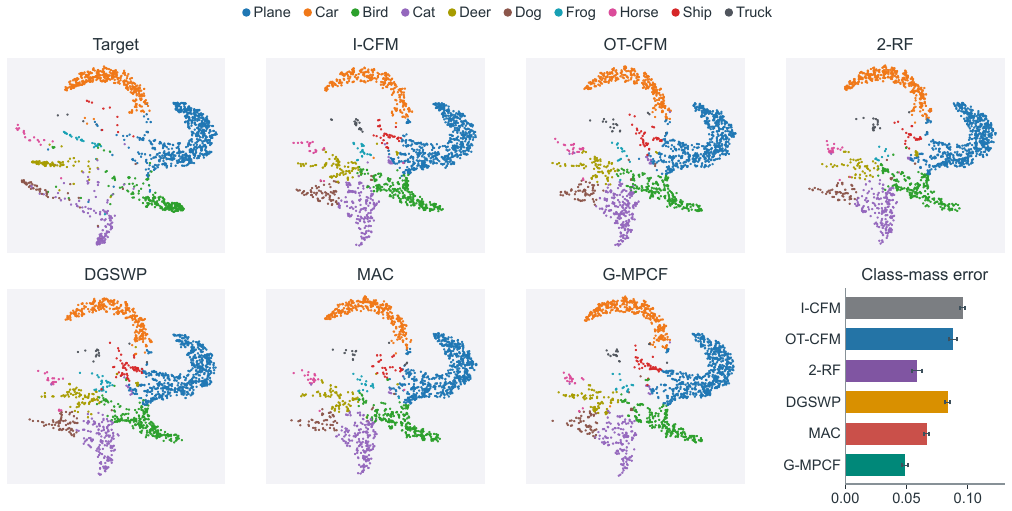}
\caption{\textbf{Sample distributions across six methods on long-tailed CIFAR-10.} Each map contains 1,536 images in a common t-SNE of classifier features, coloured by reference labels or predicted classes. Corrected MassTV is shown as mean $\pm$ standard deviation across eight training runs, with 50,000 images per run at 16 Euler evaluations.}\label{sfig:all_methods_long_tail}
\end{figure}
\clearpage
\begin{figure}[H]\centering
\includegraphics[width=\linewidth,height=0.91\textheight,keepaspectratio]{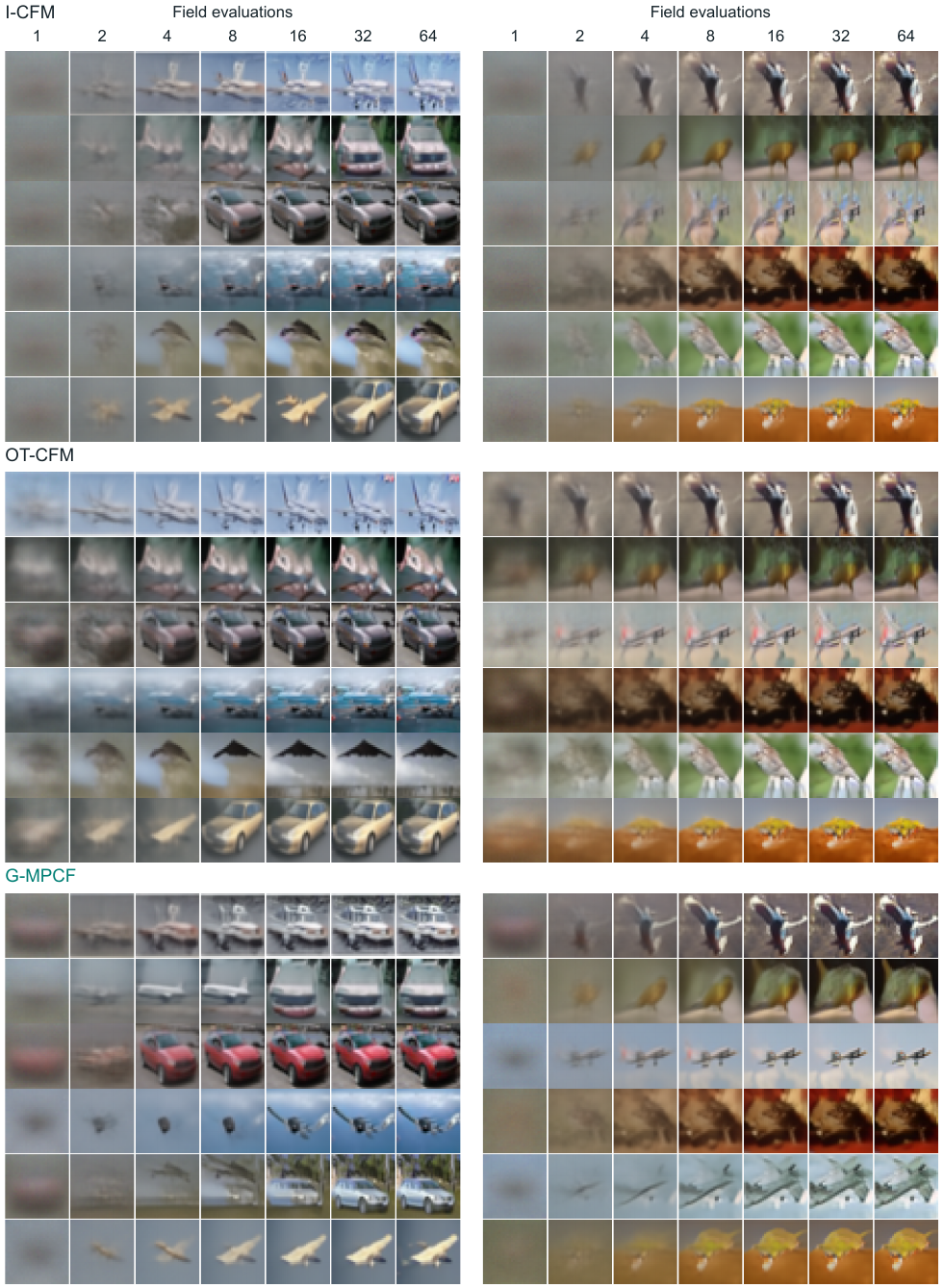}
\caption{\textbf{Generated samples across inference budgets on long-tailed CIFAR-10.} Corresponding positions compare outputs from the same Gaussian input across methods at 1, 2, 4, 8, 16, 32 and 64 Euler evaluations on WLT100. The twelve inputs comprise four from each of three training runs.}\label{sfig:budget_samples}
\end{figure}
\clearpage
\begin{figure}[H]\centering
\includegraphics[width=\linewidth,height=0.76\textheight,keepaspectratio]{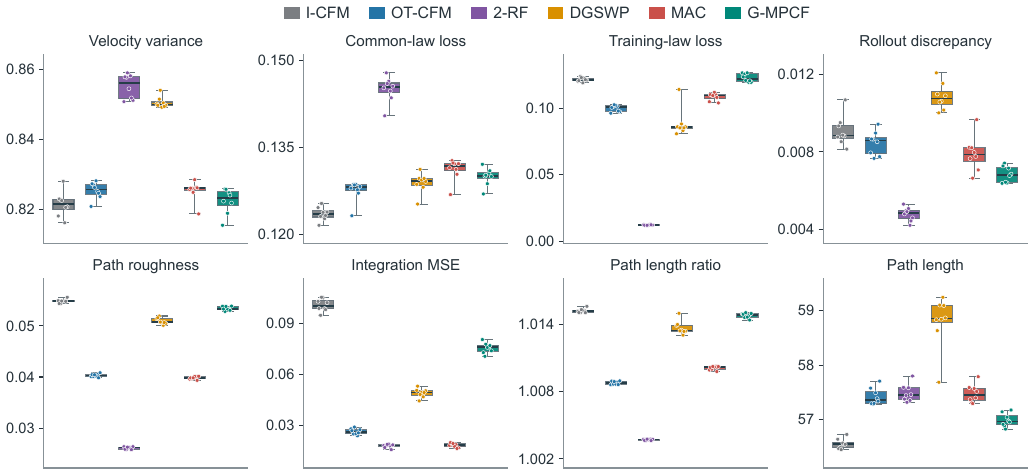}
\caption{\textbf{Field and trajectory diagnostics on WLT100.} Field regression, rollout discrepancy, path geometry and integration across six methods. Boxes show medians and interquartile ranges; whiskers span the minimum and maximum, and dots show all eight paired training runs. Euler and reference Heun trajectories use 16 and 128 evaluations, respectively.}\label{sfig:image_fields}
\end{figure}
\clearpage
\begin{figure}[H]\centering
\includegraphics[width=\linewidth,height=0.76\textheight,keepaspectratio]{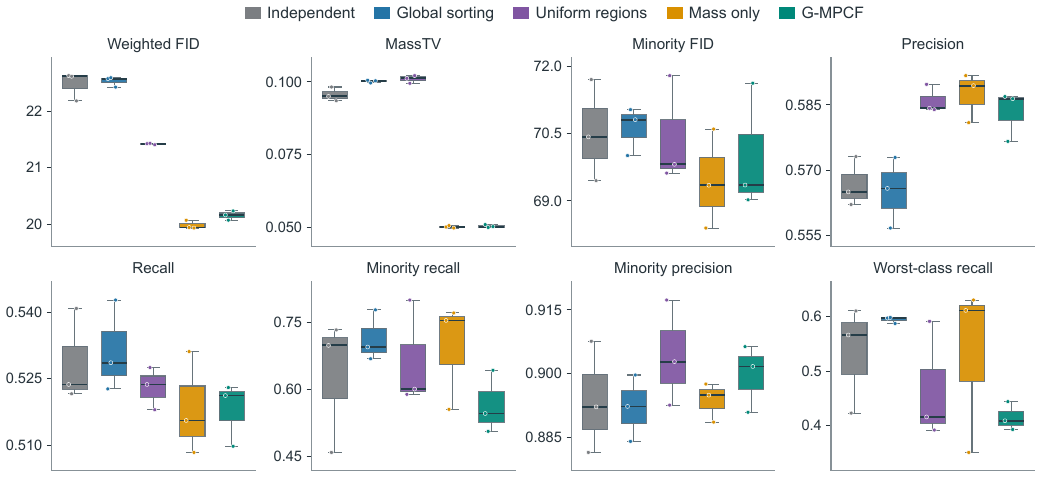}
\caption{\textbf{Allocation, fidelity and coverage under coupling interventions.} Eight image metrics compare five coupling operations. Boxes show medians and interquartile ranges; whiskers span the minimum and maximum, and dots show all three training runs. Evaluation uses 50,000 images per run and 16 Euler evaluations.}\label{sfig:component_outcomes}
\end{figure}
\clearpage
\begin{figure}[H]\centering
\includegraphics[width=\linewidth,height=0.72\textwidth,keepaspectratio]{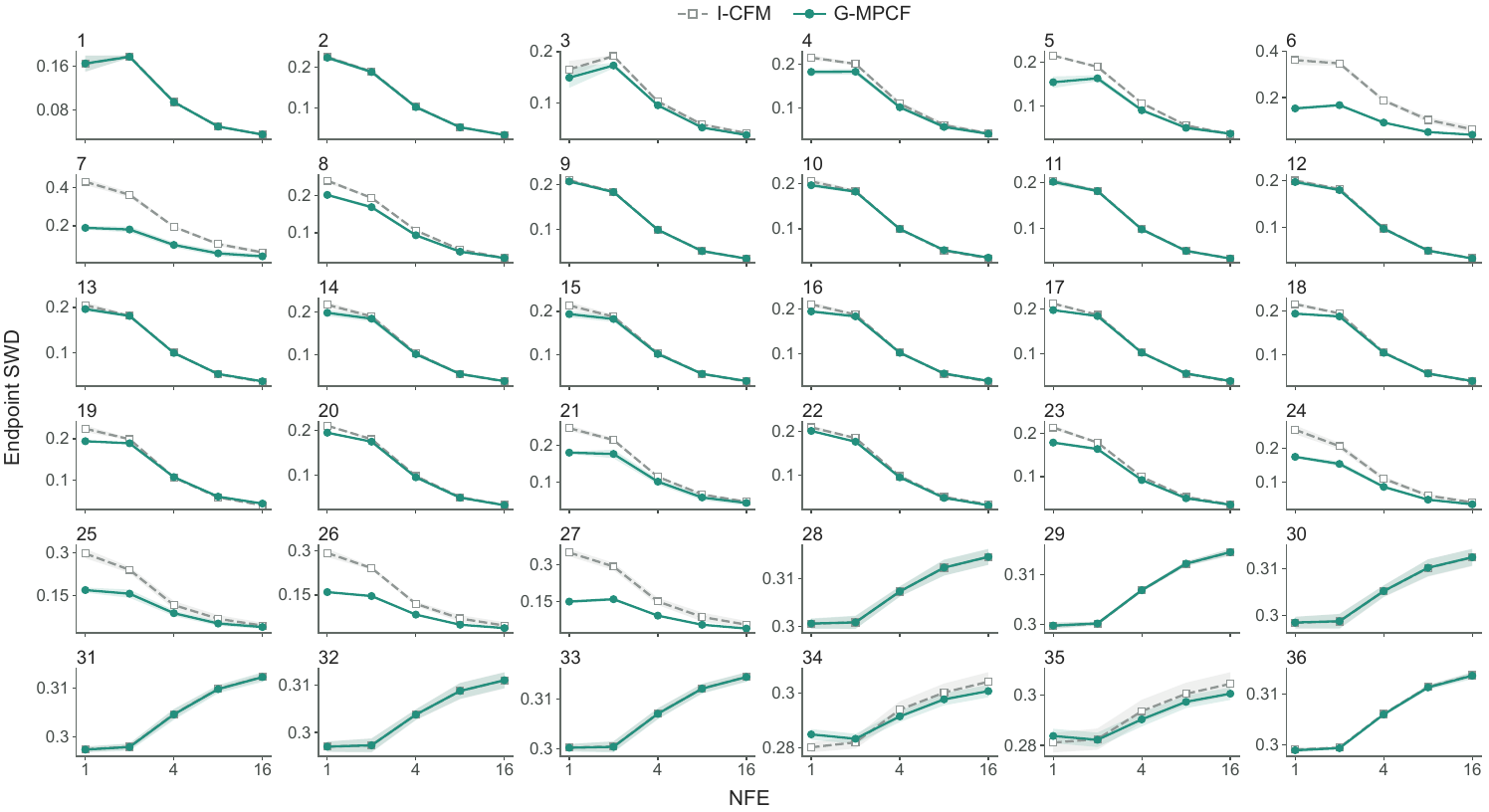}
\caption{\textbf{Inference curves across all 36 conditions.} Condition numbers follow \suppref{Table~S10}. Curves show mean endpoint SWD and bands show standard deviations across three paired training runs, with 5,000 generated points per run. NFE is logarithmic; vertical ranges vary by condition.}\label{sfig:conditions_01_06}
\end{figure}
\clearpage
\begin{figure}[H]\centering
\includegraphics[width=\linewidth,height=0.76\textheight,keepaspectratio]{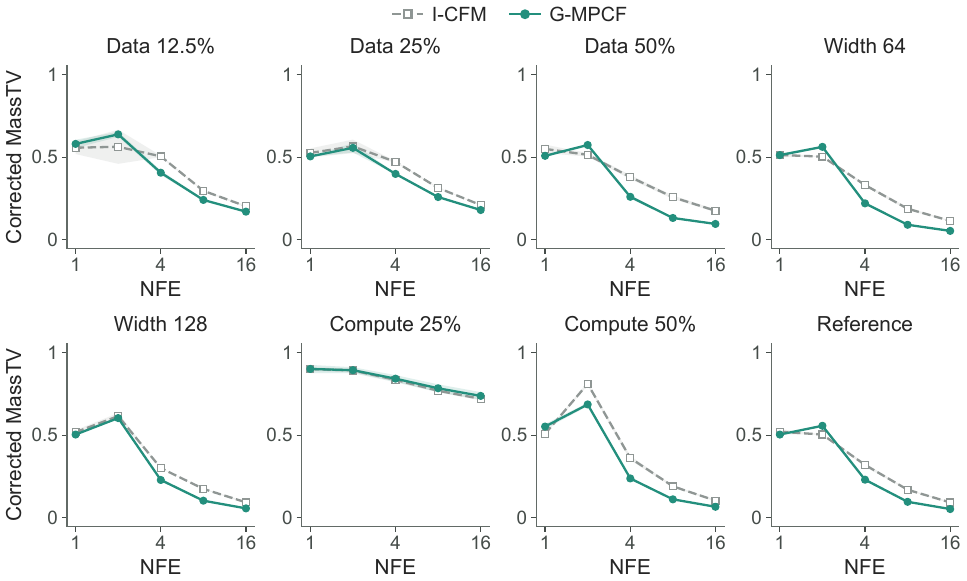}
\caption{\textbf{Class-mass error across resource settings.} The eight configurations vary data quantity, model width or computation around the full-data, width-96 reference. Curves show means and bands show standard deviations across three paired training runs, with 10,000 images per run and inference budget.}\label{sfig:scaling_mass}
\end{figure}
\clearpage
\begin{figure}[H]\centering
\includegraphics[width=\linewidth,height=0.76\textheight,keepaspectratio]{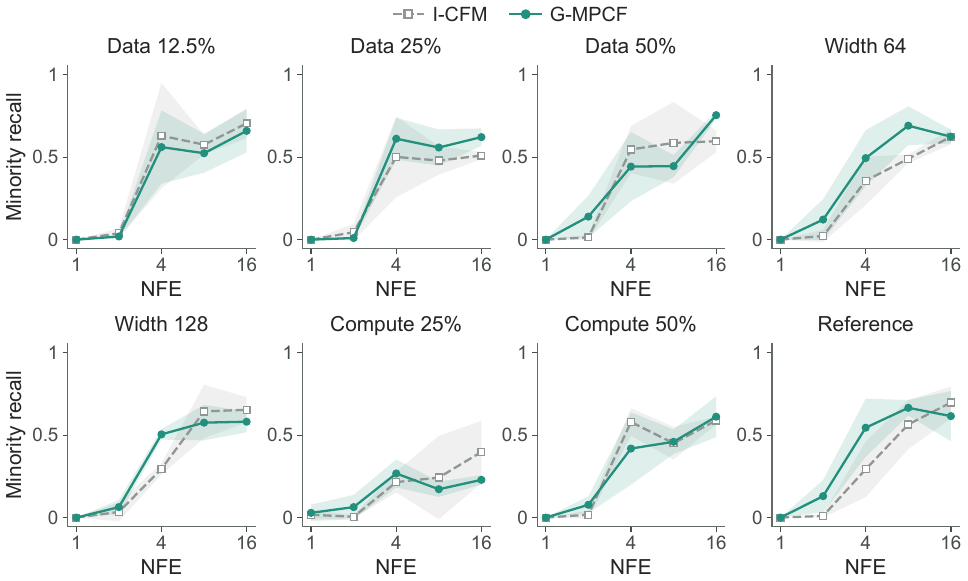}
\caption{\textbf{Minority recall across resource settings.} The eight configurations vary training-data fraction, model width or computation. Curves show means and bands show standard deviations across three paired training runs, with 10,000 images per run and inference budget.}\label{sfig:scaling_recall}
\end{figure}
\clearpage
\begin{figure}[H]\centering
\includegraphics[width=\linewidth,height=0.76\textheight,keepaspectratio]{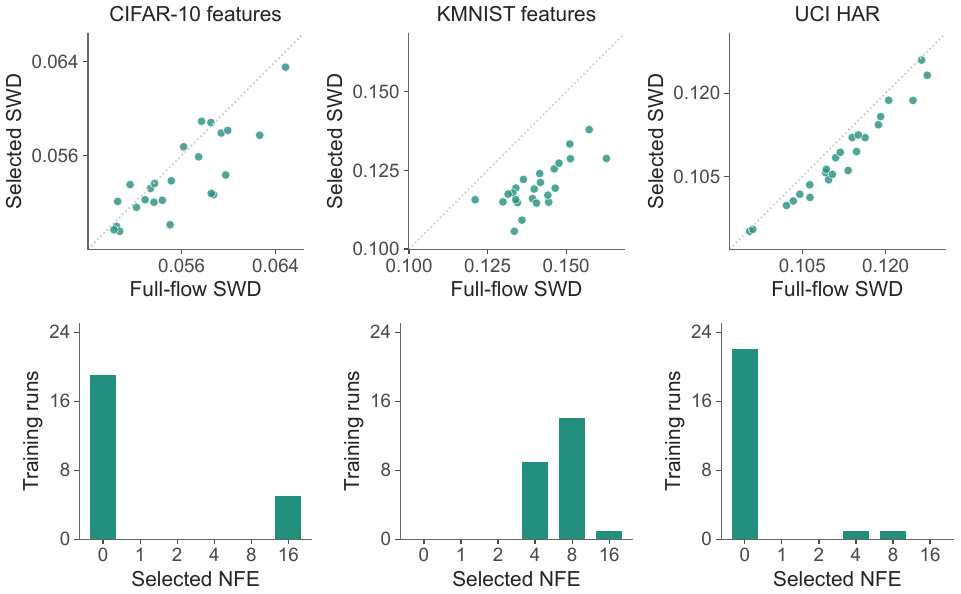}
\caption{\textbf{Endpoints and budgets after validation-based action selection.} Columns show CIFAR-10 features, KMNIST features and UCI HAR. Upper plots compare selected and full-flow test SWD for 24 training runs per dataset; dotted lines indicate equality. Lower plots count selected inference budgets. Zero evaluations retain the fitted mixture source.}\label{sfig:source_endpoints}
\end{figure}
\clearpage
\begin{figure}[H]\centering
\includegraphics[width=\linewidth,height=0.76\textheight,keepaspectratio]{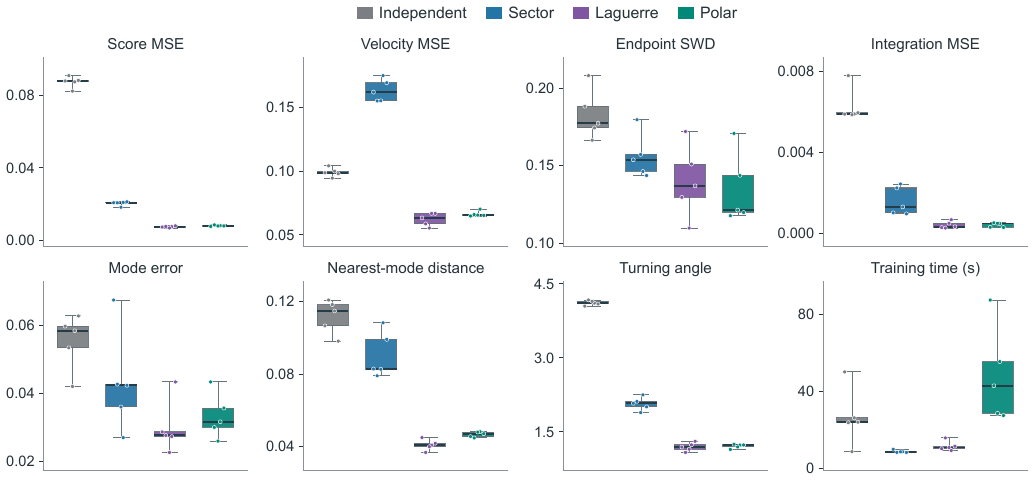}
\caption{\textbf{Accuracy and cost of four coupling constructions.} Field, endpoint, trajectory and cost diagnostics for four couplings. Boxes show medians and interquartile ranges; whiskers span the minimum and maximum, and dots show all five training runs. Endpoint and trajectory measurements use 16 evaluations.}\label{sfig:pilot_fields}
\end{figure}
\clearpage
\begin{figure}[H]\centering
\includegraphics[width=\linewidth,height=0.76\textheight,keepaspectratio]{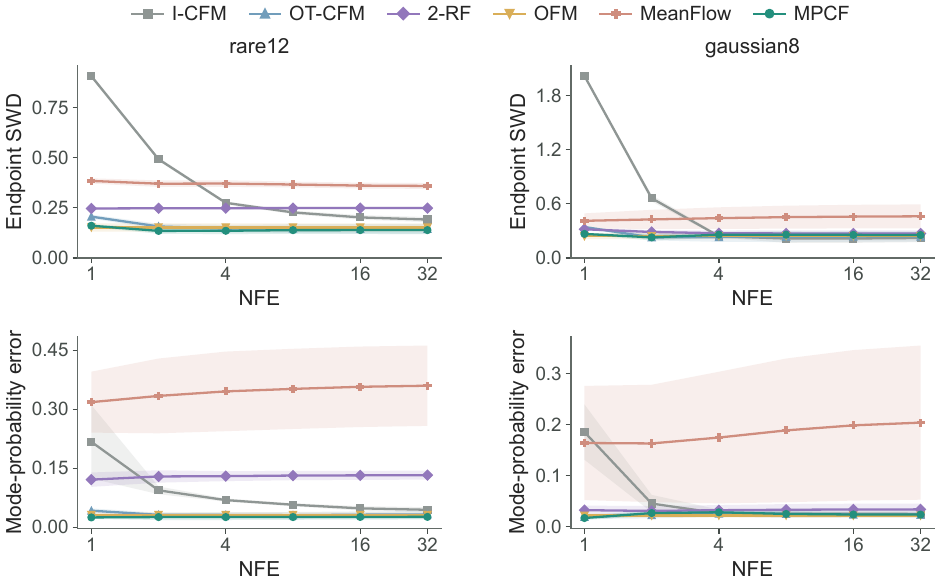}
\caption{\textbf{Endpoint quality in the six-method comparison.} Columns show rare12 and gaussian8; rows show endpoint SWD and mode-probability error. Curves show means and bands show standard deviations across five training runs. NFE denotes requested evaluations; Optimal Flow Matching uses one effective evaluation at every budget.}\label{sfig:six_swd}
\end{figure}
\clearpage
\begin{figure}[H]\centering
\includegraphics[width=\linewidth,height=0.76\textheight,keepaspectratio]{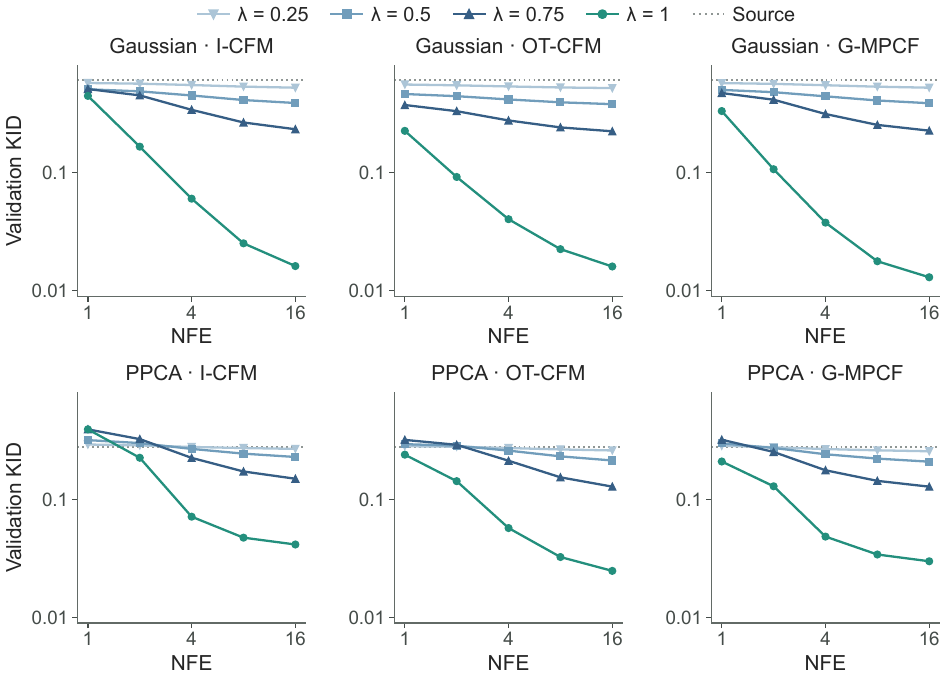}
\caption{\textbf{Validation profiles for residual displacement selection.} Rows show Gaussian and fitted PPCA sources; columns show the three couplings. Colours identify displacement fractions. Curves show mean validation KID and bands show standard errors of the mean across three training runs, with 10,000 images per run. Dotted lines indicate source retention; both axes are logarithmic.}\label{sfig:validation_gaussian}
\end{figure}
\clearpage
\begin{figure}[H]\centering
\includegraphics[width=\linewidth,height=0.76\textheight,keepaspectratio]{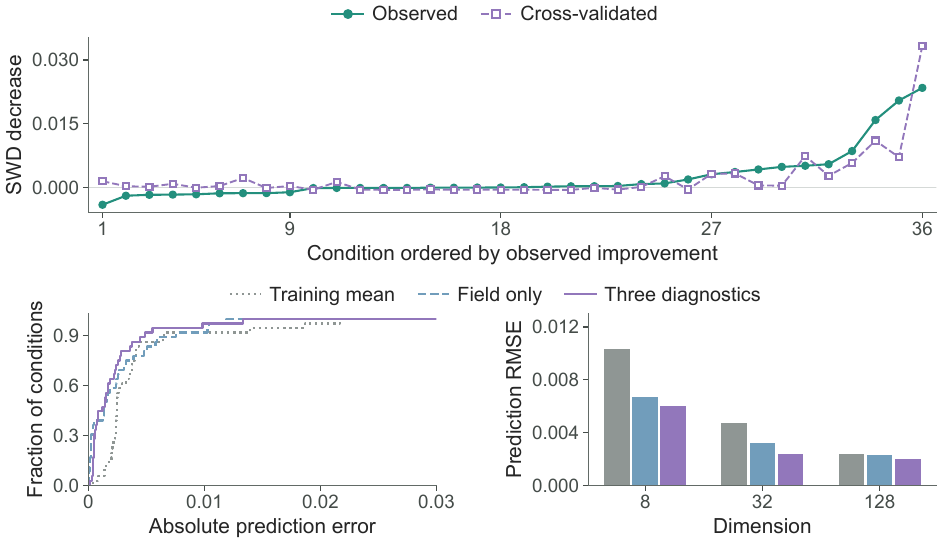}
\caption{\textbf{Cross-validated endpoint predictions and errors.} The upper plot compares observed and predicted SWD decreases across 36 conditions, each averaging three paired runs. Every prediction uses the other 35 conditions. Lower plots show absolute-error distributions and RMSE for dimensions 8, 32 and 128, containing 8, 19 and 9 conditions.}\label{sfig:prediction_profile}
\end{figure}
\clearpage
\begin{figure}[H]\centering
\includegraphics[width=\linewidth,height=0.76\textheight,keepaspectratio]{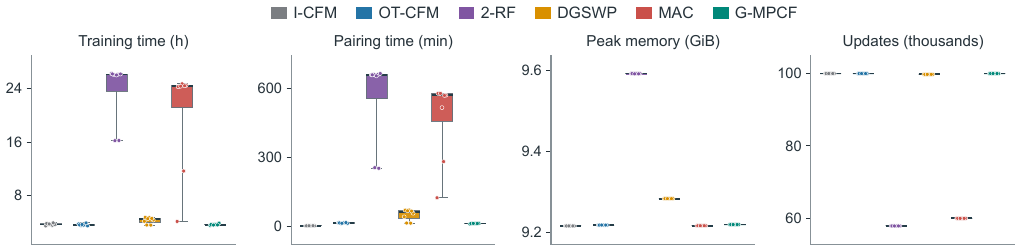}
\caption{\textbf{Computational cost of the six image methods.} Training time, pair-construction time, peak GPU memory and completed updates under a shared computation budget and GPU. Boxes show medians and interquartile ranges; whiskers span the minimum and maximum, and dots show all eight training runs. Pair construction is included in training time; rectified-flow updates combine both stages.}\label{sfig:runtime}
\end{figure}

\clearpage
\begingroup\small\setlength{\tabcolsep}{9pt}
\begin{table}[H]\centering
\caption{Paired replication. Means are over 32 independent pairs of training runs; the mean paired difference $\Delta$ is MPCF minus I-CFM, with 95\% bootstrap intervals from complete pairs. $P$ is the two-sided exact sign-flip probability. Endpoint and path metrics use 16 field evaluations.}
\begin{adjustbox}{max width=\linewidth}
\begin{tabular}{lrrrrrr}\toprule

Target & Outcome & I-CFM & MPCF & $\Delta$ & 95\% interval & $P$\\\midrule

rare12 & Score error & 0.1648 & 0.0082 & -0.1566 & [-0.1664, -0.1477] & 4.7e-10\\
rare12 & Velocity error & 0.2034 & 0.0676 & -0.1357 & [-0.1485, -0.1238] & 4.7e-10\\
rare12 & SWD & 0.2061 & 0.1335 & -0.0725 & [-0.0808, -0.0647] & 4.7e-10\\
rare12 & Mode TV & 0.0522 & 0.0294 & -0.0228 & [-0.0314, -0.0162] & 5.1e-09\\
rare12 & Solver MSE & 0.0047 & 0.0005 & -0.0042 & [-0.0044, -0.0040] & 4.7e-10\\
rare12 & Turning angle & 3.5971 & 1.3220 & -2.2750 & [-2.3212, -2.2282] & 4.7e-10\\
gaussian8 & Score error & 0.0548 & 0.0113 & -0.0435 & [-0.0443, -0.0428] & 4.7e-10\\
gaussian8 & Velocity error & 0.0173 & 0.0093 & -0.0080 & [-0.0083, -0.0077] & 4.7e-10\\
gaussian8 & SWD & 0.2543 & 0.2722 & 0.0179 & [0.0021, 0.0339] & 0.038\\
gaussian8 & Mode TV & 0.0241 & 0.0231 & -0.0010 & [-0.0031, 0.0012] & 0.36\\
gaussian8 & Solver MSE & 0.0049 & 0.0102 & 0.0053 & [0.0035, 0.0071] & 4.5e-06\\
gaussian8 & Turning angle & 3.9028 & 0.5404 & -3.3624 & [-3.3782, -3.3465] & 4.7e-10\\
\bottomrule\end{tabular}\end{adjustbox}
\end{table}\endgroup

\clearpage
\begingroup\fontsize{8.5}{10.2}\selectfont\setlength{\tabcolsep}{0.9pt}\renewcommand{\arraystretch}{1.02}
\begin{table}[H]\centering
\caption{Complete inference-budget results for the 32-pair replication. Entries are mean $\pm$ sample standard deviation.}
\begin{adjustbox}{max width=\linewidth}\begin{tabular}{@{}lrrrr@{\hspace{4mm}}lrrrr@{}}\toprule

\multicolumn{5}{c}{rare12} & \multicolumn{5}{c}{gaussian8}\\
\cmidrule(lr){1-5}\cmidrule(lr){6-10}
Method & NFE & SWD & Mode TV & Solver MSE & Method & NFE & SWD & Mode TV & Solver MSE\\\midrule

I-CFM & 1 & $0.903\pm0.017$ & $0.252\pm0.067$ & $1.009\pm0.056$ & I-CFM & 1 & $2.030\pm0.016$ & $0.317\pm0.120$ & $4.465\pm0.119$\\
I-CFM & 2 & $0.496\pm0.016$ & $0.099\pm0.011$ & $0.241\pm0.016$ & I-CFM & 2 & $0.665\pm0.026$ & $0.046\pm0.013$ & $0.674\pm0.045$\\
I-CFM & 4 & $0.274\pm0.014$ & $0.077\pm0.019$ & $0.056\pm0.006$ & I-CFM & 4 & $0.267\pm0.027$ & $0.027\pm0.008$ & $0.075\pm0.007$\\
I-CFM & 8 & $0.228\pm0.015$ & $0.061\pm0.020$ & $0.017\pm0.002$ & I-CFM & 8 & $0.253\pm0.050$ & $0.025\pm0.007$ & $0.019\pm0.004$\\
I-CFM & 16 & $0.206\pm0.018$ & $0.052\pm0.021$ & $0.005\pm0.00056$ & I-CFM & 16 & $0.254\pm0.052$ & $0.024\pm0.007$ & $0.005\pm0.001$\\
I-CFM & 32 & $0.198\pm0.020$ & $0.048\pm0.022$ & $0.001\pm0.00014$ & I-CFM & 32 & $0.257\pm0.052$ & $0.024\pm0.007$ & $0.001\pm0.00042$\\
MPCF & 1 & $0.161\pm0.011$ & $0.025\pm0.005$ & $0.031\pm0.002$ & MPCF & 1 & $0.268\pm0.023$ & $0.020\pm0.005$ & $0.113\pm0.006$\\
MPCF & 2 & $0.131\pm0.016$ & $0.025\pm0.006$ & $0.009\pm0.00089$ & MPCF & 2 & $0.258\pm0.053$ & $0.025\pm0.008$ & $0.038\pm0.008$\\
MPCF & 4 & $0.131\pm0.020$ & $0.028\pm0.007$ & $0.004\pm0.00066$ & MPCF & 4 & $0.281\pm0.059$ & $0.026\pm0.007$ & $0.035\pm0.015$\\
MPCF & 8 & $0.133\pm0.021$ & $0.029\pm0.006$ & $0.002\pm0.0005$ & MPCF & 8 & $0.277\pm0.055$ & $0.024\pm0.007$ & $0.021\pm0.010$\\
MPCF & 16 & $0.134\pm0.021$ & $0.029\pm0.006$ & $0.00047\pm0.00021$ & MPCF & 16 & $0.272\pm0.054$ & $0.023\pm0.007$ & $0.010\pm0.005$\\
MPCF & 32 & $0.134\pm0.021$ & $0.030\pm0.007$ & $0.00012\pm6.8e-05$ & MPCF & 32 & $0.270\pm0.053$ & $0.023\pm0.007$ & $0.005\pm0.003$\\
\bottomrule\end{tabular}\end{adjustbox}
\end{table}\endgroup

\clearpage
\begingroup\small\setlength{\tabcolsep}{9pt}
\begin{table}[H]\centering
\caption{Image endpoints at NFE 16 with 50,000 generated images per run. Entries are means $\pm$ sample standard deviations across $n$ independent training runs.}
\begin{adjustbox}{max width=\linewidth}\begin{tabular}{lrrrrrr}\toprule

Task & Method & $n$ & FID & MassTV & Minority recall & KID $\times10^3$\\\midrule

WLT100 & I-CFM & 8 & $22.35\pm0.20$ & $0.0959\pm0.0020$ & $0.667\pm0.058$ & $16.24\pm0.41$\\
WLT100 & OT-CFM & 8 & $22.06\pm0.15$ & $0.0878\pm0.0034$ & $0.697\pm0.090$ & $16.26\pm0.34$\\
WLT100 & 2-RF & 8 & $20.16\pm0.20$ & $0.0586\pm0.0043$ & $0.614\pm0.047$ & $8.51\pm0.20$\\
WLT100 & DGSWP & 8 & $21.28\pm0.39$ & $0.0835\pm0.0023$ & $0.603\pm0.077$ & $14.42\pm0.57$\\
WLT100 & MAC & 8 & $19.60\pm0.24$ & $0.0664\pm0.0019$ & $0.634\pm0.095$ & $11.45\pm0.57$\\
WLT100 & G-MPCF & 8 & $19.80\pm0.20$ & $0.0486\pm0.0023$ & $0.674\pm0.085$ & $13.23\pm0.36$\\
Balanced & I-CFM & 3 & $18.43\pm0.13$ & $0.0969\pm0.0017$ & $0.725\pm0.010$ & $17.29\pm0.36$\\
Balanced & OT-CFM & 3 & $17.11\pm0.07$ & $0.0763\pm0.0024$ & $0.735\pm0.005$ & $16.08\pm0.30$\\
Balanced & 2-RF & 3 & $17.64\pm0.56$ & $0.0976\pm0.0052$ & $0.689\pm0.008$ & $9.99\pm0.35$\\
Balanced & DGSWP & 3 & $19.64\pm0.30$ & $0.1028\pm0.0034$ & $0.706\pm0.020$ & $17.92\pm0.58$\\
Balanced & MAC & 3 & $16.77\pm0.11$ & $0.0886\pm0.0036$ & $0.712\pm0.003$ & $13.43\pm0.25$\\
Balanced & G-MPCF & 3 & $17.71\pm0.15$ & $0.0832\pm0.0013$ & $0.718\pm0.010$ & $15.95\pm0.12$\\
\bottomrule\end{tabular}\end{adjustbox}
\end{table}\endgroup

\clearpage
\begingroup\small\setlength{\tabcolsep}{9pt}
\begin{table}[H]\centering
\caption{WLT100 component tests, eight paired training runs. Beneficial effects are positive, and non-inferiority effects include the stated margin. NI and SUP denote non-inferiority and superiority, respectively. Intervals are unadjusted two-sided 95\% paired-bootstrap intervals. Rejection uses the Holm-adjusted one-sided probability $P_{\rm Holm}$.}
\begin{adjustbox}{max width=\linewidth}\begin{tabular}{lrrrrrr}\toprule

Outcome & Comparator & Margin & Effect & 95\% interval & $P_{\rm Holm}$ & Reject\\\midrule

MassTV SUP & I-CFM & 0.0000 & 0.0473 & [0.0453, 0.0491] & 0.016 & Yes\\
Recall SUP & I-CFM & 0.0000 & 0.0069 & [-0.0595, 0.0771] & 0.859 & No\\
MassTV SUP & OT-CFM & 0.0000 & 0.0391 & [0.0367, 0.0414] & 0.016 & Yes\\
Recall SUP & OT-CFM & 0.0000 & -0.0232 & [-0.0850, 0.0465] & 0.859 & No\\
FID NI & I-CFM & 1.1174 & 3.6675 & [3.5536, 3.7879] & 0.008 & Yes\\
FID NI & OT-CFM & 1.1030 & 3.3647 & [3.2614, 3.4778] & 0.008 & Yes\\
MassTV NI & OT-CFM & 0.0100 & 0.0491 & [0.0467, 0.0513] & 0.008 & Yes\\
Recall NI & OT-CFM & 0.0200 & -0.0032 & [-0.0645, 0.0661] & 0.539 & No\\
\bottomrule\end{tabular}\end{adjustbox}
\end{table}\endgroup

\clearpage
\begingroup\fontsize{8.5}{10.2}\selectfont\setlength{\tabcolsep}{1.2pt}\renewcommand{\arraystretch}{1.02}
\begin{table}[H]\centering
\caption{Individual image-training runs. Precision and recall measure overall feature support; minority recall uses classes 7--9. Run numbers identify matched repetitions.}
\begin{adjustbox}{max width=\linewidth}\begin{tabular}{@{}lrrrrrrr@{\hspace{4mm}}lrrrrrrr@{}}\toprule

Task & Method & Run & FID & MassTV & Min. R & Prec. & Recall & Task & Method & Run & FID & MassTV & Min. R & Prec. & Recall\\\midrule

WLT & DGSWP & 1 & 21.051 & 0.0808 & 0.6150 & 0.5685 & 0.5413 & WLT & OT-CFM & 1 & 21.998 & 0.0863 & 0.6440 & 0.5753 & 0.5477\\
WLT & DGSWP & 2 & 20.964 & 0.0848 & 0.6973 & 0.5646 & 0.5422 & WLT & OT-CFM & 2 & 21.882 & 0.0849 & 0.5853 & 0.5768 & 0.5401\\
WLT & DGSWP & 3 & 21.788 & 0.0872 & 0.5993 & 0.5504 & 0.5430 & WLT & OT-CFM & 3 & 22.018 & 0.0874 & 0.8180 & 0.5733 & 0.5462\\
WLT & DGSWP & 4 & 21.264 & 0.0832 & 0.6583 & 0.5597 & 0.5513 & WLT & OT-CFM & 4 & 21.909 & 0.0888 & 0.6937 & 0.5804 & 0.5258\\
WLT & DGSWP & 5 & 20.621 & 0.0833 & 0.6613 & 0.5616 & 0.5377 & WLT & OT-CFM & 5 & 22.254 & 0.0915 & 0.8173 & 0.5738 & 0.5453\\
WLT & DGSWP & 6 & 21.702 & 0.0838 & 0.6127 & 0.5559 & 0.5446 & WLT & OT-CFM & 6 & 22.133 & 0.0916 & 0.6303 & 0.5800 & 0.5437\\
WLT & DGSWP & 7 & 21.394 & 0.0847 & 0.5100 & 0.5603 & 0.5364 & WLT & OT-CFM & 7 & 22.286 & 0.0900 & 0.6300 & 0.5717 & 0.5249\\
\bottomrule\end{tabular}\end{adjustbox}
\end{table}\endgroup

\clearpage
\begingroup\small\setlength{\tabcolsep}{9pt}
\begin{table}[H]\centering
\caption{Image training computation. Charged computation includes method-specific pairing and teacher costs. Time is mean $\pm$ standard deviation; memory is the maximum recorded across training runs.}
\begin{adjustbox}{max width=\linewidth}\begin{tabular}{lrrrrrr}\toprule

Task & Method & Updates & FLOPs $/10^{17}$ & Max. dev., \% & Hours & GiB\\\midrule

WLT & G-MPCF & 99999 & 5.42494 & 0.00074 & $3.68\pm0.10$ & 9.22\\
WLT & I-CFM & 100000 & 5.42498 & 0.00000 & $3.71\pm0.15$ & 9.22\\
WLT & OT-CFM & 99992 & 5.42497 & 0.00027 & $3.66\pm0.14$ & 9.22\\
WLT & DGSWP & 99772 & 5.42498 & 0.00005 & $4.30\pm0.47$ & 9.28\\
WLT & MAC & 59989 & 5.42492 & 0.00116 & $20.29\pm7.90$ & 9.22\\
WLT & 2-RF & 57893 & 5.42485 & 0.00249 & $23.65\pm4.58$ & 9.59\\
Bal. & I-CFM & 100000 & 5.42498 & 0.00000 & $3.37\pm0.06$ & 9.22\\
Bal. & OT-CFM & 99992 & 5.42497 & 0.00027 & $3.66\pm0.09$ & 9.22\\
Bal. & 2-RF & 57893 & 5.42485 & 0.00249 & $3.91\pm0.02$ & 9.59\\
Bal. & DGSWP & 99772 & 5.42498 & 0.00005 & $3.64\pm0.03$ & 9.28\\
Bal. & MAC & 59989 & 5.42492 & 0.00116 & $4.16\pm0.16$ & 9.22\\
Bal. & G-MPCF & 99999 & 5.42494 & 0.00074 & $3.35\pm0.01$ & 9.22\\
\bottomrule\end{tabular}\end{adjustbox}
\end{table}\endgroup

\clearpage
\begin{table}[H]\centering\small
\caption{Shared image training specification. Parameter counts refer to the trainable velocity field; $\beta_1$ and $\beta_2$ are AdamW decay coefficients for the first and second gradient moments.}
\begin{adjustbox}{max width=\linewidth}\begin{tabular}{ll}\toprule
Quantity&Setting\\\midrule
Resolution and channels&$32\times32\times3$\\
Base width; multipliers&96; 1, 2, 2, 2\\
Encoder / decoder blocks&2 / 3 per resolution\\
Attention blocks; heads&6; $3\times64$\\
Trainable parameters&20,118,915\\
Normalization; activation&GroupNorm; SiLU\\
Dropout&0.1\\
Batch size&256\\
Optimizer&AdamW\\
Learning rate; coefficients&$2\times10^{-4}$; $\beta_1=0.9$, $\beta_2=0.999$\\
Weight decay; gradient clipping&0; 1.0\\
Warm-up; schedule&5,000 updates; cosine\\
Moving-average decay&0.9999\\
Time sampling&Uniform stratification, one jitter per bin\\
Arithmetic&bfloat16 autocast, float32 loss\\
Augmentation&Horizontal flip, probability 0.5\\
Training FLOP envelope&$\approx5.425\times10^{17}$\\
Primary sampler&Euler, NFE 16\\
Primary images per run&50,000\\\bottomrule
\end{tabular}\end{adjustbox}\end{table}

\clearpage
\begingroup\small\setlength{\tabcolsep}{9pt}
\begin{table}[H]\centering
\caption{Component interventions and mass relaxation, three independent training runs per action. Entries are mean $\pm$ sample standard deviation. Population TV concerns the sampled target rows.}
\begin{adjustbox}{max width=\linewidth}\begin{tabular}{lrrrr}\toprule

Action & FID & MassTV & Minority recall & Population TV\\\midrule

Geometry & $22.53\pm0.09$ & $0.100\pm0.00049$ & $0.714\pm0.057$ & 0.0000\\
G-MPCF & $20.16\pm0.08$ & $0.050\pm0.00051$ & $0.565\pm0.070$ & 1.2e-17\\
Independent & $22.48\pm0.25$ & $0.096\pm0.002$ & $0.631\pm0.150$ & 0.0000\\
Mass only & $19.98\pm0.08$ & $0.050\pm0.00043$ & $0.694\pm0.120$ & 1.2e-17\\
$\rho=0.25$ & $20.15\pm0.03$ & $0.060\pm0.002$ & $0.713\pm0.044$ & 0.0311\\
$\rho=0.50$ & $20.28\pm0.25$ & $0.071\pm0.00041$ & $0.704\pm0.075$ & 0.0621\\
$\rho=0.75$ & $20.56\pm0.09$ & $0.084\pm0.002$ & $0.727\pm0.068$ & 0.0932\\
Uncalibrated & $21.42\pm0.01$ & $0.101\pm0.001$ & $0.664\pm0.119$ & 0.1242\\
\bottomrule\end{tabular}\end{adjustbox}
\end{table}\endgroup

\clearpage
\begingroup\fontsize{8.5}{10.2}\selectfont\setlength{\tabcolsep}{1.7pt}\renewcommand{\arraystretch}{1.02}
\begin{table}[H]\centering
\caption{The 24 component-intervention runs, with matched repetitions identified by run number. Min. R and Min. P denote minority recall and precision.}
\begin{adjustbox}{max width=\linewidth}\begin{tabular}{@{}lrrrrr@{\hspace{4mm}}lrrrrr@{}}\toprule

Action & Run & FID & MassTV & Min. R & Min. P & Action & Run & FID & MassTV & Min. R & Min. P\\\midrule

Geometry & 1 & 22.575 & 0.1006 & 0.6953 & 0.8922 & $\rho=0.25$ & 1 & 20.171 & 0.0611 & 0.7093 & 0.8968\\
Geometry & 2 & 22.425 & 0.1004 & 0.7787 & 0.8997 & $\rho=0.25$ & 2 & 20.115 & 0.0578 & 0.7577 & 0.9090\\
Geometry & 3 & 22.591 & 0.0996 & 0.6690 & 0.8840 & $\rho=0.25$ & 3 & 20.156 & 0.0610 & 0.6707 & 0.8948\\
G-MPCF & 1 & 20.163 & 0.0510 & 0.5463 & 0.9063 & $\rho=0.50$ & 1 & 20.358 & 0.0706 & 0.6243 & 0.9185\\
G-MPCF & 2 & 20.234 & 0.0503 & 0.6433 & 0.9017 & $\rho=0.50$ & 2 & 20.003 & 0.0709 & 0.7740 & 0.9096\\
G-MPCF & 3 & 20.071 & 0.0500 & 0.5063 & 0.8909 & $\rho=0.50$ & 3 & 20.488 & 0.0714 & 0.7123 & 0.8939\\
Independent & 1 & 22.629 & 0.0950 & 0.6990 & 0.8814 & $\rho=0.75$ & 1 & 20.666 & 0.0848 & 0.6730 & 0.8980\\
Independent & 2 & 22.181 & 0.0935 & 0.7337 & 0.8921 & $\rho=0.75$ & 2 & 20.527 & 0.0852 & 0.8030 & 0.9002\\
Independent & 3 & 22.616 & 0.0982 & 0.4590 & 0.9076 & $\rho=0.75$ & 3 & 20.488 & 0.0819 & 0.7050 & 0.9041\\
Mass only & 1 & 20.070 & 0.0502 & 0.7550 & 0.8975 & Uncalibrated & 1 & 21.424 & 0.1012 & 0.5893 & 0.9173\\
Mass only & 2 & 19.932 & 0.0498 & 0.7717 & 0.8885 & Uncalibrated & 2 & 21.404 & 0.1022 & 0.6010 & 0.9028\\
Mass only & 3 & 19.942 & 0.0507 & 0.5557 & 0.8949 & Uncalibrated & 3 & 21.431 & 0.0995 & 0.8007 & 0.8926\\
\bottomrule\end{tabular}\end{adjustbox}
\end{table}\endgroup

\clearpage
\begingroup\fontsize{8.5}{10.2}\selectfont\setlength{\tabcolsep}{2.2pt}
\begin{table}[H]\centering
\caption{All 36 experimental conditions and diagnostic gains, averaged over three paired training runs. Design parameters are dimension $d$, rare probability $p$, separation $s$ and source displacement $\delta$. Field gain is a relative reduction; the other gains are absolute I-CFM minus G-MPCF differences.}
\begin{adjustbox}{max width=\linewidth}\begin{tabular}{@{}lrrrrr@{\hspace{4mm}}lrrrrr@{}}\toprule
Cond. & $d,p,s,\delta$ & Field & Rollout & Solver & SWD & Cond. & $d,p,s,\delta$ & Field & Rollout & Solver & SWD\\\midrule
1 & $8,0.01,1,0$ & 0.0024 & 0.0006 & 2.9e-06 & 4.4e-05 & 19 & $32,0.05,6,2$ & 0.0161 & 0.0488 & 0.0009 & -0.0040\\
2 & $8,0.01,1,4$ & 0.0055 & 0.0005 & 1.8e-05 & 1.2e-05 & 20 & $32,0.10,2,4$ & 0.0168 & 0.0066 & 3.0e-05 & 0.0004\\
3 & $8,0.01,8,0$ & 0.0144 & 0.0922 & 0.0068 & 0.0037 & 21 & $32,0.10,6,2$ & 0.0277 & 0.1034 & 0.0014 & 0.0032\\
4 & $8,0.02,6,2$ & 0.0362 & 0.0693 & 0.0035 & 0.0010 & 22 & $32,0.20,1,0$ & 0.0077 & 0.0015 & 9.9e-06 & 0.0020\\
5 & $8,0.10,2,1$ & 0.0401 & 0.0554 & 0.0005 & -0.0013 & 23 & $32,0.20,2,0$ & 0.0096 & 0.0180 & 2.9e-05 & 0.0009\\
6 & $8,0.10,8,1$ & 0.1345 & 0.6085 & 0.0481 & 0.0234 & 24 & $32,0.20,4,0$ & 0.0227 & 0.0998 & 0.0006 & 0.0055\\
7 & $8,0.10,8,4$ & 0.2151 & 0.3086 & 0.0497 & 0.0205 & 25 & $32,0.20,6,0$ & 0.0359 & 0.2178 & 0.0025 & 0.0052\\
\bottomrule\end{tabular}\end{adjustbox}
\end{table}\endgroup

\clearpage
\begingroup\fontsize{8.5}{10.2}\selectfont\setlength{\tabcolsep}{1.8pt}
\begin{table}[H]\centering
\caption{All 108 paired diagnostic measurements at 16 evaluations. Pair identifiers give condition/run. Field gain is a relative loss reduction; rollout, solver and SWD gains are absolute I-CFM minus G-MPCF differences. Positive values indicate improvement.}
\begin{adjustbox}{max width=\linewidth}\begin{tabular}{@{}lrrrr@{\hspace{4mm}}lrrrr@{\hspace{4mm}}lrrrr@{}}\toprule
Pair & Field & Rollout & Solver & SWD & Pair & Field & Rollout & Solver & SWD & Pair & Field & Rollout & Solver & SWD\\\midrule
1/1 & 0.0023 & 0.0006 & 1.3e-06 & -0.0003 & 13/1 & 0.0051 & 0.0094 & 4.5e-05 & -0.0021 & 25/1 & 0.0362 & 0.2146 & 0.0048 & 0.0090\\
1/2 & 0.0040 & 0.0002 & 1.1e-05 & 0.0001 & 13/2 & 0.0048 & 0.0098 & 5.2e-05 & -0.0013 & 25/2 & 0.0357 & 0.2156 & 0.0021 & 0.0077\\
1/3 & 0.0010 & 0.0009 & -3.3e-06 & 0.0002 & 13/3 & 0.0049 & 0.0099 & 5.8e-05 & -0.0004 & 25/3 & 0.0357 & 0.2233 & 0.0007 & -0.0012\\
2/1 & 0.0058 & 7.3e-05 & 4.7e-06 & -9.3e-05 & 14/1 & 0.0053 & 0.0203 & 0.0006 & -0.0019 & 26/1 & 0.0463 & 0.1917 & 0.0041 & 0.0213\\
2/2 & 0.0063 & 0.0010 & 4.1e-05 & -5.3e-06 & 14/2 & 0.0052 & 0.0254 & 9.6e-05 & -0.0005 & 26/2 & 0.0428 & 0.1920 & 0.0012 & 0.0082\\
2/3 & 0.0044 & 0.0003 & 7.4e-06 & 0.0001 & 14/3 & 0.0050 & 0.0252 & 0.0004 & -0.0008 & 26/3 & 0.0426 & 0.2017 & 0.0014 & -0.0037\\
3/1 & 0.0070 & 0.1079 & 0.0026 & 0.0008 & 15/1 & 0.0060 & 0.0190 & 0.0006 & -0.0021 & 27/1 & 0.0599 & 0.3548 & 0.0165 & 0.0329\\
\bottomrule\end{tabular}\end{adjustbox}
\end{table}\endgroup

\clearpage
\begingroup\fontsize{8.5}{10.2}\selectfont\setlength{\tabcolsep}{0.9pt}\renewcommand{\arraystretch}{1.02}
\begin{table}[H]\centering
\caption{Image scaling results across eight settings, three paired training runs and 10,000 images per method and inference budget. I and G denote I-CFM and G-MPCF. FID is mean $\pm$ sample standard deviation; TV denotes mean corrected MassTV. Data fraction $f$ is in per mille, $w$ is the U-Net base width and compute $c$ is in percent.}
\begin{adjustbox}{max width=\linewidth}\begin{tabular}{@{}lrrrrrrr@{\hspace{4mm}}lrrrrrrr@{}}\toprule

$f$ & $w$ & $c$ & NFE & FID I & FID G & TV I & TV G & $f$ & $w$ & $c$ & NFE & FID I & FID G & TV I & TV G\\\midrule

125 & 96 & 100 & 1 & $392.33\pm6.66$ & $332.49\pm3.43$ & 0.5574 & 0.5806 & 1000 & 96 & 25 & 1 & $407.46\pm9.61$ & $396.41\pm3.12$ & 0.9007 & 0.9011\\
125 & 96 & 100 & 2 & $197.48\pm9.38$ & $164.66\pm2.45$ & 0.5633 & 0.6385 & 1000 & 96 & 25 & 2 & $321.82\pm6.16$ & $317.52\pm6.37$ & 0.8920 & 0.8936\\
125 & 96 & 100 & 4 & $84.75\pm2.19$ & $68.84\pm1.26$ & 0.5051 & 0.4059 & 1000 & 96 & 25 & 4 & $270.42\pm5.14$ & $266.74\pm7.41$ & 0.8351 & 0.8422\\
125 & 96 & 100 & 8 & $46.19\pm0.72$ & $39.80\pm0.41$ & 0.2964 & 0.2409 & 1000 & 96 & 25 & 8 & $240.22\pm5.06$ & $238.33\pm7.42$ & 0.7693 & 0.7840\\
125 & 96 & 100 & 16 & $31.45\pm0.05$ & $29.20\pm0.36$ & 0.2024 & 0.1698 & 1000 & 96 & 25 & 16 & $224.42\pm5.01$ & $223.95\pm7.34$ & 0.7199 & 0.7378\\
250 & 96 & 100 & 1 & $392.78\pm5.56$ & $350.96\pm4.19$ & 0.5260 & 0.5049 & 1000 & 96 & 50 & 1 & $389.56\pm3.47$ & $344.57\pm4.85$ & 0.5080 & 0.5520\\
250 & 96 & 100 & 2 & $185.45\pm5.41$ & $149.84\pm2.16$ & 0.5657 & 0.5568 & 1000 & 96 & 50 & 2 & $223.39\pm7.61$ & $179.21\pm3.04$ & 0.8108 & 0.6859\\
\bottomrule\end{tabular}\end{adjustbox}
\end{table}\endgroup

\clearpage
\begingroup\small\setlength{\tabcolsep}{9pt}
\begin{table}[H]\centering
\caption{Scaling architecture and computation. Data fraction $f$ is in per mille, $w$ is the U-Net base width and charged computation $c$ is in percent.}
\begin{adjustbox}{max width=\linewidth}\begin{tabular}{lrrrrrrr}\toprule

$f$ & $w$ & $c$ & Method & Updates & Params $/10^6$ & FLOPs $/10^{17}$ & Hours\\\midrule

125 & 96 & 100 & I-CFM & 100000 & 20.119 & 5.42498 & 3.50\\
125 & 96 & 100 & G-MPCF & 99999 & 20.119 & 5.42494 & 3.55\\
250 & 96 & 100 & I-CFM & 100000 & 20.119 & 5.42498 & 3.50\\
250 & 96 & 100 & G-MPCF & 99999 & 20.119 & 5.42494 & 3.80\\
500 & 96 & 100 & I-CFM & 100000 & 20.119 & 5.42498 & 3.80\\
500 & 96 & 100 & G-MPCF & 99999 & 20.119 & 5.42494 & 4.09\\
1000 & 96 & 100 & I-CFM & 100000 & 20.119 & 5.42498 & 3.86\\
1000 & 96 & 100 & G-MPCF & 99999 & 20.119 & 5.42494 & 4.09\\
1000 & 64 & 100 & I-CFM & 220923 & 8.952 & 5.42498 & 6.21\\
1000 & 64 & 100 & G-MPCF & 220921 & 8.952 & 5.42496 & 6.71\\
1000 & 128 & 100 & I-CFM & 56773 & 35.746 & 5.42490 & 2.76\\
1000 & 128 & 100 & G-MPCF & 56773 & 35.746 & 5.42491 & 2.76\\
1000 & 96 & 25 & I-CFM & 25000 & 20.119 & 1.35625 & 0.90\\
1000 & 96 & 25 & G-MPCF & 24999 & 20.119 & 1.35619 & 0.87\\
1000 & 96 & 50 & I-CFM & 50000 & 20.119 & 2.71249 & 1.69\\
1000 & 96 & 50 & G-MPCF & 49999 & 20.119 & 2.71244 & 1.77\\
\bottomrule\end{tabular}\end{adjustbox}
\end{table}\endgroup

\clearpage
\begingroup\fontsize{8.5}{10.2}\selectfont\setlength{\tabcolsep}{1.8pt}\renewcommand{\arraystretch}{1.02}
\begin{table}[H]\centering
\caption{All 122 candidate actions for the two sources. I, OT and G-M denote independent, optimal-transport and G-MPCF pairings; $\lambda$ is the retained displacement fraction and Min. R is minority recall. Validation KID is multiplied by $10^3$; other outcomes use test data. FID is mean $\pm$ sample standard deviation over three training runs.}
\begin{adjustbox}{max width=\linewidth}
\begin{tabular}{@{}lrrrrrr@{\hspace{4mm}}lrrrrrr@{}}\toprule

\multicolumn{7}{c}{Gaussian source} & \multicolumn{7}{c}{PPCA mixture source}\\
\cmidrule(lr){1-7}\cmidrule(lr){8-14}
Pairing & $\lambda$ & NFE & Val. KID & Test FID & MassTV & Min. R & Pairing & $\lambda$ & NFE & Val. KID & Test FID & MassTV & Min. R\\\midrule

Keep & 0.00 & 0 & 605.60 & $454.03\pm0.57$ & 0.9813 & 0.0000 & Keep & 0.00 & 0 & 277.44 & $265.65\pm0.58$ & 0.7929 & 0.1723\\
I & 0.25 & 1 & 574.61 & $440.40\pm0.56$ & 0.9813 & 0.0000 & I & 0.25 & 1 & 290.88 & $274.62\pm0.85$ & 0.7917 & 0.1046\\
I & 0.25 & 2 & 566.95 & $436.38\pm0.59$ & 0.9813 & 0.0000 & I & 0.25 & 2 & 285.62 & $270.72\pm0.86$ & 0.7817 & 0.1108\\
I & 0.25 & 4 & 552.32 & $428.67\pm0.75$ & 0.9813 & 0.0000 & I & 0.25 & 4 & 277.25 & $263.69\pm0.46$ & 0.7796 & 0.1489\\
I & 0.25 & 8 & 535.43 & $419.81\pm0.79$ & 0.9813 & 0.0000 & I & 0.25 & 8 & 271.44 & $258.00\pm0.52$ & 0.7874 & 0.2000\\
I & 0.25 & 16 & 524.47 & $414.06\pm0.81$ & 0.9813 & 0.0000 & I & 0.25 & 16 & 267.11 & $253.81\pm0.56$ & 0.7887 & 0.2357\\
\bottomrule\end{tabular}\end{adjustbox}
\end{table}\clearpage
\begin{table}[H]\ContinuedFloat\centering
\captionsetup{labelformat=continued}
\caption[]{G-MPCF actions for the Gaussian and PPCA mixture sources. Column definitions are as on the preceding page.}
\begin{adjustbox}{max width=\linewidth}\begin{tabular}{@{}lrrrrrr@{\hspace{4mm}}lrrrrrr@{}}\toprule

\multicolumn{7}{c}{Gaussian source} & \multicolumn{7}{c}{PPCA mixture source}\\
\cmidrule(lr){1-7}\cmidrule(lr){8-14}
Pairing & $\lambda$ & NFE & Val. KID & Test FID & MassTV & Min. R & Pairing & $\lambda$ & NFE & Val. KID & Test FID & MassTV & Min. R\\\midrule

G-M & 0.25 & 1 & 573.61 & $439.88\pm0.71$ & 0.9813 & 0.0000 & G-M & 0.25 & 1 & 284.19 & $270.65\pm0.82$ & 0.7606 & 0.1026\\
G-M & 0.25 & 2 & 565.13 & $435.44\pm0.68$ & 0.9813 & 0.0000 & G-M & 0.25 & 2 & 275.54 & $264.11\pm0.99$ & 0.7623 & 0.1440\\
G-M & 0.25 & 4 & 550.26 & $427.62\pm0.75$ & 0.9813 & 0.0000 & G-M & 0.25 & 4 & 266.10 & $256.02\pm0.74$ & 0.7511 & 0.2444\\
G-M & 0.25 & 8 & 534.72 & $419.47\pm0.81$ & 0.9813 & 0.0000 & G-M & 0.25 & 8 & 259.69 & $250.21\pm0.59$ & 0.7464 & 0.2153\\
G-M & 0.25 & 16 & 524.32 & $414.01\pm0.89$ & 0.9813 & 0.0000 & G-M & 0.25 & 16 & 254.91 & $245.90\pm0.54$ & 0.7396 & 0.2093\\
G-M & 0.50 & 1 & 502.15 & $408.29\pm0.56$ & 0.9813 & 0.0000 & G-M & 0.50 & 1 & 297.33 & $281.87\pm1.22$ & 0.6125 & 0.0357\\
\bottomrule\end{tabular}\end{adjustbox}
\end{table}\endgroup

\clearpage
\begingroup\small\setlength{\tabcolsep}{9pt}
\begin{table}[H]\centering
\caption{Source adequacy in fixed feature spaces, with $n=24$ independent training runs per dataset. Intervals are 95\% bootstrap intervals from resampling complete training runs.}
\begin{adjustbox}{max width=\linewidth}\begin{tabular}{lrrrr}\toprule

Dataset & Action & $n$ & SWD & 95\% interval\\\midrule

CIFAR PCA & Selected & 24 & 0.0542 & [0.0529, 0.0557]\\
CIFAR PCA & COT, 1 cycle & 24 & 0.0956 & [0.0921, 0.0990]\\
CIFAR PCA & Gaussian I & 24 & 0.0834 & [0.0809, 0.0859]\\
CIFAR PCA & Gaussian OT & 24 & 0.0869 & [0.0849, 0.0888]\\
CIFAR PCA & Mixture I & 24 & 0.0676 & [0.0646, 0.0706]\\
CIFAR PCA & Mixture mode & 24 & 0.0557 & [0.0542, 0.0573]\\
CIFAR PCA & Retain source & 24 & 0.0542 & [0.0530, 0.0556]\\
KMNIST & Selected & 24 & 0.1200 & [0.1172, 0.1230]\\
KMNIST & COT, 1 cycle & 24 & 0.1092 & [0.1064, 0.1119]\\
KMNIST & Gaussian I & 24 & 0.1196 & [0.1166, 0.1226]\\
KMNIST & Gaussian OT & 24 & 0.1162 & [0.1133, 0.1192]\\
KMNIST & Mixture I & 24 & 0.1158 & [0.1124, 0.1189]\\
KMNIST & Mixture mode & 24 & 0.1408 & [0.1372, 0.1445]\\
KMNIST & Retain source & 24 & 0.1218 & [0.1190, 0.1247]\\
UCI HAR & Selected & 24 & 0.1086 & [0.1055, 0.1118]\\
UCI HAR & COT, 1 cycle & 24 & 0.1046 & [0.1000, 0.1098]\\
UCI HAR & Gaussian I & 24 & 0.1360 & [0.1306, 0.1413]\\
UCI HAR & Gaussian OT & 24 & 0.1343 & [0.1279, 0.1407]\\
UCI HAR & Mixture I & 24 & 0.1388 & [0.1321, 0.1462]\\
UCI HAR & Mixture mode & 24 & 0.1119 & [0.1086, 0.1153]\\
UCI HAR & Retain source & 24 & 0.1086 & [0.1054, 0.1118]\\
\bottomrule\end{tabular}\end{adjustbox}
\end{table}\endgroup

\clearpage
\begingroup\fontsize{8.5}{10.2}\selectfont\setlength{\tabcolsep}{1pt}
\begin{table}[H]\centering
\caption{Individual source-adequacy runs. Full flow uses 16 evaluations; the displacement fraction $\lambda$ and selected budget are fixed on validation data. Run numbers identify matched repetitions.}
\begin{adjustbox}{max width=\linewidth}\begin{tabular}{@{}lrrrrr@{\hspace{3mm}}lrrrrr@{\hspace{3mm}}lrrrrr@{}}\toprule
\multicolumn{6}{c}{KMNIST} & \multicolumn{6}{c}{CIFAR-10 features} & \multicolumn{6}{c}{UCI HAR}\\
\cmidrule(lr){1-6}\cmidrule(lr){7-12}\cmidrule(lr){13-18}
Run & Source & Full & Selected & NFE & $\lambda$ & Run & Source & Full & Selected & NFE & $\lambda$ & Run & Source & Full & Selected & NFE & $\lambda$\\\midrule
1 & 0.1224 & 0.1340 & 0.1194 & 8 & 0.50 & 1 & 0.0568 & 0.0562 & 0.0568 & 0 & 0.00 & 1 & 0.0951 & 0.0955 & 0.0951 & 0 & 0.00\\
2 & 0.1197 & 0.1299 & 0.1150 & 4 & 0.25 & 2 & 0.0527 & 0.0587 & 0.0527 & 0 & 0.00 & 2 & 0.1170 & 0.1187 & 0.1143 & 4 & 1.00\\
3 & 0.1412 & 0.1573 & 0.1380 & 4 & 0.25 & 3 & 0.0510 & 0.0505 & 0.0500 & 16 & 0.75 & 3 & 0.1084 & 0.1110 & 0.1084 & 0 & 0.00\\
4 & 0.1239 & 0.1462 & 0.1255 & 8 & 0.50 & 4 & 0.0559 & 0.0574 & 0.0559 & 0 & 0.00 & 4 & 0.1057 & 0.1092 & 0.1057 & 0 & 0.00\\
5 & 0.1261 & 0.1477 & 0.1273 & 8 & 0.50 & 5 & 0.0543 & 0.0598 & 0.0543 & 0 & 0.00 & 5 & 0.1158 & 0.1191 & 0.1158 & 0 & 0.00\\
6 & 0.1153 & 0.1393 & 0.1161 & 8 & 0.50 & 6 & 0.0516 & 0.0522 & 0.0516 & 0 & 0.00 & 6 & 0.1044 & 0.1097 & 0.1044 & 0 & 0.00\\
\bottomrule\end{tabular}\end{adjustbox}
\end{table}\endgroup

\clearpage
\begingroup\small\setlength{\tabcolsep}{6pt}
\begin{table}[H]\centering
\caption{Coupling-mechanism comparison with five independent training runs per method. All runs use 3,500 updates and batch size 512. Values are means $\pm$ sample standard deviations. SWD and mode TV use 16 field evaluations.}
\begin{adjustbox}{max width=\linewidth}\begin{tabular}{lrrrrr}\toprule

Coupling & Params & Score error & Velocity error & SWD & Mode TV\\\midrule

Independent joint & 9348 & $0.087\pm0.003$ & $0.099\pm0.003$ & $0.183\pm0.016$ & $0.055\pm0.008$\\
Exact sector & 9218 & $0.020\pm0.001$ & $0.163\pm0.009$ & $0.156\pm0.014$ & $0.043\pm0.015$\\
Laguerre & 9218 & $0.008\pm0.00045$ & $0.062\pm0.005$ & $0.140\pm0.023$ & $0.030\pm0.008$\\
Exact polar & 9218 & $0.008\pm0.00033$ & $0.066\pm0.002$ & $0.135\pm0.023$ & $0.033\pm0.007$\\
\bottomrule\end{tabular}\end{adjustbox}
\end{table}\endgroup

\clearpage
\begingroup\fontsize{8.5}{10.2}\selectfont\setlength{\tabcolsep}{1.5pt}\renewcommand{\arraystretch}{1.02}
\begin{table}[H]\centering
\caption{Six-method benchmark, five training runs per target and method. Requested and effective NFE are shown separately. SWD and mode TV are means $\pm$ sample standard deviations; solver values are mean Euler-to-reference MSE. A dash marks an undefined solver outcome.}
\begin{adjustbox}{max width=\linewidth}\begin{tabular}{@{}lrrrrr@{\hspace{4mm}}lrrrrr@{}}\toprule

\multicolumn{6}{c}{rare12} & \multicolumn{6}{c}{gaussian8}\\
\cmidrule(lr){1-6}\cmidrule(lr){7-12}
Method & Req. & Eff. & SWD & Mode TV & Solver & Method & Req. & Eff. & SWD & Mode TV & Solver\\\midrule

I-CFM & 1 & 1 & $0.909\pm0.011$ & $0.217\pm0.093$ & 0.9859 & I-CFM & 1 & 1 & $2.016\pm0.016$ & $0.185\pm0.054$ & 4.4530\\
I-CFM & 2 & 2 & $0.493\pm0.009$ & $0.094\pm0.010$ & 0.2328 & I-CFM & 2 & 2 & $0.661\pm0.035$ & $0.045\pm0.017$ & 0.6872\\
I-CFM & 4 & 4 & $0.276\pm0.007$ & $0.069\pm0.004$ & 0.0527 & I-CFM & 4 & 4 & $0.241\pm0.015$ & $0.026\pm0.007$ & 0.0765\\
I-CFM & 8 & 8 & $0.228\pm0.009$ & $0.057\pm0.004$ & 0.0158 & I-CFM & 8 & 8 & $0.215\pm0.047$ & $0.025\pm0.008$ & 0.0210\\
I-CFM & 16 & 16 & $0.203\pm0.011$ & $0.048\pm0.006$ & 0.0044 & I-CFM & 16 & 16 & $0.215\pm0.048$ & $0.026\pm0.008$ & 0.0058\\
I-CFM & 32 & 32 & $0.192\pm0.013$ & $0.044\pm0.007$ & 0.0011 & I-CFM & 32 & 32 & $0.221\pm0.045$ & $0.025\pm0.007$ & 0.0016\\
\bottomrule\end{tabular}\end{adjustbox}
\end{table}\endgroup

\clearpage
\begingroup\small\setlength{\tabcolsep}{7pt}
\begin{table}[H]\centering
\caption{Cross-validation across 36 conditions. Responses average three paired runs; standardization and fitting use the other 35 conditions. $R^2$ follows equation~\ref{eq:cv-score}.}\label{tab:cv}
\begin{adjustbox}{max width=\linewidth}\begin{tabular}{lrrr}\toprule
Predictors & $R^2$ & RMSE & Mean absolute error\\\midrule
Training mean & -0.058 & 0.00608 & 0.00409\\
Field only & 0.520 & 0.00410 & 0.00264\\
Three diagnostics & 0.655 & 0.00347 & 0.00220\\
\bottomrule\end{tabular}\end{adjustbox}
\end{table}\endgroup

\endgroup
\else
\ifdefined\SAcombined\fi\fi

\end{document}